\documentclass{article}

 \usepackage[letterpaper,margin=1in]{geometry}
\usepackage{natbib}
\date{}
\AtBeginDocument{\renewcommand{\And}{\and}}
\newcommand{\AND}{\end{tabular}\par\vspace{1em}\begin{tabular}[t]{c}}
\newenvironment{ack}{\section*{Acknowledgments and Disclosure of Funding}}{}

\usepackage[utf8]{inputenc} \usepackage[T1]{fontenc}    \usepackage{hyperref}       \usepackage{url}            \usepackage{booktabs}       \usepackage{amsfonts}       \usepackage{nicefrac}       \usepackage{microtype}      \usepackage{xcolor}         

\usepackage{graphicx} 
\usepackage[utf8]{inputenc}
\usepackage{pgfplots}
    \pgfplotsset{compat=1.18}
\usepackage{tikz}
    \usetikzlibrary{shapes,decorations}
    \usetikzlibrary{positioning}
    \usetikzlibrary{external} \usepackage{amsfonts, mathrsfs, mathtools}
\usepackage{nicefrac}
\usepackage{MnSymbol}
\usepackage[ruled, noline,noend]{algorithm2e} \usepackage{enumitem}
\usepackage{hyperref}
\usepackage{placeins}
\usepackage{amsthm}
\usepackage[capitalize,noabbrev]{cleveref}

\newcommand{\cE}{\mathcal{E}}
\newcommand{\cF}{\mathcal{F}}
\newcommand{\cG}{\mathcal{G}}
\newcommand{\cH}{\mathcal{H}}

\newcommand{\cO}{\mathcal{O}}

\newcommand{\cR}{\mathcal{R}}
\newcommand{\cS}{\mathcal{S}}

\newcommand{\cX}{\mathcal{X}}

\newcommand{\E}{\mathbb{E}}

\newcommand{\I}{\mathbb{I}}

\newcommand{\N}{\mathbb{N}}

\newcommand{\Pb}{\mathbb{P}}
\newcommand{\bbR}{\mathbb{R}}
\newcommand{\R}{\mathbb{R}}

 \newcommand{\e}{\varepsilon}

\newcommand{\lrb}[1]{\left(#1\right)}
\newcommand{\brb}[1]{\bigl(#1\bigr)}
\newcommand{\Brb}[1]{\Bigl(#1\Bigr)}
\newcommand{\bbrb}[1]{\biggl(#1\biggr)}
\newcommand{\lsb}[1]{\left[#1\right]}
\newcommand{\bsb}[1]{\bigl[#1\bigr]}
\newcommand{\Bsb}[1]{\Bigl[#1\Bigr]}

\newcommand{\Bbsb}[1]{\Biggl[#1\Biggr]}
\newcommand{\lcb}[1]{\left\{#1\right\}}
\newcommand{\bcb}[1]{\bigl\{#1\bigr\}}

\newcommand{\lce}[1]{\left\lceil#1\right\rceil}

\newcommand{\lfl}[1]{\left\lfloor#1\right\rfloor}

\newcommand{\labs}[1]{\left\lvert#1\right\rvert}

\newcommand{\lno}[1]{\left\lVert#1\right\rVert}

\newcommand{\lan}[1]{\left\langle#1\right\rangle}
\newcommand{\ban}[1]{\bigl\langle#1\bigr\rangle}
\newcommand{\Ban}[1]{\Bigl\langle#1\Bigr\rangle}

\DeclareMathOperator*{\argmax}{argmax}
\newcommand{\dif}{\,\mathrm{d}}
\newcommand{\lip}{Lipschitz}

\newcommand{\UCB}{\mathrm{UCB}}

\newcommand{\KL}{\mathrm{KL}}

\newcommand{\iop}{\infty}
\newcommand{\Var}{\mathrm{Var}}
\newcommand{\ceq}{\coloneqq}

\DeclareSymbolFont{extraup}{U}{zavm}{m}{n}
\DeclareMathSymbol{\clubsuit}{\mathalpha}{extraup}{84}
\DeclareMathSymbol{\spadesuit}{\mathalpha}{extraup}{81}
\DeclareMathSymbol{\varheartsuit}{\mathalpha}{extraup}{86}
\DeclareMathSymbol{\vardiamondsuit}{\mathalpha}{extraup}{87}

\renewcommand{\tilde}{\widetilde}

\newtheorem{theorem}{Theorem}
\newtheorem{lemma}{Lemma}
\newtheorem{proposition}{Proposition}
\newtheorem{corollary}{Corollary}

\newtheorem{remark}{Remark}

\newtheorem{assumption}{Assumption}

\crefalias{learningprotocol}{algocf}
\crefname{learningprotocol}{learning protocol}{learning protocols}
\Crefname{learningprotocol}{Learning protocol}{Learning protocols}

\newenvironment{proofsketch}{\paragraph{Proof sketch}}{\hfill$\blacksquare$}

\title{Online Budget Allocation\\ with Censored Semi-Bandit Feedback}

\author{François Bachoc\\
  Laboratoire Paul Painlevé \\
  Université de Lille\\
  Lille, France \\
  Institut universitaire de France (IUF) \\
  \texttt{francois.bachoc@univ-lille.fr} \\
  \And
  Nicolò Cesa-Bianchi \\
  Department of Computer Science \\
  Università degli Studi di Milano \\
  Milano, Italy \\
  \texttt{nicolo.cesa-bianchi@unimi.it} \\
  \AND
  Tommaso Cesari \\
  EECS\\
  University of Ottawa \\
  Ottawa, Canada \\
  \texttt{tcesari@uottawa.ca} \\
  \And
  Roberto Colomboni\\
  School of Mathematics \\
 University of Bristol \\
 Bristol, United Kingdom \\
\texttt{roberto.colomboni@bristol.ac.uk}
}

\begin{document}

\maketitle

\begin{abstract}
We study a stochastic budget-allocation problem over $K$ tasks.
At each round $t$, the learner chooses an allocation $X_t \in \Delta_K$.
Task $k$ succeeds with probability $F_k(X_{t,k})$, where $F_1, \dots, F_K$ are nondecreasing \emph{budget-to-success} curves, and upon success yields a random reward with unknown mean $\mu_k$.
The learner observes which tasks succeed, and observes a task's reward \emph{only} upon success (\emph{censored} semi-bandit feedback).
This model captures, for instance, splitting payments across crowdsourcing workers or distributing bids across simultaneous auctions, and subsumes stochastic multi-armed bandits and semi-bandits.
We design an optimism-based algorithm that operates under censored semi-bandit feedback and assumes access to an exact maximization oracle.
Our main result is an instance-dependent bound showing that in \emph{diminishing-returns} regimes, the regret of this algorithm scales polylogarithmically with the horizon $T$ without any \emph{ad hoc} tuning.
For general nondecreasing curves, we prove that the same algorithm (with the same tuning) achieves a worst-case regret upper bound of $\tilde O(K\sqrt{T})$.
Finally, we establish a matching worst-case regret lower bound of $\Omega(K\sqrt{T})$ that holds even for \emph{full-feedback} algorithms, highlighting the intrinsic hardness of our problem outside diminishing returns.\end{abstract}

\section{Introduction}

The multi-armed bandit, a mathematical model for sequential decision-making with partial feedback, has been applied to a wide range of domains, including clinical trials, product recommendation, online advertising, digital markets, and process optimization. A multi-armed bandit consists of an action space and an unknown reward process associated with each action. At every decision round, a bandit algorithm selects an action and observes the reward produced at that round by the process associated with the chosen action \citep{lattimore2020bandit}. 
In this work we focus on a setting where rewards are stochastic and actions have a linear structure. 
The expected reward is $\ban{\mu,F(x)}$, where $x$ is the selected action, $\mu\in [0,1]^K$ are unknown coefficients, and $F\colon\Delta_K\to [0,1]^K$ is some (possibly nonlinear) function defined on the probability simplex $\Delta_K$, with $F(x)$ to be interpreted as the means of $K$ Bernoulli random variables.
The realizations of these random variables, which the algorithm observes after each round, reveal which of the 
$K$ components of the chosen action actually contributed to the total reward; for each such contributing component, the algorithm observes the corresponding reward.
The total reward is then the sum of the rewards associated with the completed tasks.
Hence, compared to the standard bandit feedback in linear bandits \citep{dani2008stochastic}, the feedback we receive---what we call \emph{censored semi-bandit}---is more informative.
This natural censored-feedback model goes beyond the settings treated in prior analyses and leads to stronger learning guarantees.
To illustrate the scope of the model, we highlight two representative scenarios that fit our framework (with action $x \in \Delta_K$, completion probabilities $F(x)$, and expected reward $\ban{\mu,F(x)}$):
\begin{itemize}[wide]
    \item In crowdsourcing, a platform allocates in each round a fixed budget across $K$ workers.
    Paying worker $k$ a share $x_k$ induces task completion with probability $F_k(x_k)$, naturally increasing with $x_k$.
    Thus $F(x) = \brb{ F_1(x_1),\ldots,F_K(x_K) }$ models how payments translate into completion.
    The platform's expected reward in the round is $\ban{ \mu, F(x) }$, where $\mu_k\in[0,1]$ captures the (unknown) expected value of the work produced upon completion by the $k$-th worker.

    \item An autobidding agent spends in each round a fixed budget to bid in $K$ simultaneous auctions. Auctions are won with probabilities $F(x) = \brb{F_1(x_1),\ldots,F_K(x_K)}$ depending on the agent's bids $x = (x_1,\ldots,x_K)$ and the competing bids in each auction. In each round, the agent's expected reward $\ban{\mu,F(x)}$ depends on the unknown value of each auctioned item, which is modeled via $K$ random variables with means $\mu\in [0,1]^K$.
\end{itemize}
Other applications include online model selection, where a computing time budget is split among several instances of a learning system (e.g., a LLM trained with different hyperparameter values) and the expected reward is computed based on the quality of the output produced by those instances that terminated training (or inference) within the allocated time---see, e.g., \cite{xia2024llm} for work along these lines.
Note that our theoretical model is representative of simplified versions of these applicative examples.
In crowdsourcing, our model means that payment
mainly affects acceptance or completion probability while
the output
quality is invariant \citep{tran2014efficient,kang2020task}.
In autobidding, our model does not capture global budget constraints affecting multiple time steps \citep{aggarwal2025multi,avadhanula2021stochastic,deng2023multi}.
Similarly, for compute-allocation, our model does not capture multi-round jobs,
queuing, and delayed observations \citep{xia2024llm,li2025llm,nguyen2024metallm}.

\subsection{Setting}
\label{s:setting}

For any $n\in\N$, we let $[n]$ be the set $\{1,\dots,n\}$.
The action space is the $(K-1)$-dimensional simplex $\Delta_K \coloneqq \{ x \in [0,1]^K : \sum_{k \in [K]} x_k = 1 \}$, for some natural number $K \ge 2$, which we fix for the remainder of this work.
Let $G_{t,k}$ be the stochastic reward associated with the completion of task $k \in [K]$ at time $t$.
For each $k \in [K]$ and $t \in \N$, we model $G_{t,k}$ as an unknown $[0,1]$-valued random variable with unknown expectation $\mu_k$, not depending on $t$. We further assume the vectors $(G_{t,1},\dots,G_{t,K})_{t \in \N}$ to be independent. For each task $k\in[K]$, the distribution of $G_{t,k}$ does not depend on $t$.
If we allocate a budget $x_k \in [0,1]$ on a task $k \in [K]$, then the probability of completing the task $k$ is $F_k(x_k)$, where $F_1,\dots,F_K \colon [0,1] \to [0,1]$ are known\footnote{In Appendix~\ref{s:necessity},
we investigate the case where the budget-to-success curves are unknown.
In this case, we obtain a \emph{strong} impossibility result: even for $K=2$, when the budget-to-success curves $F_1$ and $F_2$ are indicator functions, and the random variables $G_t$ are deterministically equal to $(1,1)$, any algorithm must incur worst-case linear regret. See also \Cref{rem:unknown:F}.} non-decreasing functions, that we call \emph{budget-to-success} curves.
Also, we let $B_{t,k}(x_k)$ be the Bernoulli random variable that indicates the completion of the task $k$ at time $t$, so that $\E\bsb{ B_{t,k}(x_k) } = F_k(x_k)$.
We model the family of random functions $(B_{t,k})_{t\in \N, k \in [K]}$ as an independent family, 
independent of $(G_{t,k})_{t \in \N,k \in [K]}$.
\footnote{Technically, we need the function $(x,\omega) \mapsto B_{t,k}(x)(\omega)$---as a joint function of the index $x$ and the realization $\omega$---to be jointly measurable as well.
This is necessary to avoid measure theoretic pathologies.
On the other hand, most common constructions naturally satisfy this assumption. For example, $B_{t,k}(x) \coloneqq \I\{V_{t,k} \le F_k(x)\}$, where $(V_{t,k})_{t \in \N, k\in [K]}$ is an i.i.d.\ family of uniform random variables on $[0,1]$.
For the sake of simplicity, we shall not discuss measurability issues further in the rest of the paper.}
For notational convenience, we define: 
$\mu \coloneqq (\mu_1,\dots,\mu_K)$, 
$F \colon \Delta_K \to [0,1]^K, (x_1,\dots,x_K) \mapsto \brb{F_1(x_1), \dots, F_K(x_K)}$, and $G_t \coloneqq \lrb{G_{t,1},\dots,G_{t,K}}$. 
Moreover, for each $t \in \N$, we let $B_t$ be the function that maps each budget allocation $x \ceq (x_1,\dots,x_K) \in \Delta_K$ to the vector of random variables $B_t(x) \coloneqq \brb{ B_{t,1}(x_1),\dots,B_{t,K}(x_K) }$.
Finally, if $a,b \in \R^K$, we denote by $\langle a,b \rangle$ their scalar product $\sum_{k \in [K]} a_k b_k$.

Given these definitions, the online learning interaction is summarized in \Cref{a:learning-protocol}.

\begin{algorithm}[t]
\renewcommand*{\algorithmcfname}{Learning protocol}
\For{time step $t=1,2,\dots$}{
  The learner chooses $X_t \in \Delta_K$\;
  The learner gains $\ban{G_t,\, B_t(X_t)}$\;
  The learner observes some feedback $\mathscr{F}_t$\;
}
\caption{Learning protocol}
\label{a:learning-protocol}
\end{algorithm}

We focus on the following two feedback models.\footnote{Except in \Cref{t:LB:linear:bandit}, where we introduce the third model of \emph{bandit} feedback to highlight the gap with the standard linear-bandit literature.
Outside of \Cref{t:LB:linear:bandit}, we establish all \emph{upper bounds} under the realistic and challenging censored semi-bandit model, and we prove the corresponding \emph{lower bounds} under the simpler full-feedback model.
In particular, our minimax lower bounds under full feedback apply \emph{a fortiori} to censored feedback, showing that our worst-case guarantees under censored semi-bandit feedback match what is achievable even by powerful algorithms with full-feedback access.
}
\begin{itemize}[wide]
    \item \textbf{Censored semi-bandit:}  $\mathscr{F}_t$ is $B_t(X_t)$ and $\brb{B_{t,1}(X_{t,1})\,G_{t,1}, \, \dots, \, B_{t,K}(X_{t,K})\,G_{t,K}}$.
    In words, the learner observes which tasks have been completed and the rewards of \emph{only} these completed tasks.
    \item \textbf{Full feedback:} $\mathscr{F}_t$ is $B_t(X_t)$ and $G_t$. In this case, the learner observes the rewards of \emph{all} tasks.
\end{itemize}
The performance of the learner at any time horizon $T \in \N$ is evaluated with the \emph{regret}, defined as
\begin{equation}
\label{e:regret}
    R_T
\coloneqq
   T \left( \sup_{x \in \Delta_K} \ban{ \mu, F(x) } \right) - \sum_{t=1}^T \E\Bsb{ \ban{ \mu, F(X_t) } }    
\end{equation}
i.e., the cumulative difference between the expected reward obtained by the best allocation of the budgets and the expected  performance of the learning algorithm generating the sequence $X_1, X_2, \dots$.
If $\sup_{x \in \Delta_K} \ban{ \mu, F(x) }$ is attained, we will denote one of its maximizers by  $x^\star(\mu)$ to stress the dependence on $\mu$, or simply by $x^\star$ if this dependence is clear from context.
To give guarantees that are uniform over all possible problem instances, we will also consider the worst-case regret $R^\star_T \ceq \sup R_T$, where the supremum is taken over all choices of $F_1,\dots,F_K$, and joint distributions of $(G_{t,1},\dots, G_{t,K})$.

We conclude this section by noting that our model subsumes the stochastic combinatorial semi-bandit model \citep{chen2013combinatorial} with $M$-sets and $M\in[K]$ (and, in particular, $K$-armed bandits when $M=1$).
Indeed, let $\I_A$ denote the indicator of a set $A$, and let $\boldsymbol e_1,\dots,\boldsymbol e_K$ be the standard basis of $\R^K$.
If $F_1 \coloneqq \cdots \coloneqq F_K \coloneqq \I_{[1/M,1]}$, then task $k$ is completed if and only if it receives at least a $1/M$ budget share.
In this case, the Pareto-optimal allocations are exactly those of the form $x=\frac{1}{M}\sum_{m=1}^M \boldsymbol e_{i_m}$ with $i_1,\ldots,i_M$ distinct in $[K]$, so in each round the learner selects $M$ arms and earns and observes their individual rewards, recovering the standard semi-bandit feedback model with $M$ plays per round.

\subsection{Our contributions}

We introduce a flexible online learning setting that generalizes the classic semi-bandit framework and models, in particular, crowdsourcing, multi-platform autobidding, and the allocation of compute resources for training large language models (LLMs). 
We make the following contributions.
\begin{itemize}[wide]
\item We design an intuitive optimism-based algorithm (\Cref{algo:ucbowski}) that operates using only censored semi-bandit feedback observations and assumes access to an exact maximization oracle, as discussed in \Cref{rem:approx-oracle}.
\item For a practically relevant and technically challenging \emph{diminishing-returns} regime, we prove an instance-dependent bound showing that \Cref{algo:ucbowski} achieves a fast polylogarithmic regret rate \emph{without any ad hoc tuning} (\Cref{t:upper-bound-speed-up}).

\item We prove that this fast rate \emph{cannot} be obtained from \emph{bandit} feedback alone: any (linear) bandit algorithm incurs $\Omega(\sqrt{T})$ regret even in the diminishing-returns regime (\Cref{t:LB:linear:bandit}).
\item For general nondecreasing curves, we prove that \Cref{algo:ucbowski} achieves a worst-case regret of $\tilde{\cO}\brb{K\sqrt{T}}$ (again, without any \emph{ad hoc} tuning; \Cref{t:upper-bound-worst-case}).
\item For general nondecreasing curves, we prove that no algorithm can guarantee worst-case regret better than $\Omega\brb{K\sqrt{T}}$, not even with full-feedback access (\Cref{t:lower-bound-worst-case}).
\end{itemize}
These results not only establish the worst-case optimality of \Cref{algo:ucbowski}, but also highlight its ability to adaptively exploit structure, illustrating its broad applicability in real-world domains.

\subsection{Techniques and challenges}

\paragraph{Why off-the-shelf bandit techniques fail.}
Under censored semi-bandit feedback, the learner can compute the realized \emph{total} reward in each round, and therefore has access to standard bandit feedback.
Thus, in principle, one could attempt to apply bandit algorithms directly over the (infinite) action set $\Delta_K$.

However, classical bandit guarantees scale with the number of arms $H$ (typically as $\sqrt{HT}$), and in our problem $H$ is not finite.
A standard workaround in $\cX$-armed bandits is to discretize the action space, provided the expected reward function is sufficiently regular (e.g., Lipschitz).
In our setting, this regularity is \emph{not} guaranteed: the budget-to-success curves $F_k$ need not be Lipschitz, and consequently the induced expected reward may exhibit arbitrarily steep variations.

Moreover, even under an $L$-Lipschitz assumption, discretization suffers from the curse of dimensionality.
Since the $(1/n)$-covering number of the unit simplex $\Delta_K$ scales as $\Theta(n^{K-1})$, an $(1/n)$-net requires $\Theta(n^{K-1})$ grid points to ensure that every point (in particular, any optimum) lies within distance $1/n$ of the grid.
This yields regret guarantees of the form
$
    O\lrb{\frac{L}{n}T + \sqrt{n^{K-1}T}},
$
which, after optimizing over $n$, give the nonparametric rate
$
    O\brb{L^{(K-1)/(K+1)}\,T^{K/(K+1)}}.
$

One might hope to do better by exploiting bandit feedback together with the \emph{linear} structure of the realized reward, which can be written as $\langle G_t, B_t(X_t)\rangle$.
Yet, as discussed in \Cref{s:related_works} and further shown in \Cref{t:LB:linear:bandit}, even linear-bandit techniques are suboptimal in our setting.

\textbf{Diminishing-returns regime.}
Proving (poly-)logarithmic regret when the budget-to-success curves exhibit strong diminishing returns raises three main challenges.

\emph{(i) From estimation error to second-order regret.}
We first develop a self-bounding argument showing that, when the optimal allocation lies in the interior of the simplex, optimizing the objective computed from estimated task means yields a second-order \emph{instantaneous} regret that scales quadratically with the estimation error in the mean-reward vector.

\emph{(ii) Interior optimality without assuming it.}
A subtler issue is that interior optimality cannot be taken for granted.
We show that it is in fact an implication of our regime assumptions: when every task has strictly positive mean reward, the optimal allocation is forced away from the boundary, and no component is starved.

\emph{(iii) Quantifying the feedback rate under optimism.}
Finally, we must control how quickly the mean estimates concentrate under the optimistic allocation rule.
Using concentration inequalities, we prove that once the cumulative mass allocated to a task is large enough (logarithmic in $T$), the number of observed completions is of the same order with high probability.
Consequently, each mean can be estimated at the $1/\sqrt{t}$ rate (the same as in full-feedback regimes, despite only having access to censored semi-bandit feedback), which combined with the second-order regret property yields the desired (poly-)logarithmic guarantee.

\emph{(iv) Not another standard bandit $\log$ bound.}
Our fast-rate result should not be viewed as a routine analog of classical logarithmic bandit bounds.
In finite-armed bandits, logarithmic regret is tied to a positive gap between the best and second-best arm.
Here, by contrast, the action space is continuous, so the usual gap-based intuition breaks down: the standard bandit gap can be zero, and yet we still obtain a polylogarithmic rate.
The mechanism is fundamentally different, and
the fast rate does not come from a classical gap condition, but from a much subtler interplay between curvature, interior optimality induced by \(\mu_k>0\), and the richer observability of censored semi-bandit feedback.
This point is further developed in \Cref{t:LB:linear:bandit}: bandit feedback alone does \emph{not} yield such a fast rate, even in very benign instances where curvature and interior optimality are both present.
So the result is not simply ``logarithmic regret with a complicated constant,'' but rather a qualitatively new fast-rate phenomenon that is specific to this feedback/geometry combination.

\textbf{Worst-case upper bound.}
We obtain optimal worst-case regret by exploiting two key features of our model.
First, the per-round reward decomposes as a sum of contributions from $K$ tasks.
Second, allocating more budget to a task increases its chance of completion and hence the probability of observing its reward; in expectation, the rate at which we collect observations from a task is coupled with the reward we can obtain from it.
Together, these properties naturally suggest an optimistic strategy: as soon as a component looks promising, we funnel more budget into it.
This increases the flow of feedback from that component, which either confirms the choice or quickly reveals it to be suboptimal, allowing the learner to reallocate before incurring large regret.

We formalize this intuition via an argument reminiscent of the elliptical potential lemma \cite[Lemma~19.4]{lattimore2020bandit}.
In particular, under our optimistic allocation rule, the regret attributable to each task can be controlled by the square root of the number of times its reward is observed, which yields an overall worst-case regret of order $K\sqrt{T}$ (\Cref{t:upper-bound-worst-case}).

\textbf{Worst-case lower bound.}
To show that $K\sqrt{T}$ is the best possible rate in general---\emph{even with full-feedback access}---we construct a hard instance with $2K$ tasks, arranged into $K$ disjoint pairs.
Within each pair, one task is slightly better than the other.
The budget-to-success curves $F_1,\ldots,F_{2K}$ are designed so that, in each round, any allocation effectively selects a subset of tasks that are \emph{active} (success probability~$1$) while the remaining tasks are \emph{inactive} (success probability~$0$).
This structure allows us to embed $K$ independent two-armed bandit problems, one per pair.

A key technical difficulty is making this embedding rigorous for \emph{arbitrary} algorithms over the simplex.
Indeed, the learner is not restricted to activating exactly one task per pair: it may allocate budget so as to activate none in some pairs and both in others.
We prove that, nevertheless, any strategy can be converted into (or compared against) a strategy that induces $K$ independent two-armed decisions without increasing regret.
We are not aware of an analogous reduction in the existing literature.

Once this reduction is established, the remaining argument is somewhat closer to the existing literature than for the reduction mentioned above, and follows a classical bandit lower-bound blueprint:
since each of the $K$ embedded two-armed problems incurs $\Omega(\sqrt{T})$ regret in the worst case, the overall regret is at least $\Omega(K\sqrt{T})$.
However, making this step precise requires a careful information-theoretic analysis tailored to our feedback model, relying on tools from statistical hypothesis testing and information theory.

We conclude this section by remarking that our lower bound is consistent with a standard observation from combinatorial semi-bandits: in the $M$-set semi-bandit, when $M=\lfl{K/2}$ (so that half of the coordinates is observed each round), the minimax worst-case rate is of the same order under semi-bandit and full-feedback.
However, this connection is not entirely plug-and-play in our model: semi-bandits arise as a special case only after restricting attention to Pareto-optimal allocations (which, under the threshold curves, correspond exactly to uniform $1/M$ allocations over an $M$-subset), and our proof makes this reduction explicit.

\subsection{Related work}

For a more detailed discussion of related work, see Appendix~\ref{s:related_works}.
The closest related work is the concurrent and independent work of \citet{montanari2026learning}. 
They study a resource-allocation problem with censored feedback and Bernoulli rewards, and analyze learning in a specific smooth parametric threshold model with unknown parameters, which is outside the scope of our impossibility result in Appendix~\ref{s:necessity}. 
On the shared subclass, our $\sqrt T$ guarantee is uniform over instances satisfying our oracle assumption (\Cref{rem:approx-oracle}), whereas theirs is instance-dependent; Appendix~\ref{s:related_works} details the dependence on $K$.
For fast rates, we derive interiority from the behavior of the curves at the origin, whereas they assume it explicitly; the two sets of assumptions are not directly comparable. 

Our setting is also related to several other resource-allocation and bandit
models.
Resource-allocation bandits in
\citep{lattimore2014optimal,lattimore2015linear} learn unknown capped-linear
completion maps, whereas here the budget-to-success curves are known and the
unknown reward means are learned from censored observations.
Combinatorial semi-bandits with probabilistically triggered arms
\citep{chen2016combinatorial} provide a finite-action analogue, while our
action space is continuous and triggering probabilities may become arbitrarily
small near the boundary (for small budget).
The censoring mechanism in \citet{verma2019censored} operates in the opposite
direction to ours, and \citet{fontaine2020adaptive} consider concave resource
allocation with noisy first-order rather than censored reward feedback.
Finally, there is related work in online learning with auctions \citep{hummel2016machine,iyer2014mean} but with different objectives, constraints, and feedback models.

\section{Algorithm and polylogarithmic regret rate under diminishing returns}
\label{s:speed-ups}

First, we introduce \Cref{algo:ucbowski},
a UCB-style algorithm designed to operate under censored semi-bandit feedback. 
At each round, the learner selects an allocation that optimistically maximizes estimated rewards based on past feedback. 
Since rewards are observed only for completed tasks, the algorithm maintains task-wise upper confidence bounds (UCBs).
The confidence width evolves over time as a function of the cumulative completion count and required confidence level. 

To lighten the notation, from \Cref{algo:ucbowski} onward, we adopt the convention that $0/0\ceq0$.

\begin{algorithm}[t]
\DontPrintSemicolon
  \textbf{Input:}
  Time horizon $T\in\N$,
  confidence parameter $\delta\in (0,1)$\;
  \textbf{Initialization:} Let $\UCB_{0}^\delta \ceq (\UCB_{0,1}^\delta, \dots, \UCB_{0,K}^\delta)$ where, for all $k\in[K]$, $\UCB_{0,k}^\delta \ceq \sqrt{\log \frac{2}{\delta}} \!\!\!$\;
    \For{time step $t=1,\dots,T$}
    {
        Play $X_t \in \Delta_K$ such that $
        X_t \in 
        \argmax_{x \in \Delta_K} \lan{\UCB_{t-1}^\delta , F(x)} 
        $
        \label{state:ucb-pick}\;
        Update $\UCB_{t}^\delta \ceq (\UCB_{t,1}^\delta,\dots, \UCB_{t,K}^\delta)$, where, for all $k\in[K]$,
        \[
            \UCB_{t,k}^\delta
        \coloneqq
            \frac{ \sum_{s=1}^{t} B_{s,k}(X_{s,k})G_{s,k}}{\sum_{s=1}^{t} B_{s,k}(X_{s,k})}
            +
            \sqrt{\frac{\log\frac{2}{\delta}}{1+\sum_{s=1}^{t} B_{s,k}(X_{s,k})}}
        \]
        with the understanding that $\UCB_{t,k}^\delta \ceq \UCB_{0,k}^\delta$ if $B_{1,k}(X_{1,k}) =  \dots =  B_{t,k}(X_{t,k}) = 0$
    }
\caption{\label{algo:ucbowski}
}
\end{algorithm}

We now investigate the realistic setting where the functions $F_1, \dots, F_K$ exhibit diminishing returns, i.e., are concave. We show that concavity, coupled with the assumption that the components of $\mu$ are non-zero, results in instance-dependent
polylogarithmic rates.
We flesh out the general case in Appendix~\ref{s:appe-speed-ups}, but for the sake of readability, in this section, we will restrict our attention to the special case of \emph{concave power functions}, i.e., of functions $f \colon [0,1]\to[0,1]$ for which there exists $a\in(0,1)$ such that, for all $x\in[0,1]$, $f(x)\ceq x^a$.
This case suffices to present some of the main ideas of the general analysis in Appendix~\ref{s:appe-speed-ups} while keeping the presentation relatively light.
Furthermore, this is a special case of particular interest as it models many important behaviors like \emph{scaling laws} of LLMs.
\begin{assumption}
    \label{ass:speed-up} For any $k\in[K]$, $\mu_k$ is non-zero and $F_k$ is a concave power function $x \mapsto x^{a_k}$ for $a_k \in (0,1)$.
\end{assumption}
\begin{theorem}
    \label{t:upper-bound-speed-up}
    Under \Cref{ass:speed-up}, for any $T \ge 2$, if we run \Cref{algo:ucbowski} with parameters $T$ and $\delta \coloneqq \frac{1}{(KT)^2}$, letting $\mu_{\inf} \ceq \min_{k\in[K]} \mu_k$, $a_{\min} \ceq \min_{k\in[K]} a_k$, and $a_{\max} \ceq \max_{k\in[K]} a_k$, we have\footnote{We remark that we did not optimize the dependence on
    $K,\mu_{\inf},a_{\min}$, and $a_{\max}$.
    The dependence displayed above is deliberately loose in order to obtain a simpler, more readable upper bound.}
\[
    R_T
\le 
    800
    \lrb{
    \frac{
    K
    }{
    a_{\min}\mu_{\inf}(1-a_{\max})}
    }^{\frac{9}{1-a_{\max}}}
    (\log T)^2.
\]
\end{theorem}
\begin{proofsketch}
Fix any $T\ge2$ and $a_1,\dots,a_K\in(0,1)$ such that, for all $k\in[K]$ and $x_k\in(0,1)$, $F_k(x_k) = x_k^{a_k}$, and recall $a_{\max} \ceq \max_{k\in[K]} a_k$.
We prove this result in three steps.

\emph{Step 1: \eqref{e:lemma:LB:bernoullis}.} The first key observation is that, with high probability, at any time step, the number of times any given task has been completed is at least proportional to a deterministic function of the budget invested in the task so far, assuming that a sufficient amount of budget has been invested in it.
More precisely, there exist two numerical constants $c_1,c_2>0$ (independent of $T$ and $K$) such that, with probability at least $1-\nicefrac{K}{T}$, for all $t\in[T]$ and $k\in[K]$, the following event $\cG_T$ holds:
\begin{equation}
\label{e:lemma:LB:bernoullis}
    \sum_{s=1}^t X_{s,k}^{a_k}
\ge 
    c_1 (\log T)^2
\implies
    \sum_{s=1}^t
    B_{s,k}(X_{s,k}) 
\ge 
    c_2 \sum_{s=1}^t X_{s,k}^{a_k} \;.
\end{equation}
We prove \eqref{e:lemma:LB:bernoullis} in
\Cref{lemma:LB:bernoullis} in Appendix~\ref{s:appe-prelim}, 
for the general diminishing-return setting, leveraging the fact that, for all $k\in[K]$,  
$
\brb{ 
X_{t,k}^{a_k}
-
B_{t,k}( X_{t,k} )
}_{t \in \N}
$
is a martingale difference sequence, together with some specific properties of our setting and a concentration argument (a version of Bernstein's inequality for martingales).

To be able to exploit the consequent in \eqref{e:lemma:LB:bernoullis}, we must show that its antecedent is satisfied. 
To this end, we prove in
\Cref{cor:Hoeffding:gains} in 
Appendix~\ref{s:appe-prelim}
that one can
define a ``good event'' $\cE^\delta_T$, 
as the event that all the $\UCB_{t,k}^\delta$ estimators of all the $\mu_k$ maintained by \Cref{algo:ucbowski} are sufficiently accurate approximations of $\mu_k$ (see also \eqref{eq:good:event:ET-main} in \Cref{s:upper-bound} for the precise expression of $\cE^\delta_T$), 
such that $\cE^\delta_T$ occurs with probability at least $1-KT\delta$.
Consequently, we prove that we can find a $T_0 = \Theta\brb{ (\log T)^2 }$ such that, under $\cE^\delta_T$ and $\cG_T$, it holds that  
$
    \sum_{s=1}^t X_{s,k}^{a_k}
\ge 
    c_1 (\log T)^2
$
for all $t\ge T_0$.
Fix such a $T_0$ and $\cE^\delta_T$.

\emph{Step 2: \eqref{e:quadratic-loss}.}
For any $m\ceq(m_1,\dots,m_K)\in[0,\iop)^K$, define the function $\Phi_m \colon \Delta_K \to [0,\infty)$, for all $x \ceq (x_1,\dots,x_K)\in\Delta_K$, by $\Phi_m(x) \ceq \sum_{k \in [K]} m_k x_k^{a_k}$.
In words, $\Phi_m$ is the expected reward function when the underlying mean reward of each task $k\in[K]$ is $m_k$.
Hence, denoting, for any $m\in [0,\iop)^K$, a maximizer (chosen arbitrarily in case of ties) of $\Phi_m$ by $x^\star(m)$, we have that $x^\star(\mu)$ is the optimal budget allocation and $x^\star(\UCB_{t-1}^\delta)$ is the action played by \Cref{algo:ucbowski} at time $t$.

The next claim is the keystone to obtaining a polylogarithmic regret rate.
It states that the instantaneous regret due to maximizing a slightly misspecified reward function $\Phi_{\tilde{m}}$ when the correct one is $\Phi_m$ is only quadratic in the distance between $m$ and $\tilde m$, as long as the components of $m$ and $\tilde m$ are bounded away from zero.
More precisely, for any $m_{\inf}\in(0,1)$, there exists $c_3>0$ (independent of $T$) such that, for any two vectors $m,\tilde m \in [m_{\inf},2]^K$, it holds that
\begin{equation}
\label{e:quadratic-loss}
    \Phi_m \brb{ x^\star(m) } - \Phi_m \brb{ x^\star (\tilde{m}) }
\le 
    c_3 \lno{m-\tilde{m}}_2^2 \;.
\end{equation}
We prove \eqref{e:quadratic-loss} in 
\Cref{lemma:stability} in
Appendix~\ref{s:appe-speed-ups}
for the general diminishing-return setting, leveraging the strong concavity of $F_1, \dots, F_K$, as well as their smoothness on $[x_{\inf},1]$, for any $x_{\inf}>0$.

Now, to be able to exploit \eqref{e:quadratic-loss}, we must show that all components of $x^\star(\mu)$ and $x^\star(\UCB_{t}^\delta)$, for $t$ large enough, are uniformly bounded below by some $c_4>0$ (independent of $T$, and playing the role of $x_{\inf}$ above).
We prove this in 
\Cref{lemma:lb:cum:budget} in
Appendix~\ref{s:appe-speed-ups} 
for the general diminishing-return setting.

\emph{Step 3: Bounding the regret.}
Recalling that $\delta \coloneqq \frac{1}{(KT)^2}$, $T_0 = \Theta\brb{  (\log T)^2 }$ and exploiting \eqref{e:lemma:LB:bernoullis} and \eqref{e:quadratic-loss}, we have, under $\cE^\delta_T$ and $\cG_T$ (thus, with probability at least $1-\nicefrac{K}{T}-KT\delta$), that, with the realized regret
$\cR_T$ in \eqref{eq:idealized:regret}
satisfying $\E(\cR_T) = R_T$,
\begin{align*}
 \cR_T
&
\overset{\phantom{(\circ)}}{\le}
    K T_0
    +
    \sum_{t=T_0}^T 
    \Brb{
        \Phi_{\mu} \brb{ x^\star( \mu ) }
        -
        \Phi_{\mu} \brb{ x^\star( \UCB^\delta_{t-1} ) }
    }
\overset{\eqref{e:quadratic-loss}}{\le} 
    K T_0
    +
    \sum_{t=T_0}^T 
    c_3 \lno{\mu - \UCB_{t-1}^\delta}_2^2
\\ 
& 
\overset{(\circ)}{\le}
    K T_0
    +
    4 c_3
    \sum_{t=T_0}^T 
    \sum_{k=1}^K 
    \frac{\log(\frac{2}{\delta})}{1+\sum_{s=1}^{t-1}
    B_{s,k}(X_{s,k})
    }
\overset{\eqref{e:lemma:LB:bernoullis}}{\le} 
    K T_0
+
    \frac{4 c_3}{c_2}
    \log \lrb{ \frac{2}{\delta} }
    \sum_{t=T_0}^T
    \sum_{k=1}^K
    \frac{1}{\sum_{s=1}^{t-1}
    X_{s,k}^{a_k}
} 
\\
& \overset{(*)}{\le} 
    K (T_0+2)
    + 
    \frac{4 c_3 }{ c_2 c_4^{a_{\max}}}
    K
    \log \lrb{ \frac{2}{\delta} }
    \sum_{t=T_0+2}^T
    \frac{1}{ t-T_0} 
\le 
    K (T_0+2)
    + 
    \frac{4 c_3 }{ c_2 c_4^{a_{\max}}}
    K
    \log \lrb{ \frac{2}{\delta} }
    \log(T),
\end{align*}
where $(\circ)$ follows from being under the good event $\cE_T^\delta$, and $(*)$ from $X_{s,k}^{a_k} \ge c_4^{a_k} \ge c_4^{a_{\max}}$.
Finally, \Cref{prop:explicit:constant} in Appendix~\ref{s:appe-speed-ups} converts this bound into the one stated in \Cref{t:upper-bound-speed-up}.
\end{proofsketch}

\begin{remark}[Interpretation of the assumptions for the fast-rate results]
\label{rem:fast-rate-scope}
The assumptions used here and for the general fast-rate result in
Appendix~\ref{s:appe-speed-ups} provide three main ingredients for the proofs:
eventual interiority (the selected $X_t$ is bounded away from zero eventually, see \Cref{lemma:lb:cum:budget} for details), local quadratic
stability of the UCB optimizer when the mean vector is changed (see \Cref{lemma:stability} for details), and a bounded away from zero observation rate (see \Cref{lemma:LB:bernoullis} for details).
Note that eventual interiority is obtained thanks to the infinite derivative of $F_k$ at zero.
\end{remark}

\begin{remark}[Open problem: optimality of the $(\log T)^2$ rate]
We do not know if the $(\log T)^2$ dependence in
\Cref{t:upper-bound-speed-up} is exactly optimal.
Note that it is possible to provide a lower bound of order $\log T$ in some cases, still leaving a gap of order $\log T$. We do not provide the details for the sake of concision.
\end{remark}

\begin{remark}[Approximate optimization of the UCB function]
\label{rem:approx-oracle}
Note that in \Cref{algo:ucbowski} we assume that the set of exact maximizers of the UCB function is non-empty and that the algorithm selects such a maximizer measurably from the available history and its independent internal randomness, in line with many other references \citep{srinivas2012information,bect2019supermartingale,iwazaki2026improved}. Hence our results can be interpreted as being information-theoretic. For this algorithm, worst-case regret is taken over instances satisfying this oracle assumption.
Nevertheless, our results could be extended to the case where the chosen $X_t$ in \Cref{algo:ucbowski} maximizes the UCB function up to a gap $\varepsilon_t$. If the sequence $(\varepsilon_t)_t$ decays fast enough, one could extend \Cref{t:upper-bound-speed-up}, and also \Cref{t:upper-bound-worst-case} below, preserving the same rates.
Obtaining maximizers up to a gap $\varepsilon_t$ can be achieved in a computationally efficient way under suitable assumptions on the $F_k$'s, such as concavity as in \Cref{t:upper-bound-speed-up}.
\end{remark}

\begin{remark}[Approximately known budget-to-probability functions $F_k$s]
\label{rem:unknown:F}
If the functions $F_k$s are only known up to a gap $\varepsilon'_t$ in the supremum norm, then optimizing the UCB function defined from approximated $F_k$s boils down to optimizing the real UCB function (defined from the true $F_k$s) up to an optimization error $\varepsilon_t$. Hence the previous remark applies and theoretical extensions of our results are possible (in settings that exclude Appendix~\ref{s:necessity}, of course).
\end{remark}

\paragraph{Comparison with linear bandits}
The reader may wonder whether the fast rates achieved by our algorithm can also be obtained by existing methods from the bandit literature (including standard linear bandit algorithms), such as those reviewed in \Cref{s:related_works}.
We now show that this is \emph{not} the case: even in extremely benign instances among those from diminishing-return regimes, \emph{any} method that relies solely on bandit feedback must incur regret of order $\Omega(\sqrt{T})$.
Specifically, the following holds.

\begin{theorem} \label{t:LB:linear:bandit}
    Suppose $K=2$ and for every $b \in [0,1]$, the budget-to-success curves are defined by $F_1(b) \coloneqq \sqrt{b} \eqqcolon F_2(b)$.
    Consider any algorithm running on bandit feedback, that is, in \Cref{a:learning-protocol}, the feedback $\mathscr{F}_t$ is the gain $\ban{G_t,\, B_t(X_t)}$. Then, for any time horizon $T \ge 4$, there exists a valid instance of our problem with expected reward vector $\mu$ such that $\min \lrb{\mu_1, \mu_2} \ge \frac{3}{8}$ and $R_T \ge \frac{1}{140}\cdot\sqrt{T}$ over that instance.
\end{theorem}
We present here a proof sketch and defer the full proof of \Cref{t:LB:linear:bandit} to Appendix~\ref{s:bandit:counter}.
\begin{proofsketch}
Under $F_1(b)=\sqrt{b}=F_2(b)$ and $K=2$, any budget allocation $x\in\Delta_{2}$ corresponds to the unit vector $(\sqrt{x_1},\sqrt{x_2})$ in the first quadrant. Hence, we can parameterize actions by angles $\vartheta\in[0,\pi/2]$ via the relation $(\sqrt{x_1},\sqrt{x_2})=(\cos\vartheta,\sin\vartheta)$.
Thus, the expected reward is $\langle \mu,(\cos\vartheta,\sin\vartheta)\rangle$, and the optimal action corresponds to the angle whose direction is
aligned with the (unknown) mean vector $\mu$.

We construct two hard instances by perturbing $(1/2,1/2)$:
\[
    \mu^+ \coloneqq \lrb{\tfrac12+\e,\tfrac12-\e},
    \qquad
    \mu^- \coloneqq \lrb{\tfrac12-\e,\tfrac12+\e}.
\]
Their optimal directions lie on opposite sides of $\pi/4$ (indeed, they are aligned with $\mu^+$ and $\mu^-$, respectively).
If the learner plays near $\vartheta=\pi/4$ (the direction associated to $(1/2,1/2)$), then bandit feedback is essentially
uninformative: the one-step information to distinguish $\mu^+$ from $\mu^-$ is
$O\!\brb{\e^2\labs{\vartheta-\pi/4}^2}$.
Consequently, up to constants, the total KL is controlled by
$\e^2\sum_{t=1}^T \E\!\lsb{\labs{\vartheta_t-\pi/4}^2}$.

Therefore, to gather enough information to distinguish the two instances, the learner must play at angles that are non-negligibly away from $\pi/4$. However, such angles are necessarily far from the optimal direction in at least one of the two instances, and the incurred regret is lower bounded by a quantity proportional to the squared angular distance from that instance's optimum.
This creates an ``apple-tasting''-type exploration-information tradeoff, with the additional difficulty of a
continuum of actions.
Balancing the two effects with $\e=\Theta(T^{-1/4})$ yields a $\Omega(\sqrt{T})$ lower bound.
\end{proofsketch}

We stress that, since the instance in \Cref{t:LB:linear:bandit} belongs to the diminishing-return class analyzed in \Cref{t:upper-bound-speed-up}, then, in contrast to bandit algorithms, our algorithm achieves polylogarithmic regret guarantee in this setting: 
when $K=2$ and $F_1,F_2$ are as in \Cref{t:LB:linear:bandit}, there exist universal constants $c,C>0$ such that for any valid instance with $\min(\mu_1,\mu_2)\ge \frac{3}{8}$, the regret of our algorithm is at most $c+C \cdot (\log T)^2$.
For a precise formulation of this claim, 
see \Cref{lemma:explicit:Ktwo} in Appendix~\ref{s:appe-speed-ups}.

\section{Worst-case analysis of Algorithm \ref{algo:ucbowski} and matching lower bound}
\label{s:upper-bound}

We now analyze the worst-case performance of \Cref{algo:ucbowski} for any family of budget-to-success curves and reward distributions.
We show that \Cref{algo:ucbowski} is worst-case optimal, both in $K$ and $T$,
up to a logarithmic factor. Due to space constraints, we present here a proof sketch of the next result while deferring the full proof to Appendix~\ref{s:appe-upper-bound-distribution-free}.
We also provide a schematic of the main proof steps in \Cref{fig:worst-case-mechanism}.

\begin{figure}[t]
\centering
\begin{tikzpicture}[
    font=\scriptsize,
    node distance=4mm,
    box/.style={
        draw,
        rounded corners,
        align=center,
        inner sep=3pt,
        text width=.205\linewidth
    }
]
\node[box] (a) {
we need to bound\\
$\displaystyle
\E
\sum_k
\sum_t
\frac{F_k(X_{t,k})}
{\sqrt{1+N_{t-1,k}}}$
};

\node[box, right=of a] (b) {
we condition on\\
the past\\
$\E[B_{t,k}(X_{t,k})\mid\mathcal \mbox{past}]
 =F_k(X_{t,k})$
};

\node[box, right=of b] (c) {
we obtain \\
$\displaystyle
\E
\sum_k \sum_t
\frac{B_{t,k}(X_{t,k})}
{\sqrt{1+N_{t-1,k}}}
\lesssim 
\E \sum_k \sqrt{N_{T,k}}$
};

\node[box, right=of c] (d) {
Cauchy--Schwarz:\\
$\displaystyle
\sum_k\E\sqrt{N_{T,k}}
\le
\sqrt{K\sum_k\E N_{T,k}}$
};

\draw[->] (a) -- (b);
\draw[->] (b) -- (c);
\draw[->] (c) -- (d);
\end{tikzpicture}
\caption{Schematic of the main proof steps for \Cref{t:upper-bound-worst-case}.}
\label{fig:worst-case-mechanism}
\end{figure}

\begin{samepage}
\begin{theorem}
    \label{t:upper-bound-worst-case}
    For any time horizon $T$, if we run \Cref{algo:ucbowski} with parameters $T$, $\delta \coloneqq \frac{1}{(KT)^2}$, its regret satisfies
    \[
        R_T
    =
        \widetilde{\cO} \lrb{ 
        1+
        \sqrt{K}
        \sqrt{
        \sum_{t=1}^{T}  \E \lsb{ \sum_{k=1}^K F_k(X_{t,k}) } } }
        \;.
    \]
   Upper bounding $F_k\le 1$, the worst-case regret of \Cref{algo:ucbowski} satisfies
    $
        R_T^\star
    =
        \widetilde{\cO} \brb{ K \sqrt{T} }
    $.
\end{theorem}
\end{samepage}
\begin{proofsketch}
Fix any time horizon $T\ge 2$, and let $\delta \ceq \frac{1}{(KT)^2}$.
For conciseness, we introduce, for each $t\in [T]$ and $k\in[K]$, 
$
    N_{t,k}
\coloneqq
    \sum_{s=1}^t B_{s,k}(X_{s,k})\;,
$
and, for all $\tau \in [N_{t,k}]$, we define $\tilde{G}_{\tau,k} \coloneqq G_{S_\tau,k}$, where $S_\tau$ is the random time step $s\in \N$ such that $B_{s,k}=1$ for the $\tau$-th time,
i.e., such that $\sum_{s<S_\tau} B_{s,k} = \tau-1$ and $\sum_{s \le S_\tau} B_{s,k} = \tau$. 
Then, note that we have for each $t\in [T]$ and $k\in[K]$,
\begin{equation} 
\label{eq:intro:N:Gtilde}
    \UCB^\delta_{t,k}
=
    \frac{ \sum_{\tau=1}^{N_{t,k}} \tilde{G}_{\tau,k}}{N_{t,k}}
    +
    \sqrt{\frac{\log\frac{2}{\delta}}{1+N_{t,k}}}\;.
\end{equation}
Define the ``good event'' as 
\begin{equation} 
\label{eq:good:event:ET-main}
    \cE_{T}^\delta
\coloneqq
    \bigcap_{k \in [K]}
    \bigcap_{t \in [N_{T,k}]}
    \lcb{ 
        \labs{ 
        \frac{1}{t}
        \sum_{\tau=1}^t 
        \tilde{G}_{\tau,k}
        -
        \mu_k
        }
    \le 
        \sqrt{
        \frac{\log\frac{2}{\delta}}{2t}
        }
    }.
\end{equation}
As we prove in \Cref{cor:Hoeffding:gains} in Appendix~\ref{s:appe-prelim}
(note that a simple Hoeffding's inequality does not suffice because the random variables $\tilde G_{\tau,k}$ depend on random time steps chosen as a function of the learner's actions), we have
$
    \Pb \bsb{\cE_{T}^\delta}
\ge
    1- K T \delta
$. 
Note that under the good event, for all $t \in [T]$, for all $k \in [K]$,
\begin{equation} \label{eq:UCB:good}
    \left| 
    \frac{\sum_{\tau=1}^{N_{t,k}}
	\tilde{G}_{\tau,k}
    }{N_{t,k}}
    -
    \mu_k
    \right| 
\le 
    \sqrt{
	\frac{\log\frac{2}{\delta}}{1+N_{t,k}}}.
\end{equation}
When $N_{t,k} = 0$, under the convention $0/0 = 0$, this is indeed equivalent to $| \mu_k| \le \sqrt{\log\frac{2}{\delta}}$ which is true because $\log\frac{2}{\delta} \ge \log 4 >1$. When $N_{t,k} \ge 1$, the above display holds because $\sqrt{
	\frac{\log\frac{2}{\delta}}{1+N_{t,k}}} \ge \sqrt{
	\frac{\log\frac{2}{\delta}}{2N_{t,k}}}$ and because the good event holds.
Now, fix an arbitrary $x^\star \in \Delta_K$.
We have
{\allowdisplaybreaks
\begin{align*}
&
    \sum_{t=1}^T
    \E \Bsb{ 
    \ban{\mu,F(x^\star)} 
    - 
    \ban{\mu, 
    F(X_t)}
    }
\le 
    K T \underbrace{ \brb{ 1-\Pb[ \cE_{T}^\delta ] }}_{\le KT\delta}
+
    \E \lsb{ 
    \I_{\cE_{T}^\delta}
    \sum_{t=1}^T  
    \ban{\mu,F(x^\star) - F(X_t)} 
    }
\\
&\quad\le
    \underbrace{K^2 T^2 \delta}_{\le 1}
+
    \E \Bbsb{ \I_{\cE_{T}^\delta} \sum_{t=1}^T \underbrace{\ban{ \UCB^\delta_{t-1}, F(x^\star) - F(X_t) }}_{\le 0 \text{ by definition of $X_t$}} }
+
    \E \Bbsb{ \I_{\cE_{T}^\delta} \sum_{t=1}^T \ban{ \underbrace{\UCB^\delta_{t-1} - \mu}_{\ge 0 \text{ in } \cE_{T}^\delta}, F(X_t) - F(x^\star) }}
\\
&\quad\le
    1
+
    \E \lsb{ \I_{\cE_{T}^\delta} \sum_{t=1}^T \Ban{ \UCB^\delta_{t-1} - \mu, \brb{F(X_t) - F(x^\star)}_+}}
=
    1
+
    (\star)\;.
\end{align*}
}
{
\allowdisplaybreaks
Hence, to complete the proof, it is sufficient to control $(\star)$ as follows:
\begin{align*}
    (\star)
&\le
    \E \lsb{
    \I_{\cE_{T}^{\delta}}
    \sum_{t=1}^T \sum_{k=1}^K
    2\sqrt{
        \frac{\log \frac{2}{\delta}}
        {1+\sum_{s=1}^{t-1}B_{s,k}(X_{s,k})}
    }
    \brb{F_k(X_{t,k}) - F_k(x^\star_k)}_+
    }
\\
&\le
    2\sqrt{\log \frac{2}{\delta}}\,
    \E \lsb{
    \sum_{t=1}^T \sum_{k=1}^K
    \frac{F_k(X_{t,k})}
    {\sqrt{1+\sum_{s=1}^{t-1}B_{s,k}(X_{s,k})}}
    }
=
    2\sqrt{\log \frac{2}{\delta}}\,
    \E \lsb{
    \sum_{t=1}^T \sum_{k=1}^K
    \frac{B_{t,k}(X_{t,k})}
    {\sqrt{1+\sum_{s=1}^{t-1}B_{s,k}(X_{s,k})}}
    }
\\
&=
    2\sqrt{\log \frac{2}{\delta}}\,
    \E \lsb{
    \sum_{k=1}^K
    \sum_{\tau=1}^{N_{T,k}}
    \frac{1}{\sqrt{\tau}}
    }
\le
    4\sqrt{\log \frac{2}{\delta}}\,
    \E \lsb{
    \sum_{k=1}^K
    \sqrt{N_{T,k}}
    }
\le
    4\sqrt{K\log \frac{2}{\delta}}\,
    \E \lsb{
    \sqrt{\sum_{k=1}^K N_{T,k}}
    }
\\
&\le
    4\sqrt{K\log \frac{2}{\delta}}\,
    \sqrt{
    \sum_{k=1}^K \E \lsb{N_{T,k}}
    }
=
    4\sqrt{K\log \frac{2}{\delta}}\,
    \sqrt{
    \sum_{t=1}^T
    \E \lsb{
    \sum_{k=1}^K F_k(X_{t,k})
    }
    }\;.
\end{align*}
}
Here, the first inequality holds because we are on the good event $\cE_T^\delta$ and by~\eqref{eq:UCB:good}.
The first equality follows from the tower rule after conditioning on the history up to the beginning of round $t$.
The second equality rewrites the sum over successful observations: on the $\tau$-th success of task $k$, the denominator is $\sqrt{\tau}$.
The next inequality uses $\sum_{\tau=1}^n \tau^{-1/2}\le 2\sqrt n$ for all $n\ge 0$.
The following inequality is Cauchy's inequality, and the last inequality is Jensen's inequality.
\end{proofsketch}

As an immediate corollary of \Cref{t:upper-bound-worst-case}, we recover the best known worst-case rate $\widetilde{\cO} \brb{  \sqrt{ KMT} }$ for combinatorial semi-bandits, see 
\citet{chen2013combinatorial} and
\citet[Theorem 30.2]{lattimore2020bandit}, as discussed in \Cref{s:setting}.

\begin{corollary}
In \Cref{t:upper-bound-worst-case}, if there is $M \in [K]$ such that $F_k =  \I_{[1/M,1]} $ for $k \in [K]$, then we have $R_T^\star = \widetilde{\cO} \brb{  \sqrt{ KMT} }$. 
\end{corollary}

We now state a 
worst-case lower bound on the regret,
not only for algorithms 
operating under censored semi-bandit feedback like \Cref{algo:ucbowski} but also for those operating under the richer full feedback.
This lower bound is tight up to a logarithmic factor.  
\begin{theorem}
\label{t:lower-bound-worst-case}
The worst-case regret of any full-feedback algorithm satisfies
$
    R_T^\star = \Omega \brb{ K \sqrt{T} }
$.
\end{theorem}
Due to space constraints, we defer the proof of this result to Appendix~\ref{s:appe-lower-bound-worst-case}.

\section{Conclusions, limitations and future work}
\label{s:conclusion}
We introduce a new stochastic bandit framework, which captures applications ranging from crowdsourcing markets to multi-platform autobidding, and to compute allocations in large-scale model training.
We design an algorithm and establish its minimax optimality with a regret guarantee of $\tilde{\Theta}\brb{ K \sqrt{T} }$. The same algorithm with no \emph{ad-hoc} tuning enjoys an instance-dependent polylogarithmic regret bound, when the budget-to-success curves satisfy a natural diminishing-returns property. To the best of our knowledge, these are the first tight guarantees for sequential budget allocation with censored semi-bandit feedback.

Beyond matching the leading-order rates, our analysis yields three structural insights.  
First, the censored semi-bandit feedback lets the learner decouple exploration across coordinates from reward estimation, giving a clean $\sqrt{T}$ frontier in the worst case.  Second, strong concavity acts as an implicit regularizer that yields logarithmic learning rates.
Third, even though our logarithmic rate is reminiscent of those in bandits with strongly convex losses, our result is proved through novel, non-trivial techniques. 
Indeed, we provide a counterexample showing that any bandit algorithm (including linear bandit algorithms) has to suffer $\Omega(\sqrt{T})$ regret even in benign instances of the diminishing-return setting, illustrating that standard bandit techniques are not sufficient in this setting.

A limitation of our theory is that it does not capture
other interesting variants: when the per-round budget varies with time (stochastically or adversarially), or when tasks (with deterministic or random completion times) can run for multiple time steps, and the feedback is only observed with some delay, at the completion of the tasks.
Lastly, one could consider a contextual version of this problem where some context related to the tasks is revealed to the learner before they are asked to allocate budgets at each time step.
We leave these extensions open for future work.

\begin{ack}
We are grateful to the anonymous reviewers for their constructive comments that led to an improvement of the paper.
FB has benefitted from the AI Interdisciplinary Institute ANITI, funded by the France 2030 program under the Grant agreement n°ANR-23-IACL-0002. NCB acknowledges the financial support from the EU Horizon RIA under grant agreement 101120237, project ELIAS (European Lighthouse of AI for Sustainability).
TC gratefully acknowledges the support of the Natural Sciences and Engineering Research Council
of Canada (NSERC) through grant RGPIN-2023-03688 (Discovery Grants Program).
\end{ack}

\bibliographystyle{abbrvnat}
\bibliography{biblio}

\appendix

\section{What if the budget-to-success curves are unknown?}
\label{s:necessity}
In this section we investigate the setting in which the learner does not know the budget-to-success curves in advance.
We prove a particularly strong impossibility result: even when $K=2$, the learner knows that $F_1$ and $F_2$ are indicator functions, and the rewards are deterministic with $G_{t}\equiv (1,1)$, the worst-case regret is linear.

\begin{theorem}[Unknown budget-to-success curves imply linear minimax regret]
\label{t:unknown-F-linear-lb}
Let $K\coloneqq 2$ and deterministic rewards $G_{t,1}=G_{t,2}\equiv 1$ for all $t$.
Let $\Delta_2=\{x\in[0,1]^2:x_1+x_2=1\}$.
For each $b\in[0,1]$, define the budget-to-success curves
\[
    F_1^{(b)}
\coloneqq
    \I_{[b,1]}\;,
\qquad
    F_2^{(b)}
\coloneqq
    \I_{[1-b,1]}\;.
\]
Suppose that the learner knows that $(F_1,F_2)$ belongs to the family $\brb{(F_1^{b},F_2^{b})}_{b \in [0,1]}$ and that $G_{t,1}=G_{t,2}\equiv 1$ for all $t$.
Then any (randomized) learning algorithm incurs linear worst-case regret:
\[
    \sup_{b\in[0,1]} R_T(F_1^{b},F_2^{b})
=
    T,
\qquad\text{for all }T\in\N.
\]
\end{theorem}

\begin{proof}
Fix a (randomized) learning algorithm $\alpha$.
We assume that $\alpha$ has access sequentially to i.i.d.\ $[0,1]$-valued uniform random variables $U_1,\dots,U_T$ as random seeds.

Fix $b\in[0,1]$ and consider the instance with deterministic rewards $G_t\equiv (1,1)$ and budget-to-success curves $F_1^{(b)}=\I_{[b,1]}$ and $F_2^{(b)}=\I_{[1-b,1]}$.
If at time $t$ the learner plays an allocation $x\in\Delta_2$, then the completion feedback is
\[
    B_t(x)=\lrb{B_{t,1}(x_1),B_{t,2}(x_2)}\;,
\]
and the reward feedback is
\[
    \lrb{G_{t,1}B_{t,1}(x_1),G_{t,2}B_{t,2}(x_2)}\;.
\]
Since $G_t\equiv (1,1)$, these two vectors coincide, so the reward part carries no additional information beyond the completion vector.
Moreover, since $F_1^{(b)}$ and $F_2^{(b)}$ are indicator functions, the corresponding Bernoulli variables are degenerate, and we have
\begin{equation}
\label{eq:01_10_11}
    B_{t,1}(x_1)
=
    \I\lcb{b\le x_1}\;,
\qquad
    B_{t,2}(x_2)
=
    \I\lcb{1-b\le x_2}
=
    \I\lcb{x_1\le b}\;.
\end{equation}
In particular,
\[
    B_t(x)=\lrb{B_{t,1}(x_1),B_{t,2}(x_2)}\in\lcb{(0,1),(1,0),(1,1)}\;.
\]

Hence, without loss of generality, we can assume that $\alpha$ is represented by a sequence of functions $\alpha_1,\dots,\alpha_T$ such that
\[
    \alpha_t:([0,1]\times\lcb{(0,1),(1,0),(1,1)})^{t-1}\times[0,1]\to\Delta_2\;,
\]
and that for each $b\in[0,1]$ the generated allocations $(X_t^{(b)})_{t\in[T]}$ satisfy $X_1^{(b)}\coloneqq\alpha_1(U_1)$ and, for $t\ge 2$,
\[
    X_t^{(b)}
\coloneqq
    \alpha_t\lrb{U_1,B_{1}(X_{1}^{(b)}),\dots,U_{t-1},B_{t-1}(X_{t-1}^{(b)}),U_t}\;.
\]

From \eqref{eq:01_10_11} it follows that for the instance indexed by $b$ the unique allocation for which both tasks are completed is
\[
    \bar{x}(b)\coloneqq (b,1-b)\;.
\]
In particular, the optimal (expected) reward per round equals $2$, because
\[
    \sup_{x\in\Delta_2}\ban{\mu,F^{(b)}(x)}=\max_{x\in\Delta_2}\ban{\mu,F^{(b)}(x)}=\ban{(1,1),(1,1)}=2\;.
\]
For a generic allocation $x\in\Delta_2$, \eqref{eq:01_10_11} implies that exactly one task is completed unless $x_1=b$, in which case both are completed.
Therefore the reward at time $t$ as a function of the allocated budget $x\in\Delta_2$ is
\[
    r_t^{(b)}(x)
\coloneqq
    G_{t,1}B_{t,1}(x_1)+G_{t,2}B_{t,2}(x_2)
=
    1+\I\lcb{x_1=b}\;,
\]
and for the sequence $(X_t^{(b)})_{t\in[T]}$ we obtain that the regret, when the underlying instance has parameter $b \in [0,1]$, is
\begin{equation}
\label{e:decomposition}
    R_T^b
\coloneqq
    2T-\sum_{t=1}^T \E\lsb{r_t^{(b)}\lrb{X_t^{(b)}}}
=
    T-\sum_{t=1}^T \E\lsb{\I\lcb{X_{t,1}^{(b)}=b}}\;.
\end{equation}
Consequently, using Fubini/Tonelli's theorem, we have
\[
    \sup_{b\in[0,1]} R_T^b
\ge
    \int_0^1 R_T^b \dif b
=
    T-\sum_{t=1}^T \E\lsb{\int_0^1 \I\lcb{X_{t,1}^{(b)}=b}\dif b}\;.
\]
Thus it suffices to prove that for every fixed $t\in[T]$ ($\Pb$-almost surely),
\begin{equation}
\label{eq:key_integral_zero}
    \int_0^1 \I\lcb{X_{t,1}^{(b)}=b}\dif b=0\;.
\end{equation}

Fix $t\in[T]$ and fix arbitrary seeds $u_1,\dots,u_T\in[0,1]$.
With these seeds fixed, the entire trajectory is deterministic for each parameter $b$.
We denote by $x_t^{(b)}(u_1,\dots,u_t)\in\Delta_2$ the allocation at time $t$ when the instance parameter is $b$.
Define the set of fixed points
\[
    A_t(u_1,\dots,u_t)
\coloneqq
    \lcb{b\in[0,1]:x_{t,1}^{(b)}(u_1,\dots,u_t)=b}\;.
\]
We claim that $A_t(u_1,\dots,u_t)$ is finite.
This would imply
\[
    \int_0^1 \I\lcb{x_{t,1}^{(b)}(u_1,\dots,u_t)=b}\dif b=0\;,
\]
and since the seeds $(u_1,\dots,u_T)$ were arbitrary,  \eqref{eq:key_integral_zero} would follow from
\[
    \int_0^1 \I\lcb{X_{t,1}^{(b)}=b}\dif b
=
    \int_0^1 \I\lcb{x_{t,1}^{(b)}(U_1,\dots,U_t)=b}\dif b=0\;.
\]

It remains to prove that $A_t(u_1,\dots,u_t)$ is finite.
We do so by induction on $t$.
For $t=1$, $x_{1,1}^{(b)}(u_1)=(\alpha_1(u_1))_1$ does not depend on $b$, hence $A_1(u_1)$ is a singleton.

Assume now that the claim holds up to time $t-1$ and fix $(u_1,\dots,u_t)$.
Define
\[
    S_{t-1}(u_1,\dots,u_{t-1})
\coloneqq
    \bigcup_{s=1}^{t-1} A_s(u_1,\dots,u_s)\;.
\]
Since each $A_s(u_1,\dots,u_s)$ is finite by the induction hypothesis, the set $S_{t-1}(u_1,\dots,u_{t-1})$ is finite as well.

Consider any interval $I$ contained in $[0,1]\setminus S_{t-1}(u_1,\dots,u_{t-1})$.
For every $b\in I$ and every $s\le t-1$ we have $b\ne x_{s,1}^{(b)}(u_1,\dots,u_s)$, hence by \eqref{eq:01_10_11} the completion vector
\[
    B_{s}\lrb{x_{s}^{(b)}(u_1,\dots,u_s)}
=
    \lrb{\I\lcb{b\le x_{s,1}^{(b)}(u_1,\dots,u_s)},\I\lcb{x_{s,1}^{(b)}(u_1,\dots,u_s)\le b}}
\]
takes values in $\lcb{(1,0),(0,1)}$ and is constant over $I$ (as a function of $b$).
Therefore the entire history
\[
    \lrb{u_1,B_{1}\lrb{x_{1}^{(b)}(u_1)},\dots,u_{t-1},B_{t-1}\lrb{x_{t-1}^{(b)}(u_1,\dots,u_{t-1})}}
\]
is constant over $b\in I$, and since $\alpha_t$ is deterministic given this history and $u_t$, it follows that $x_{t,1}^{(b)}(u_1,\dots,u_t)$ is constant over $b\in I$.

We have shown that the map $b\mapsto x_{t,1}^{(b)}(u_1,\dots,u_t)$ is piecewise constant on $[0,1]$, with possible breakpoints contained in the finite set $S_{t-1}(u_1,\dots,u_{t-1})$.
On any piece where $b\mapsto x_{t,1}^{(b)}(u_1,\dots,u_t)$ is constant, the equation $x_{t,1}^{(b)}(u_1,\dots,u_t)=b$ can hold for at most one value of $b$.
Since there are finitely many pieces, this shows that $A_t(u_1,\dots,u_t)$ is finite, completing the induction.

Combining the above with \eqref{eq:key_integral_zero}, we obtain $\int_0^1 R_T^b \dif b=T$, and therefore $\sup_{b\in[0,1]} R_T^b\ge T$.
To prove that this is in fact an equality, we notice that \eqref{e:decomposition} immediately implies that for every $b \in [0,1]$
\[
    R_T^b \le T\;,
\]
and hence $\sup_{b \in [0,1]} R_T^b \le T$, from which it follows
\[
   \sup_{b \in [0,1]}  R_T^b = T\;.
\]
\end{proof}

\section{Preliminary results}
\label{s:appe-prelim}

Let $U$ encode the learner's internal randomization, independently of the environment, and write $D_t=(B_{t,1}(X_{t,1}),\ldots,B_{t,K}(X_{t,K}))$. For the analysis, define
\begin{equation} \label{eq:Gt}
\cF_0\coloneqq\sigma(U),
\qquad
\cF_t\coloneqq\sigma\bigl(U,(D_s,G_s)_{s=1}^t\bigr),
\quad t\in\N.
\end{equation}
This filtration includes past unobserved rewards only for the analysis and does not change the learner's feedback. Since the learner chooses $X_t$ from its internal randomization and past feedback, $X_t$ is $\cF_{t-1}$-measurable.

The next lemma and its two corollaries show that successive rewards observed for a given task can be controlled as if they were i.i.d.\ variables, despite their observation times depending on past choices of the algorithm.

\begin{lemma} \label{lemma:observed:gains:iid}
Fix $k \in [K]$. Let $N \in [T]$ be deterministic and such that 
\[
\mathbb{P}
\left[
\sum_{t=1}^\infty
B_{t,k}(X_{t,k})
\ge N
\right]
=1.
\]
Let $(Y_{1,k},\ldots,Y_{N,k}) \ceq (G_{t_1,k} , \ldots , G_{t_N,k})$
where $t_1 < \cdots < t_N$ are the (random) $N$ first indices $t \in \N$ such that $B_{t,k}(X_{t,k}) =1$. Then $Y_{1,k},\ldots,Y_{N,k}$ are i.i.d.\ and distributed as $G_{1,k}$.
\end{lemma}

\begin{proof}
Let $\psi$ be a bounded continuous function on $[0,1]$. 
Let $A_1 = \bcb{ 
B_{1,k}(X_{1,k}) = 1
} = \{ t_1 = 1\}$
and for $t \ge 2$, $A_t = \bcb{ 
B_{1,k}(X_{1,k}) = 0 , \ldots, B_{t-1,k}(X_{t-1,k}) = 0,
B_{t,k}(X_{t,k}) = 1
} = \{ t_1 = t\}$.  We have, using the law of total expectation,
\begin{equation}
\label{eq:decom:Epsi}
\E \bsb{
\psi(Y_{1,k})
}
=
\sum_{t=1}^{\infty}
\E \bsb{
\psi(Y_{1,k})
\mid
A_t
}
\Pb [A_t]
=
\sum_{t=1}^{\infty}
\E \bsb{
\psi(G_{t,k})
\mid
A_t
}
\Pb [A_t].
\end{equation}

Then, using $\cF_t$ from \eqref{eq:Gt}, noting that $A_t$ is a measurable event of $B_{t,k}(X_{t,k})$ and the variables in $\cF_{t-1}$, and using the law of total expectation,
\[
\E \bsb{
\psi(G_{t,k})
\mid
A_t
}
=
\E \Bsb{
\E \bsb{ 
\psi(G_{t,k})
\mid
\cF_{t-1}
,
B_{t,k}(X_{t,k})
}
\mid
A_t
}.
\]
Since $G_{t,k}$ is independent of $B_{t,k}(X_{t,k})$ and $\cF_{t-1}$, and has the same distribution as $G_{1,k}$, we have
\[
\E \bsb{
\psi(G_{t,k})
\mid
\cF_{t-1}
,
B_{t,k}(X_{t,k})
}
=
\E \bsb{ 
\psi(G_{1,k})
}.
\]
Hence
\[
\E \bsb{ 
\psi(G_{t,k})
\mid
A_t
}
=\E \bsb{ 
\psi(G_{1,k})
}\;.
\]
It follows from \eqref{eq:decom:Epsi} that 
\[
    \E \bsb{ 
    \psi(Y_{1,k})
    }
=
    \sum_{t=1}^{\infty}
    \E \bsb{ 
    \psi(G_{1,k})
    }
    \Pb [A_t]
=
    \E \bsb{
    \psi(G_{1,k})
    }.
\]
Hence $Y_{1,k}$ is distributed as $G_{1,k}$. 

Next, for $\ell \in [N]$, let \[
\cS_{\ell}
=
\lcb{ (s_1 , \dots , s_{\ell}) \in \N^\ell :
s_1 < \dots < s_\ell
}.
\]
Note that $\cS_{n+1}$
is the set of values for $(t_1,\ldots,t_{n+1})$ that can be reached. Consider any two bounded measurable functions $\phi \colon [0,1]^n \to \bbR$ and $\psi \colon [0,1] \to \bbR$. Define
\[
    T(s_1 , \ldots,s_{n+1})
\coloneqq
    \lcb{
    (t_1,\ldots,t_{n+1})
    =
    (s_1,\ldots,s_{n+1} )}.
\]
Then we have, using the law of total expectation, 
\begin{align*}
&
    \E \bsb{ 
    \phi(Y_{1,k} , \ldots ,Y_{n,k})
    \psi(Y_{n+1,k})
    }
\\ 
&= 
    \sum_{(s_1,\ldots,s_{n+1}) \in \cS_{n+1}}
    \E \bsb{ 
    \phi(Y_{1,k} , \ldots ,Y_{n,k})
    \psi(Y_{n+1,k})
    \mid
    T(s_1,\ldots,s_{n+1})
    }
    \Pb \bsb{ T(s_1,\ldots,s_{n+1}) }
\\
&=
    \sum_{(s_1,\ldots,s_{n+1}) \in \cS_{n+1}}
    \E \bsb{ 
    \phi(G_{s_1,k} , \ldots ,G_{s_n,k})
    \psi(G_{s_{n+1},k})
    \mid
    T(s_1,\ldots,s_{n+1})
    }
    \Pb \bsb{ T(s_1,\ldots,s_{n+1}) }
\\ 
& =
    \sum_{(s_1,\ldots,s_{n+1}) \in \cS_{n+1}}
    \E \Bsb{ 
    \E \bsb{ 
    \phi(G_{s_1,k} , \ldots ,G_{s_n,k})
    \psi(G_{s_{n+1},k})
    \mid
    \cF_{s_{n+1}-1},
    B_{s_{n+1},k}(X_{s_{n+1},k})
    }
    \mid
    T(s_1,\ldots,s_{n+1})
    }
    \\
    &
    \hspace{3cm} \times
    \Pb \bsb{ T(s_1,\ldots,s_{n+1}) },
\end{align*}
where the last equality holds by the law of total expectation since
$T(s_1,\ldots,s_{n+1})$ is a measurable event of $B_{s_{n+1},k}(X_{s_{n+1},k})$ and the variables in $\cF_{s_{n+1}-1}$.

Then, conditional on $\cF_{s_{n+1}-1}$ and  $B_{s_{n+1},k}(X_{s_{n+1},k})$, only $G_{s_{n+1},k}$ remains random (and keeps its unconditional distribution). Hence, we obtain
\begin{align*}
&
    \E \bsb{ 
    \phi(Y_{1,k} , \ldots ,Y_{n,k})
    \psi(Y_{n+1,k})
    }
    \\ 
&= 
    \sum_{(s_1,\ldots,s_{n+1}) \in \cS_{n+1}}
    \E \Bsb{ 
    \phi(G_{s_1,k} , \ldots ,G_{s_n,k})
    \E \bsb{ 
    \psi(G_{s_{n+1},k})
    \mid
    \cF_{s_{n+1}-1},
    B_{s_{n+1},k}(X_{s_{n+1},k})
    }
    \mid
    T(s_1,\ldots,s_{n+1})
    }
    \\
&
    \hspace{3cm} \times
    \Pb \bsb{ T(s_1,\ldots,s_{n+1}) }
\\
& =
    \E \bsb{
    \psi(G_{1,k})
    }
    \sum_{(s_1,\ldots,s_{n+1}) \in \cS_{n+1}}
    \E \bsb{
    \phi(G_{s_1,k} , \ldots ,G_{s_n,k})
    \mid
    T(s_1,\ldots,s_{n+1})
    }
    \Pb \bsb{ T(s_1,\ldots,s_{n+1}) }
\\
& =
    \E \bsb{
    \psi(G_{1,k})
    }
    \sum_{(s_1,\ldots,s_{n+1}) \in \cS_{n+1}}
    \E \bsb{
    \phi(Y_{1,k} , \ldots ,Y_{n,k})
    \mid
    T(s_1,\ldots,s_{n+1})
    }
    \Pb \bsb{ T(s_1,\ldots,s_{n+1}) }
\\
& =
    \E \bsb{
    \psi(G_{1,k})
    }
    \E \bsb{ 
    \phi(Y_{1,k} , \ldots ,Y_{n,k})
    },
    \end{align*}
where the last equality follows from the law of total expectation.

Hence, by induction it follows that for any bounded measurable functions $\psi_1 , \ldots , \psi_N$ from $[0,1]$ to $\bbR$,
\[
\E \left[ 
\prod_{i=1}^N
\psi_i(Y_{i,k})
\right]
=
\prod_{i=1}^N
\E \left[ 
\psi_i(G_{1,k})
\right].
\]
This concludes the proof.
\end{proof}

\begin{corollary} \label{cor:as:iid}   Fix $N,T \in \N$ with $N \le T$. Fix $k \in [K]$. Let $E_{N,T}$ be the event 
\[
E_{N,T} \coloneqq
\left\{ 
\sum_{t=1}^T 
B_{t,k}(X_{t,k}) 
\ge N 
\right\}.
\]
Under the event $E_{N,T}$, define $t_1 , \ldots , t_N$ as the sequence of the $N$ first random times $t \in [T]$ where $B_{t,k}(X_{t,k}) =1$, with $1 \le t_1 < \dots < t_N \le T$.
Then, for any Borel subset $A \subset [0,1]^N$, 
\[
    \Pb \bsb{ 
    E_{N,T}
    \cap 
    \left\{
    \left( 
    G_{t_1,k},
    \dots,
    G_{t_N,k}
    \right)
    \in A
    \right\}
    }
\le 
    \Pb \bsb{ 
    \left( 
    G_{1,k},
    \dots,
    G_{N,k}
    \right)
    \in A
    }. 
    \]
\end{corollary}
\begin{proof}
    Consider the budget-to-success curves $F_1,\dots,F_K$.
    Consider new (time-dependent) functions with, for $t\in\N$, $k\in[K]$, and $x\in[0,1]$,
    \[
    \widetilde{F}_{t,k}(x) = 
    \begin{cases}
F_k(x) & ~ ~ \text{if $t \le T$} \\
1 ~ \text{(constant function)} & ~ ~ \text{if $t \ge T+1$}.
    \end{cases}
    \]
    These new $\widetilde{F}_{t,1}, \dots \widetilde{F}_{t,K}$ 
    define an auxiliary sampling process, for which we denote probabilities with $\widetilde{\Pb}$.
    We keep the reward process unchanged and extend the learner's allocation rule arbitrarily after round $T$. The proof of \Cref{lemma:observed:gains:iid} also applies to this construction, since it does not require the completion probabilities to be constant over time.
    The behavior of the gains and algorithm choices is unchanged from times $t=1$ to $t=T$ between the two settings. Hence
    \[
    \Pb \bsb{ 
E_{N,T}
\cap 
\left\{
\left( 
G_{t_1,k},
\dots,
G_{t_N,k}
\right)
\in A
\right\}
}
=
  \widetilde{\Pb} \bsb{ 
E_{N,T}
\cap 
\left\{
\left( 
G_{t_1,k},
\dots,
G_{t_N,k}
\right)
\in A
\right\}
}. 
    \]
    Consider $N\in[T]$.
Under the new setting, each arm is infinitely explored over times $t \in \N$ and so we can apply \Cref{lemma:observed:gains:iid} (with $t_1,
\dots,t_N$ defined as in this lemma with probability one) yielding
    \begin{align*}
      \widetilde{\Pb} \bsb{ 
    E_{N,T}
    \cap 
    \left\{
    \left( 
    G_{t_1,k},
    \dots,
    G_{t_N,k}
    \right)
    \in A
    \right\}
    }
& \le 
    \widetilde{\Pb} \bsb{ 
    \left( 
    G_{t_1,k},
    \dots,
    G_{t_N,k}
    \right)
    \in A
    }
\\
&=
    \widetilde{\Pb} \bsb{ 
    \left( 
    G_{1,k},
    \dots,
    G_{N,k}
    \right)
    \in A
    }
\\
&=
    \Pb \bsb{ 
    \left( 
    G_{1,k},
    \dots,
    G_{N,k}
    \right)
    \in A
    }. 
\end{align*}
This concludes the proof.
\end{proof}

\begin{corollary} 
    \label{cor:Hoeffding:gains}
With the notation of \Cref{cor:as:iid}, for any $\delta \in (0,1)$, we have, for $N,T \in \N$, $N \le T$,
\[
\Pb \lsb{
E_{N,T} 
\cap 
\left\{
\left|
\frac{1}{N}
\sum_{i=1}^N 
G_{t_i,k}
-
\mu_k
\right|
>
\sqrt{
\frac{\log(\frac{2}{\delta})}{2N}
}
\right\}
}
\le
\delta.
\]
\end{corollary}

\begin{proof}
    We apply \Cref{cor:as:iid} to the set 
    \[
    A = \left\{ 
(g_1,\ldots,g_N) \in [0,1]^N :
\left|
\frac{1}{N}
\sum_{i=1}^N 
g_i
-
\mu_k
\right|
>
\sqrt{
\frac{\log(\frac{2}{\delta})}{2N}
}
    \right\}
    \]
    and then we apply Hoeffding's inequality.
\end{proof}

The next crucial lemma shows that the number of times a task has been completed is not much smaller than the total amount of probability invested in the task (once at least a polylogarithmic amount of probability has been invested in the task).

\begin{lemma} \label{lemma:LB:bernoullis}
For any $T \ge 2$, with probability at least $1 - \frac{K}{T}$, the following event holds:
For $t \in [T]$ and $k \in [K]$, if
\[
\sum_{s=1}^t F_k(X_{s,k})
\ge 
100 \log(T)^2
\]
then 
\[
\sum_{s=1}^t
B_{s,k}(X_{s,k}) 
\ge 
\frac{\sum_{s=1}^t F_k(X_{s,k})}{2}.
\]

\end{lemma}

\begin{proof}
For $k \in [K]$, the sequence 
\[
    \brb{
     F_k(X_{t,k})
    -
     B_{t,k}(X_{t,k})
    }_{t \in \N}
\]
is a martingale difference sequence, together with the filtration $(\cF_t)_{t\ge0}$ defined in \eqref{eq:Gt}. Indeed, for all $t\in\N$,  $F_k(X_{t,k}) - B_{t,k}(X_{t,k})$ is  $\cF_t$-measurable and 
\[
    \E \bsb{ 
    F_k(X_{t,k})
    -
    B_{t,k}(X_{t,k})  \mid \cF_{t-1}
    }
=
    0
\]
since conditional on $\cF_{t-1}$, $F_k(X_{t,k})$ is fixed and $B_{t,k}(X_{t,k})$ follows a Bernoulli distribution with parameter $F_k(X_{t,k})$.

For the same reason,
\[
\Var \bsb{ 
    F_k(X_{t,k})
    -
    B_{t,k}(X_{t,k})  \mid \cF_{t-1}
    }
=
    F_k(X_{t,k}) \brb{1 - F_k(X_{t,k})}.
\]
Let us define
\[
    v_t 
=
    \sum_{s=1}^t
    \Var
    \bsb{ 
    F_k(X_{s,k})
    -
    B_{s,k}(X_{s,k})  \mid \cF_{s-1}
    }
\le 
    \sum_{s=1}^t F_k( X_{s,k} ).
\]
Hence from Corollary 16 in \cite{cesa2005minimizing} (a version of Bernstein's inequality for martingales), we have, for $\delta \in (0,1)$ and $t \in [T]$,
\[
    \Pb
    \lsb{ 
    \sum_{s=1}^t F_k(X_{s,k})
    -
    \sum_{s=1}^t
    B_{s,k}(X_{s,k})
>
    \sqrt{2(v_t + 1)
    \log(t/\delta)}
    +
    \frac{\sqrt{2}}{3}
    \log(t/\delta)
    }
\le 
    \delta.
\]
Let us take $\delta  = \frac{1}{T^2}$ and replace $v_t$ by its upper bound yielding
\[
    \Pb
    \lsb{ 
    \sum_{s=1}^t F_k(X_{s,k})
    -
    \sum_{s=1}^t
    B_{s,k}(X_{s,k})
>
    \sqrt{2\left(\sum_{s=1}^t F_k( X_{s,k} ) + 1\right)
    3 \log(T)}
    +
    \frac{3 \sqrt{2}}{3}
    \log(T)
    }
\le 
    \frac{1}{T^2}.
\]
If $\sum_{s=1}^t F_k(X_{s,k}) \ge 100 \log(T)^2$, we have
\begin{align*}
    \sqrt{2\left(\sum_{s=1}^t F_k(X_{s,k}) + 1\right)
    3 \log(T)}
    +
    \frac{3 \sqrt{2}}{3}
    \log(T)
&\le 
    \left( 3
    +
    \frac{3 \sqrt{2}}{3}
    \right)
    \sqrt{\sum_{s=1}^t F_k(X_{s,k})} 
    \log(T)
\\
&\le 
    5 \log(T)
    \sqrt{\sum_{s=1}^t F_k(X_{s,k})}.
\end{align*}
Hence,
\[
    \Pb
    \left[ 
    \sum_{s=1}^t F_k(X_{s,k})  \ge 100 \log(T)^2
    ~ \mbox{and} ~
    \sum_{s=1}^t F_k(X_{s,k})
    -
    \sum_{s=1}^t
    B_{s,k}(X_{s,k})
>
    5 \log(T)
    \sqrt{\sum_{s=1}^t F_k(X_{s,k})}
    \right]
\le 
    \frac{1}{T^2}.
\]
Then, if
\[
    \sum_{s=1}^t F_k(X_{s,k})
\ge 
    100 \log(T)^2,
\]
we have 
\[
    \frac{
    \sum_{s=1}^t F_k(X_{s,k})
    }{2}
\ge 
    5 \log(T)
    \sqrt{\sum_{s=1}^t F_k(X_{s,k})}\;.
\]
Hence we have
\[
    \Pb
    \left[ 
    \sum_{s=1}^t F_k(X_{s,k})
\ge 
    100 \log(T)^2
    ~
    ~
    \mbox{and}
    ~
    ~
    \sum_{s=1}^t
    B_{s,k}(X_{s,k}) 
    <
    \frac{\sum_{s=1}^t F_k(X_{s,k})}{2}
    \right]
\le 
    \frac{1}{T^2}.
\]
We conclude the proof with a union bound. 
\end{proof}

\section{Polylogarithmic regret with diminishing returns: missing details}
\label{s:appe-speed-ups}

We now introduce two key assumptions that greatly generalize \Cref{ass:speed-up} (in \Cref{s:speed-ups}), but will still allow us to prove a polylogarithmic regret in the diminishing-returns setting.
The first one is a mild regularity assumption on $F_1,\dots, F_K$.
The second one is a concavity assumption weaker than the classic strong concavity assumption. 

\begin{assumption} \label{assumption:F}
$ $
\begin{enumerate}
\item For each $k \in [K]$, $F_k$ is continuously differentiable on $(0,1]$ with $\displaystyle \lim_{u \to 0^+} F_k'(u) = + \infty$. 
\item There is a function $\alpha: (0,1) \to (0 , \infty)$ such that for each $k \in [K]$ and for each $\epsilon \in (0,1)$, $F_k$ is $\alpha(\epsilon)$-strongly concave on $[\epsilon,1]$.
\end{enumerate}
\end{assumption}

Note that Assumption \ref{assumption:F} is strictly weaker than \Cref{ass:speed-up}, since it is implied by but it does not imply \Cref{ass:speed-up}.
Note also that, under \Cref{assumption:F}:
\begin{itemize}[wide]
    \item There exists a function $M \colon (0,1) \to (0 , \infty)$ such that for each $k \in [K]$ and for each $\epsilon \in (0,1)$, $F_k$ is $M(\epsilon)$-{\lip} on $[\epsilon,1]$. 
    \item There exists a function $D_{\min}\colon (0,1] \to [0,\iop)$ defined, for all $\epsilon \in (0,1]$, by
    $
    D_{\min}(\epsilon) \ceq
    \min_{k \in [K]}
    \inf_{u \in (0,\epsilon] }
    F_k'(u)
    $.
    \item
    There exists a function $D_{\max}\colon (0,1] \to [0,\iop)$ defined, for all $\epsilon \in (0,1]$, by
    $
    D_{\max}(\epsilon) \ceq
    \max_{k \in [K]}
    \sup_{u \in [\epsilon,1] }
    F_k'(u).
    $
    \item 
    There exists a function $D_{\min}^{-1} \colon \bigl[D_{\min}(1) ,  \infty\bigr) \to (0,1]$, defined, for all $y \in \bigl[D_{\min}(1) ,  \infty\bigr)$, by
    $
    D_{\min}^{-1}(y)
    \ceq \inf \bcb{
    \epsilon \in (0,1] :
    D_{\min}(\epsilon)
    \le y
    }
    $
    ---indeed, the set in the $\inf$ before is non-empty (it contains $1$) and 
    $D_{\min}^{-1}(y) \in (0,1]$ because $\displaystyle \lim_{\epsilon\to 0^+}D_{\min}(\epsilon) = \infty $.  
\end{itemize}

For the remainder of this section, we assume fixed such functions $M$, $D_{\min}$, $D_{\max}$, and $D_{\min}^{-1}$.

We now introduce some notation that will be used to state and prove the next theorem, that is our main result of this section.
As in \Cref{s:speed-ups},
for each $m \in [0,\infty)^K$, define $\Phi_m \colon [0,1]^K \to [0,\infty)$, $x \mapsto \sum_{k \in [K]} m_k F_k\brb{x_k}$. 
Also as in \Cref{s:speed-ups},
we use the notation $x^\star (m)$ to denote a maximizer (chosen arbitrarily if there are more than one of them) of $\Phi_m$ over $\Delta_K$ and we
remark that, for any $t \ge 1$, \Cref{algo:ucbowski} selects allocations
\begin{equation*} 
(X_{t,1},\ldots,X_{t,K})
= 
x^\star 
( \UCB_{t-1}^\delta ),
\end{equation*} 
where $\UCB_{t-1}^\delta$ is defined in the algorithm. 
We can now state our generalized version of the polylogarithmic rate of \Cref{algo:ucbowski} in the general diminishing-returns setting. 

\begin{theorem} \label{thm:boost}
Under \Cref{assumption:F}, for any time horizon $T \ge 3$ and mean vector
$\mu = (\mu_1,\ldots,\mu_K) \in (0,1]^K$, define
$
\mu_{\inf} \ceq \min_{k \in [K]} \mu_k
$,
\[
u_{\inf}
\ceq
D_{\min}^{-1}
\lrb{
\frac{2}{\mu_{\inf}}
D_{\max}\lrb{\frac{1}{K}}
}
\]
and the realized regret
\begin{equation} 
\label{eq:idealized:regret}
\cR_T
\coloneqq
\sum_{t=1}^T
\sum_{k=1}^K
\bbrb{
\mu_k F_k\lrb{\brb{x^\star(\mu)}_k}
-
\mu_k F_k(X_{t,k})
}.
\end{equation}
If we run \Cref{algo:ucbowski} with parameters $T$ and
$\delta \in \bigl[2e^{-50(\log T)^2},\frac{1}{2}\bigr)$, then its realized regret satisfies,
with probability at least $1-\frac{K}{T}-KT\delta$,
\begin{equation} \label{eq:realized:reg:polylog}
\cR_T
\le
\frac{
64 M(u_{\inf})^2 \log\lrb{\frac{2}{\delta}} K
}{
\mu_{\inf}\alpha(u_{\inf}) \min_{k \in [K]} F_k(u_{\inf})
}
\log T
+
3K^2
+
\frac{
4K^2 \lrb{100(\log T)^2+1}
}{
\min_{k \in [K]} F_k
\lrb{
D_{\min}^{-1}
\lrb{
\frac{2}{\mu_{\inf}}
D_{\max}\lrb{\frac{1}{K^2}}
}
}
}.
\end{equation}
Consequently, if we run \Cref{algo:ucbowski} with parameters $T$ and
$\delta=(KT)^{-2}$, then the expected regret satisfies
\begin{equation} \label{eq:final:regret:fast}
R_T
\le
\frac{
64 M(u_{\inf})^2 \log\lrb{2K^2T^2} K
}{
\mu_{\inf}\alpha(u_{\inf}) \min_{k \in [K]} F_k(u_{\inf})
}
\log T
+
5K^2
+
\frac{
4K^2 \lrb{100(\log T)^2+1}
}{
\min_{k \in [K]} F_k
\lrb{
D_{\min}^{-1}
\lrb{
\frac{2}{\mu_{\inf}}
D_{\max}\lrb{\frac{1}{K^2}}
}
}
}.
\end{equation}
\end{theorem}

Note that in \eqref{eq:realized:reg:polylog}, taking
$\delta=T^{-2}$ gives $\log\lrb{2/\delta}=\log\lrb{2T^2}=O(\log T)$ and the
good event has probability at least
\[
1-\frac{K}{T}-KT\delta
=
1-\frac{2K}{T}.
\]
Thus, for fixed problem parameters, \eqref{eq:realized:reg:polylog} yields a
$\mathrm{polylog}(T)$ realized-regret bound with probability at least
$1-\frac{2K}{T}$.

To prove \Cref{thm:boost}, we will need three technical lemmas, which we state and prove first.
The first lemma is a first step toward obtaining quantitative lower bounds on the optimal budget allocation as well as the budget allocations selected by \Cref{algo:ucbowski}.

\begin{lemma} \label{lemma:Lb:budget}
Let \Cref{assumption:F} hold.
    For any $m \in (0,+\infty)^K$, for any $k, \ell \in [K]$, 
    $k \neq \ell$,
    we have
    \[
\brb{x^\star(m)}_\ell 
\ge D_{\min}^{-1}
\left( 
\frac{ m_k }{m_\ell}
D_{\max} \Brb{\brb{x^\star(m)}_k }
\right).
    \]
\end{lemma}

\begin{proof}
Since $\displaystyle \lim_{u \to 0^+} F_k'(u) = + \infty$ and the components of $m$ are non-zero, we can see that all the components of $x^\star(m)$ are non-zero.
Hence decreasing $\brb{x^\star(m)}_k $ by a quantity
$\gamma >0$, increasing $\brb{x^\star(m)}_\ell $ by the same $\gamma >0$,
and taking the limit for $\gamma\to0^+$, 
leads to a modification of $\sum_{r=1}^K 
m_r F_r\Brb{\brb{x^\star(m)}_r}$ satisfying
\begin{align} \label{eq:after:before}
0 \ge 
F_\ell'\Brb{ \brb{x^\star(m)}_\ell } 
m_\ell 
-   
F_k' \Brb{\brb{x^\star(m)}_k } 
m_k.
\end{align}
Hence
\[
F_\ell'\Brb{ \brb{x^\star(m)}_\ell} 
\le 
\frac{m_k}{m_\ell }
F_k'\Brb{ \brb{x^\star(m)}_k }\;,
\]
and thus
\[
D_{\min} \Brb{
 \brb{x^\star(m)}_\ell 
}
\le 
\frac{m_k}{m_\ell }
D_{\max} \Brb{\brb{x^\star(m)}_k  }.
\]
This yields the lemma.
\end{proof}

Let us use the same notation $ N_{T,k}$ and
$\tilde{G}_{\tau,k}$
as in \eqref{eq:intro:N:Gtilde}
in \Cref{s:upper-bound}. 
For $\delta \in (0,1/2]$,
consider a first ``good event'',
as the event in \eqref{eq:good:event:ET-main} in \Cref{s:upper-bound},
\begin{equation} 
\label{eq:good:event:ET}
\cE_{T}^\delta
\ceq 
\underset{k \in [K]}{\bigcap}
\underset{t \in [ N_{T,k}]}{\bigcap}
\left\{ 
\left| 
\frac{1}{t}
\sum_{s=1}^t 
\tilde{G}_{s,k} 
-
\mu_k
\right|
\le 
\sqrt{
\frac{\log\lrb{\frac{2}{\delta}}}{2t}
}
\right\}.
\end{equation}

Then, define a second good event $\cG_{T}$ as the event of Lemma \ref{lemma:LB:bernoullis} (that holds with probability at least $1-\frac{K}{T}$).
The next lemma shows that each arm is explored at least a linear function of $t$, when $t$ is larger than some logarithmic function of $T$, under these good events. 

\begin{lemma}
\label{lemma:lb:cum:budget}
Let \Cref{assumption:F} hold.
Assume $\mu_{\inf}\coloneqq\min_{k\in[K]}\mu_k>0$.
Let $T \ge 3$.
    Consider $t\in[T]$ with 
   \begin{equation}
   \label{eq:log:cond:t}
       t \ge 
       \frac{
       400 K (\log T)^2
       }{\min_{k \in [K]} F_k\lrb{\frac{1}{K^2}} }.
   \end{equation}
    Assume that
    the two good events
     $\cE_{T}^\delta$ in \eqref{eq:good:event:ET}
     and
     $\cG_{T}$ 
    hold. For $\cE_{T}^\delta$ assume that $\delta \ge 2e^{-50(\log T)^{2}}$.
Then for each arm $\ell \in [K]$ it holds
\begin{equation}
\label{eq:lower:bound:budget}
\sum_{s=1}^t 
F_\ell(X_{s,\ell} )
\ge 
\left(
\frac{t}{4K} 
- \frac{1}{2}
\right)
F_\ell 
\left( 
D_{\min}^{-1}
\left( 
\frac{2}{\mu_{\inf}}
D_{\max} \left( \frac{1}{K^2} \right)
\right)
\right).
\end{equation} 
In addition, if we also have
\begin{equation}
\label{eq:cond:t:instant:budget}
t \ge 
2K +
\frac{4K \left(100 (\log T)^2 + 1 \right)}{ \min_{k \in [K]} F_k 
\left( 
D_{\min}^{-1}
\left( 
\frac{2}{\mu_{\inf}}
D_{\max} \left( \frac{1}{K^2} \right)
\right)
\right) }
\end{equation}
it also holds that, for each $\ell \in [K]$,
\begin{equation}
\label{eq:lower:bound:budget:instant}   
X_{t,\ell} 
\ge 
D_{\min}^{-1}
\left( 
\frac{2}{\mu_{\inf}}
D_{\max} \left( \frac{1}{K} \right)
\right).
\end{equation}
\end{lemma}

\begin{proof}
For the time $t$ under consideration, we consider $k$ such that arm $k$ received the maximum budget $\sum_{s=1}^t X_{s,k}$. We also 
consider any other arm $\ell$ than $k$. 
We let 
\[
S_t
=
\# \lcb{ s \in [t] : X_{s,k} \ge \frac{1}{K^2}  }. 
\]
Then $ S_t + \frac{t}{K^2} \ge \sum_{s=1}^t X_{s,k} \ge \frac{t}{K}$ leads to $ S_t \ge \frac{t}{K} - \frac{t}{K^2}$ and thus
\[
S_t \ge \frac{t}{2K}.
\]

In the list $s_1,\ldots,s_{S_t}$ of times $s \in [t]$ for which $ X_{s,k} \ge \frac{1}{K^2}$, take an element $s$ in the second half of it (composed of $r$ elements if $S_t=2r+1$ is odd).

Using that $\cE_{T}^\delta$ in \eqref{eq:good:event:ET} holds we obtain that $\UCB_{s-1,\ell}^\delta \ge \mu_\ell $. Also, 
\begin{equation}
\label{eq:budget:k:lb}
\sum_{s'=1}^{s-1} 
F_k(X_{s',k})
\ge 
\frac{S_t}{2} 
F_k\left( \frac{1}{K^2} \right)
\ge
\frac{t}{4K} 
F_k\left( \frac{1}{K^2} \right).
\end{equation}
From \eqref{eq:log:cond:t}, we thus have
\begin{equation} \label{eq:sum:X:ge}
\sum_{s'=1}^{s-1} 
F_k(X_{s',k})  
\ge 100 (\log T)^2. 
\end{equation}
Hence using that the event of Lemma \ref{lemma:LB:bernoullis} holds,
\[
\sum_{s'=1}^{s-1}
B_{s',k}(X_{s',k}) 
\ge 
\frac{\sum_{s'=1}^{s-1} F_k(X_{s',k})}{2} 
\ge 
50 (\log T)^2.
\]
Let us recall the notation 
$
    N_{t,k}
=
    \sum_{s=1}^t B_{s,k}(X_{s,k})
$
and $\tilde{G}_{\tau,k}$
as in \eqref{eq:intro:N:Gtilde}. 
This yields, with the same understanding that $0/0\coloneqq 0$,
\begin{equation*}
\UCB_{s-1,k}^\delta 
=
\frac{ \sum_{s'=1}^{N_{s-1,k}}
\tilde{G}_{s',k}
}{N_{s-1,k}}
+
\sqrt{
\frac{\log(\frac{2}{\delta})}{1+N_{s-1,k}}}
\le 
1 + \sqrt{
\frac{\log(\frac{2}{\delta})}{50 (\log T)^2}}.
\end{equation*}
From $\delta \ge 2e^{-50(\log T)^{2}}$,
we obtain, 
\[
\UCB_{s-1,k}^\delta 
\le 2. 
\]
Hence applying \Cref{lemma:Lb:budget} to $m = \UCB_{s-1}^{\delta}$, with $\UCB_{s-1,\ell}^\delta \ge \mu_\ell $ and $\UCB_{s-1,k}^\delta 
\le 2$, we get
\begin{equation}
\label{eq:lower:bound:regret:instant:s}    
X_{s,\ell} 
\ge 
D_{\min}^{-1}
\left( 
\frac{2}{\mu_{\inf}}
D_{\max} \left( \frac{1}{K^2} \right)
\right).
\end{equation}
Since \eqref{eq:lower:bound:regret:instant:s} is true for all $s$ in the second half of the list described above, we obtain 
\begin{align*} \label{eq:linear:budget}
\sum_{s=1}^t 
F_\ell(X_{s,\ell} )
& \ge 
\frac{S_t-1}{2} 
F_\ell 
\left( 
D_{\min}^{-1}
\left( 
\frac{2}{\mu_{\inf}}
D_{\max} \left( \frac{1}{K^2} \right)
\right)
\right) \notag \\
& \ge
\left(
\frac{t}{4K} 
- \frac{1}{2}
\right)
F_\ell 
\left( 
D_{\min}^{-1}
\left( 
\frac{2}{\mu_{\inf}}
D_{\max} \left( \frac{1}{K^2} \right)
\right)
\right).
\end{align*}

This yields \eqref{eq:lower:bound:budget} for all $\ell \ne k$. From \eqref{eq:budget:k:lb}, one can see that \eqref{eq:lower:bound:budget} also holds for $\ell = k$.

Let us now show \eqref{eq:lower:bound:budget:instant}, using Lemma \ref{lemma:Lb:budget}. Consider any $\ell \in [K]$ and $t$ such that \eqref{eq:cond:t:instant:budget} holds. There is $\ell' \in [K]$ such that $X_{t,\ell'} \ge \frac{1}{K}$. We have
\[
\UCB_{t-1,\ell'}^{\delta} 
=
\frac{ \sum_{s=1}^{N_{t-1,\ell'}}
\tilde{G}_{s,\ell'}
}{N_{t-1,\ell'}}
+
\sqrt{
\frac{\log(\frac{2}{\delta})}{1+N_{t-1,\ell'}}}
\le 
1 +\sqrt{
\frac{\log(\frac{2}{\delta})}{1+N_{t-1,\ell'}}}.
\]
From \eqref{eq:lower:bound:budget} and then \eqref{eq:cond:t:instant:budget}, we have
\begin{equation} \label{eq:lb:cum:budg:t}
\sum_{s=1}^{t-1}
F_{\ell'} (X_{s,\ell'})
\ge 
\left( 
\frac{t}{4K} 
- \frac{1}{2}
\right)
F_{\ell'} 
\left( 
D_{\min}^{-1}
\left( 
\frac{2}{\mu_{\inf}}
D_{\max} \left( \frac{1}{K^2} \right)
\right)
\right)
-1
\ge 100 (\log T)^2.
\end{equation}
Hence, using that the event of Lemma \ref{lemma:LB:bernoullis} holds,  
\[
\UCB_{t-1,\ell'}^{\delta} 
\le 1 + 
\sqrt{
\frac{2 \log(\frac{2}{\delta})}{\sum_{s=1}^{t-1} F_{\ell'}(X_{s,\ell'})}} 
\le 2,
\]
using
$\delta \ge 2e^{-50(\log T)^{2}}$.
Finally, under the good event $\cE_{T}^\delta$ in \eqref{eq:good:event:ET}, we have $\UCB_{t-1,\ell}^{\delta} \ge \mu_\ell$. Hence from Lemma \ref{lemma:Lb:budget}, \eqref{eq:lower:bound:budget:instant} holds.
\end{proof}

The last crucial lemma establishes the stability of the average reward under reward mean $m$, when the budget allocation goes from the optimal $x^\star(m)$ to the suboptimal $x^\star(\tilde{m})$, computed from a presumed $\tilde{m} \neq m$.

\begin{lemma}
\label{lemma:stability}
Let \Cref{assumption:F} hold.
Let $u_{\inf} \in (0,\frac{1}{K}]$ and $0 < m_{\inf} < m_{\sup} < \infty$.
Consider mean vectors $m,\tilde{m} \in [m_{\inf},m_{\sup}]^K$.
Assume that $x^\star(m) \in [u_{\inf},1]^K$ and $x^\star(\tilde{m}) \in [u_{\inf},1]^K$.
Then
\[
\Phi_m\brb{x^\star(m)} - \Phi_m\brb{x^\star(\tilde{m})}
\le
\frac{2 M(u_{\inf})^2}{m_{\inf}\alpha(u_{\inf})}\lno{m-\tilde{m}}_2^2.
\]
\end{lemma}

\begin{proof}
For conciseness, write $\alpha = \alpha(u_{\inf})$ and $M = M(u_{\inf})$.
Since $m \in [m_{\inf},m_{\sup}]^K$, the function $\Phi_m$ is $m_{\inf}\alpha$-strongly concave on $[u_{\inf},1]^K$ with respect to the Euclidean norm.
Indeed, for each $u,v \in [u_{\inf},1]^K$,
\begin{align*}
    \Phi_m(v)
&=
    \sum_{k \in [K]} m_k F_k\brb{v_k}
\\
&\le
    \sum_{k \in [K]} m_k
    \Brb{
    F_k\brb{u_k}
    +
    F_k'\brb{u_k}\brb{v_k-u_k}
    -
    \frac{\alpha}{2}\brb{v_k-u_k}^2
    }
\\
&\le
    \Phi_m(u)
    +
    \sum_{k \in [K]} m_k F_k'\brb{u_k}\brb{v_k-u_k}
    -
    \frac{m_{\inf}\alpha}{2}\lno{v-u}_2^2.
\end{align*}
Similarly, $\Phi_{\tilde{m}}$ is $m_{\inf}\alpha$-strongly concave on $[u_{\inf},1]^K$ with respect to the Euclidean norm.

For conciseness, write $\Phi^\star = \Phi_m$ and $\tilde{\Phi} = \Phi_{\tilde{m}}$.
Write also $x^\star = x^\star(m)$ and $\tilde{u}=x^\star(\tilde{m})$.
For $v \in [0,\frac{1}{K}]$, write
\[
\Delta_K(v)
=
\Delta_K \cap [v,1]^K.
\]
Since $x^\star(m) \in [u_{\inf},1]^K$ and $x^\star(\tilde{m}) \in [u_{\inf},1]^K$, we have
\[
x^\star \in \argmax_{u \in \Delta_K(u_{\inf})} \Phi^\star(u)
\qquad\text{and}\qquad
\tilde{u} \in \argmax_{u \in \Delta_K(u_{\inf})} \tilde{\Phi}(u).
\]

Since $u_{\inf}>0$, the point $x^\star$ is in the relative interior of $\Delta_K$.
Moreover, $x^\star$ maximizes $\Phi^\star$ on $\Delta_K$.
Therefore, for each $u \in \Delta_K$,
\begin{equation}
\label{eq:zero:proj:grad}
\big\langle\nabla\Phi^\star(x^\star),u-x^\star\big\rangle = 0.
\end{equation}

Leveraging first the $m_{\inf}\alpha$-strong concavity of $\Phi^\star$, and then the $M$-{\lip}ness of $F_1,\dots,F_K$ on $[u_{\inf},1]$, we have
\begin{align}
&
    \frac{m_{\inf}\alpha}{2}\lno{\tilde{u}-x^\star}_2^2
    \label{eq:first}
\\
&
\qquad
\le
    \Phi^\star(x^\star)
    +
    \underbrace{\big\langle\nabla\Phi^\star(x^\star),\tilde{u}-x^\star\big\rangle}_{=0~\text{by \eqref{eq:zero:proj:grad}}}
    -
    \Phi^\star(\tilde{u})
    \nonumber
\\
&
\qquad
=
    \Phi^\star(x^\star)-\Phi^\star(\tilde{u})
    \label{eq:second}
\\
&
\qquad
=
    \Phi^\star(x^\star)-\tilde{\Phi}(x^\star)
    +
    \underbrace{\tilde{\Phi}(x^\star)-\tilde{\Phi}(\tilde{u})}_{\le 0}
    +
    \tilde{\Phi}(\tilde{u})-\Phi^\star(\tilde{u})
    \nonumber
\\
&
\qquad
\le
    \Phi^\star(x^\star)-\tilde{\Phi}(x^\star)
    +
    \tilde{\Phi}(\tilde{u})-\Phi^\star(\tilde{u})
    \nonumber
\\
&
\qquad
=
    \sum_{k \in [K]}
    \brb{m_k-\tilde{m}_k}
    \brb{F_k(x^\star_k)-F_k(\tilde{u}_k)}
    \nonumber
\\
&
\qquad
\le
    M\lno{m-\tilde{m}}_2\lno{x^\star-\tilde{u}}_2.
    \label{eq:third}
\end{align}
From \eqref{eq:first}--\eqref{eq:third}, we get
\begin{equation}
\label{eq:fourth}
\lno{x^\star-\tilde{u}}_2
\le
\frac{2M}{m_{\inf}\alpha}\lno{m-\tilde{m}}_2.
\end{equation}
Putting \eqref{eq:second}--\eqref{eq:third} together with \eqref{eq:fourth}, we obtain
\[
\Phi^\star(x^\star)-\Phi^\star(\tilde{u})
\le
M\lno{m-\tilde{m}}_2\lno{x^\star-\tilde{u}}_2
\le
\frac{2M^2}{m_{\inf}\alpha}\lno{m-\tilde{m}}_2^2.
\]
This concludes the proof.
\end{proof}

We now have all the tools to prove our polylogarithmic regret bound under concavity, i.e., to complete the proof of Theorem \ref{thm:boost}.

\begin{proof}[Proof of Theorem \ref{thm:boost}]
Let $\delta \in \bigl[2e^{-50(\log T)^2},\frac{1}{2}\bigr)$ be the confidence parameter used by \Cref{algo:ucbowski}.
Thus the algorithm selects $X_t=x^\star\brb{\UCB_{t-1}^\delta}$ at every time $t$.
The aim is to use \Cref{lemma:stability} with
\[
m_{\inf} = \mu_{\inf},
\qquad
m_{\sup} = 2,
\qquad
u_{\inf}
=
D_{\min}^{-1}
\lrb{
\frac{2}{\mu_{\inf}}
D_{\max}\lrb{\frac{1}{K}}
}.
\]
For conciseness, write $\alpha = \alpha(u_{\inf})$ and $M = M(u_{\inf})$.

By \Cref{cor:Hoeffding:gains} and the union bound over $k\in[K]$ and over at most $T$ possible positive values of the completion counter,
\[
\Pb\bsb{\cE_T^\delta}\ge 1-KT\delta.
\]
Moreover, by \Cref{lemma:LB:bernoullis}, $\Pb\bsb{\cG_T}\ge 1-\frac{K}{T}$.
Hence
\[
\Pb\bsb{\cE_T^\delta\cap \cG_T}
\ge
1-\frac{K}{T}-KT\delta.
\]
We work on the event $\cE_T^\delta\cap \cG_T$ for the rest of the proof.

Define
\[
v_0
\ceq
D_{\min}^{-1}
\lrb{
\frac{2}{\mu_{\inf}}
D_{\max}\lrb{\frac{1}{K^2}}
}
\]
and
\[
S(T)
\coloneqq
2K
+
\lce{
\frac{4K \lrb{100(\log T)^2+1}}
{
\min_{k\in[K]} F_k\lrb{v_0}
}
}.
\]
This is the quantity appearing in \eqref{eq:cond:t:instant:budget}.
We first note that the definition of $S(T)$ implies the conditions needed to apply \Cref{lemma:lb:cum:budget}.
Indeed, by the monotonicity properties of $D_{\min}^{-1}$ and since $D_{\min}\lrb{\frac{1}{K^2}}\le D_{\max}\lrb{\frac{1}{K^2}}$, we have
\[
v_0
\le
\frac{1}{K^2}.
\]
Since the $F_k$ are non-decreasing, this gives
\[
\min_{k\in[K]}F_k(v_0)
\le
\min_{k\in[K]}F_k\lrb{\frac{1}{K^2}}.
\]
Therefore, $t\ge S(T)$ implies both \eqref{eq:log:cond:t} and \eqref{eq:cond:t:instant:budget}.
By \eqref{eq:lower:bound:budget:instant} in \Cref{lemma:lb:cum:budget}, for every $t\ge S(T)$ and every $k\in[K]$,
\[
X_{t,k}
=
\brb{x^\star\brb{\UCB_{t-1}^\delta}}_k
\ge
u_{\inf}.
\]

We next check the remaining hypotheses of \Cref{lemma:stability}.
First, $\mu\in[\mu_{\inf},2]^K$.
Second, on $\cE_T^\delta$, for every $t\in[T]$ and every $k\in[K]$,
\[
\UCB_{t,k}^\delta\ge \mu_k\ge \mu_{\inf}.
\]
Moreover, for every $t\ge S(T)$ and every $k\in[K]$, the argument leading to \eqref{eq:lb:cum:budg:t} in the proof of \Cref{lemma:lb:cum:budget} gives
\[
\sum_{s=1}^{t-1}F_k(X_{s,k})
\ge
100(\log T)^2.
\]
Since $\cG_T$ holds, this implies
\[
N_{t-1,k}
=
\sum_{s=1}^{t-1}B_{s,k}(X_{s,k})
\ge
\frac{1}{2}
\sum_{s=1}^{t-1}F_k(X_{s,k})
\ge
50(\log T)^2.
\]
Therefore, using also $\delta\ge 2e^{-50(\log T)^2}$,
\[
\UCB_{t-1,k}^\delta
\le
1+
\sqrt{
\frac{\log\lrb{\frac{2}{\delta}}}{1+N_{t-1,k}}
}
\le
1+
\sqrt{
\frac{\log\lrb{\frac{2}{\delta}}}{50(\log T)^2}
}
\le
2.
\]
Thus, for every $t\ge S(T)$,
\[
\UCB_{t-1}^\delta\in[\mu_{\inf},2]^K.
\]

It remains to lower bound the optimal allocation $x^\star(\mu)$.
Choose $j\in[K]$ such that $\brb{x^\star(\mu)}_j\ge \frac{1}{K}$.
For any $\ell\in[K]$, \Cref{lemma:Lb:budget}, applied with $m=\mu$, gives
\[
\brb{x^\star(\mu)}_\ell
\ge
D_{\min}^{-1}
\lrb{
\frac{\mu_j}{\mu_\ell}
D_{\max}\lrb{\brb{x^\star(\mu)}_j}
}.
\]
Since $\mu_j/\mu_\ell\le 1/\mu_{\inf}\le 2/\mu_{\inf}$ and $\brb{x^\star(\mu)}_j\ge 1/K$, the monotonicity properties of $D_{\max}$ and $D_{\min}^{-1}$ yield
\[
\brb{x^\star(\mu)}_\ell
\ge
D_{\min}^{-1}
\lrb{
\frac{2}{\mu_{\inf}}
D_{\max}\lrb{\frac{1}{K}}
}
=
u_{\inf}.
\]
Hence both $x^\star(\mu)$ and $x^\star\brb{\UCB_{t-1}^\delta}$ have all coordinates at least $u_{\inf}$ for every $t\ge S(T)$, and the hypotheses of \Cref{lemma:stability} are satisfied.

If $S(T)>T$, then the deterministic bound $\cR_T\le KT\le KS(T)$ already gives the desired high-probability bound.
We may therefore assume $S(T)\le T$.
The regret incurred before time $S(T)$ is at most $KS(T)$.
Therefore,
\[
\cR_T
\le
K S(T)
+
\sum_{t=S(T)}^T
\Brb{
\Phi_\mu\brb{x^\star(\mu)}
-
\Phi_\mu\brb{x^\star\brb{\UCB_{t-1}^\delta}}
}.
\]
Applying \Cref{lemma:stability} with $m=\mu$ and $\tilde m=\UCB_{t-1}^\delta$ gives
\begin{align*}
\cR_T
&\le
K S(T)
+
\sum_{t=S(T)}^T
\frac{2M^2}{\mu_{\inf}\alpha}
\lno{\mu-\UCB_{t-1}^\delta}_2^2
\\
&\le
K S(T)
+
\sum_{t=S(T)}^T
\frac{4M^2}{\mu_{\inf}\alpha}
\sum_{k=1}^K
\frac{4\log\lrb{\frac{2}{\delta}}}
{1+\sum_{s=1}^{t-1}B_{s,k}(X_{s,k})}.
\end{align*}
The last inequality uses $\cE_T^\delta$, which implies
\[
0
\le
\UCB_{t-1,k}^\delta-\mu_k
\le
2
\sqrt{
\frac{\log\lrb{\frac{2}{\delta}}}
{1+N_{t-1,k}}
}
\]
for every $t$ and $k$.

Using $\cG_T$ and the lower bound $\sum_{s=1}^{t-1}F_k(X_{s,k})\ge 100(\log T)^2$ from above, we get
\[
\sum_{s=1}^{t-1}B_{s,k}(X_{s,k})
\ge
\frac{1}{2}
\sum_{s=1}^{t-1}F_k(X_{s,k}).
\]
Moreover, for every $t\ge S(T)$ and every $k\in[K]$,
\[
\sum_{s=1}^{t-1}F_k(X_{s,k})
\ge
\sum_{s=1}^{S(T)-1}F_k(X_{s,k})
+
\sum_{s=S(T)}^{t-1}F_k(X_{s,k})
\ge
100(\log T)^2
+
\brb{t-S(T)}F_k(u_{\inf}).
\]
Since $F_k(u_{\inf})\le 1$ and $T\ge3$, this implies
\[
\sum_{s=1}^{t-1}F_k(X_{s,k})
\ge
\brb{t-S(T)+1}F_k(u_{\inf}).
\]
Hence
\begin{align*}
\cR_T
&\le
K S(T)
+
\frac{16M^2\log\lrb{\frac{2}{\delta}}}{\mu_{\inf}\alpha}
\sum_{t=S(T)}^T
\sum_{k=1}^K
\frac{2}
{\sum_{s=1}^{t-1}F_k(X_{s,k})}
\\
&\le
K S(T)
+
\frac{32M^2\log\lrb{\frac{2}{\delta}}K}
{\mu_{\inf}\alpha\min_{k\in[K]}F_k(u_{\inf})}
\sum_{r=1}^{T-S(T)+1}
\frac{1}{r}
\\
&\le
K S(T)
+
\frac{64M^2\log\lrb{\frac{2}{\delta}}K}
{\mu_{\inf}\alpha\min_{k\in[K]}F_k(u_{\inf})}
\log T.
\end{align*}
Substituting the definition of $S(T)$ yields
\begin{align*}
\cR_T
&\le
\frac{64M(u_{\inf})^2\log\lrb{\frac{2}{\delta}}K}
{\mu_{\inf}\alpha(u_{\inf})\min_{k\in[K]}F_k(u_{\inf})}
\log T
+
2K^2
+
K
\lce{
\frac{4K\lrb{100(\log T)^2+1}}
{
\min_{k\in[K]}F_k
\lrb{
D_{\min}^{-1}
\lrb{
\frac{2}{\mu_{\inf}}
D_{\max}\lrb{\frac{1}{K^2}}
}
}
}
}
\\
&\le
\frac{64M(u_{\inf})^2\log\lrb{\frac{2}{\delta}}K}
{\mu_{\inf}\alpha(u_{\inf})\min_{k\in[K]}F_k(u_{\inf})}
\log T
+
3K^2
+
\frac{4K^2\lrb{100(\log T)^2+1}}
{
\min_{k\in[K]}F_k
\lrb{
D_{\min}^{-1}
\lrb{
\frac{2}{\mu_{\inf}}
D_{\max}\lrb{\frac{1}{K^2}}
}
}
}.
\end{align*}
This proves the high-probability bound on an event of probability at least $1-\frac{K}{T}-KT\delta$ as stated in \eqref{eq:realized:reg:polylog}.

We now prove \eqref{eq:final:regret:fast}.

Set $\delta=(KT)^{-2}$.
If $\delta\ge 2e^{-50(\log T)^2}$, the previous bound applies.
Since the reward function is $[0,K]$-valued, $\cR_T\le KT$ always, and therefore
\[
R_T
=
\E\bsb{\cR_T}
\le
\mathrm{RHS}(\delta)
+
KT\lrb{\frac{K}{T}+KT\delta}
\le
\mathrm{RHS}(\delta)
+
2K^2,
\]
where $\mathrm{RHS}(\delta)$ denotes the right-hand side of
\eqref{eq:realized:reg:polylog}.
This gives \eqref{eq:final:regret:fast}.

If $\delta<2e^{-50(\log T)^2}$, then
\[
K
>
\frac{e^{25(\log T)^2}}{\sqrt{2}T}.
\]
For $T\ge3$, the right-hand side is at least $T/5$, and hence $KT\le 5K^2$.
Since always $R_T\le KT$, the bound \eqref{eq:final:regret:fast} follows trivially in this case.
Thus \eqref{eq:final:regret:fast} holds in all cases.
\end{proof}

In the next proposition, we make explicit the upper bound on $R_T$ of \Cref{thm:boost}, for the special case of concave power functions $F_1,\dots,F_K$, as in \Cref{t:upper-bound-speed-up}.

\begin{proposition} 
\label{prop:explicit:constant}
Consider the case where there are $a_1,\dots,a_K \in (0,1)$ such that, for all $k \in [K]$ and $x \in [0,1]$, $F_k(x) = x^{a_k}$.
Define $a_{\min} \ceq \min_{k \in [K]} a_k$ and $a_{\max} \ceq \max_{k \in [K]} a_k$.
Define
$
    \mu_{\inf} \ceq \min_{k \in [K]} \mu_k
$
and assume $\mu_{\inf} >0$.
If we run \Cref{algo:ucbowski} with parameters $T \ge 3$ and confidence parameter $\delta = \frac{1}{(KT)^2}$, and if
\[
    \frac{1}{(KT)^2} \ge 2e^{-50(\log T)^2},
\]
then
\begin{align}
\label{eq:large:log:bound}
R_T \le & 
\frac{64
K^{1+\frac{(2(1-a_{\min})+a_{\max})(1-a_{\min})}{1-a_{\max}}}
}{
\left(
\frac{a_{\min} \mu_{\inf}}{
2 a_{\max}
}
\right)^{\frac{2(1-a_{\min})+a_{\max}}{1-a_{\max}}}
\mu_{\inf} \min_{k \in [K]} a_k(1-a_k ) } 
\log(2 K^2 T^2)  \log(T)
+ 5 K^2
\notag
 \\ 
 & + 
 \frac{
4 \left(100 (\log T)^2 + 1 \right)
K^{2+\frac{2 a_{\max} (1-a_{\min})}{1-a_{\max}}}
}{
  \left(
\frac{a_{\min} \mu_{\inf}}{
2 a_{\max}
}
\right)^{\frac{a_{\max}}{1-a_{\max}}}
}. 
\end{align}
In particular, for any $T\ge2$, if we run \Cref{algo:ucbowski} with parameters $T$ and confidence parameter $\delta = \frac{1}{(KT)^2}$, then
\begin{equation}
\label{eq:simpler:estimate}
    R_T
    \le
    800
    \left(
    \frac{
    K
    }{
    a_{\min}\mu_{\inf}(1-a_{\max})
    }
    \right)^{\frac{9}{1-a_{\max}}}
    (\log T)^2 .
\end{equation}
\end{proposition}

\begin{proof}
First, assume that $T\ge3$ and
\[
    \frac{1}{(KT)^2} \ge 2e^{-50(\log T)^2}.
\]
As already discussed, \Cref{assumption:F} holds, and we can apply \Cref{thm:boost} with $\delta=(KT)^{-2}$.
Let us first bound the individual quantities in \eqref{eq:final:regret:fast}.

Since $F_k'(x) = a_k x^{a_k-1}$, the function $F_k$ is $a_k \e^{a_k-1}$-Lipschitz on $[\e,1]$ and thus $\e^{a_k-1}$-Lipschitz on $[\e,1]$.
Hence, $M$ can be chosen such that, for $\e \in (0,1]$,
\[
    M(\e) \le \e^{a_{\min}-1}.
\]
Also, for $x \in (0,\e]$, $F_k'(x) \ge a_k \e^{a_k-1} \ge a_{\min} \e^{a_{\max} - 1}$.
Hence, for $\e \in (0,1]$,
\[
    D_{\min}(\e) \ge a_{\min} \e^{a_{\max} - 1}.
\]
Then, for $x \in [\e,1]$, $F_k'(x) \le a_k \e^{a_k-1} \le a_{\max} \e^{a_{\min} - 1}$.
Hence, for $\e \in (0,1]$,
\[
    D_{\max}(\e) \le a_{\max} \e^{a_{\min} - 1}.
\]
Next, for $\e,y > 0$ such that
\[
    \e < \left(\frac{a_{\min}}{y}\right)^{\frac{1}{1-a_{\max}}},
\]
we have $D_{\min}(\e) \ge a_{\min} \e^{a_{\max}-1} > y$.
Hence, for $y \in [D_{\min}(1), \infty)$,
\[
    D_{\min}^{-1}(y) 
    \ge \left( \frac{a_{\min}}{y} \right)^{\frac{1}{1-a_{\max}}}.
\]

Then,
\begin{align*}
u_{\inf}
&= 
D_{\min}^{-1}
\lrb{
\frac{2}{\mu_{\inf}}
D_{\max}\lrb{\frac{1}{K}}
}
\ge
\left(
\frac{a_{\min}}{
\frac{2}{\mu_{\inf}}
D_{\max}\lrb{\frac{1}{K}}
}
\right)^{\frac{1}{1-a_{\max}}} 
\\ 
& \ge 
\left(
\frac{a_{\min}}{
\frac{2}{\mu_{\inf}}
a_{\max} \lrb{\frac{1}{K}}^{a_{\min}-1}
}
\right)^{\frac{1}{1-a_{\max}}} 
=
\left(
\frac{a_{\min} \mu_{\inf}}{
2 a_{\max}
}
\right)^{\frac{1}{1-a_{\max}}}
K^{\frac{a_{\min}-1}{1-a_{\max}}}.
\end{align*}

Similarly,
\[
D_{\min}^{-1}
\lrb{
\frac{2}{\mu_{\inf}}
D_{\max}\lrb{\frac{1}{K^2}}
}
\ge 
\left(
\frac{a_{\min} \mu_{\inf}}{
2 a_{\max}
}
\right)^{\frac{1}{1-a_{\max}}}
K^{\frac{2(a_{\min}-1)}{1-a_{\max}}}.
\]
Also, for $\e \in (0,1]$,
\[
    \min_{k \in [K]} F_k(\e) 
    \ge 
    \e^{a_{\max}}.
\]
Finally, for $k \in [K]$, for $\e \in (0,1]$ and $x \in (\e,1]$,
\[
    \frac{\partial^2}{\partial x^2} x^{a_k} 
    =
    a_k (a_k-1) x^{a_k -2}
    \le
    a_k (a_k-1).
\]
Thus, $F_k$ is $a_k(1-a_k)$-strongly concave on $(\e,1]$.
Hence, we can choose $\alpha$ such that, for $\e \in (0,1]$,
\[
    \alpha(\e) 
    \ge 
    \min_{k \in [K]}
    a_k (1-a_k).
\]

Now, let us consider the regret upper bound in \eqref{eq:final:regret:fast} and use the individual bounds above:
\begin{align*}
     R_T
\le &
    \frac{64 M(u_{\inf})^2 \log(2K^2T^2) K}{\mu_{\inf}\alpha(u_{\inf}) \left( \min_{k \in [K]} F_k(u_{\inf}) \right)} 
   \log(T)
   +
   5 K^2
   +
  \frac{4K^2 \left(100 (\log T)^2 + 1 \right)}{ \min_{k \in [K]} F_k
\left( 
D_{\min}^{-1}
\left( 
\frac{2}{\mu_{\inf}}
D_{\max}\lrb{\frac{1}{K^2}}
\right)
\right) } 
\\ 
\le &
 \frac{64 u_{\inf}^{2(a_{\min}-1)-a_{\max}} K}{\mu_{\inf} \min_{k \in [K]} a_k(1-a_k ) } 
\log(2 K^2 T^2)  \log(T)
+ 5 K^2
 + 
  \frac{4K^2 \left(100 (\log T)^2 + 1 \right)}{ 
  \left(
\left(
\frac{a_{\min} \mu_{\inf}}{
2 a_{\max}
}
\right)^{\frac{1}{1-a_{\max}}}
K^{\frac{2(a_{\min}-1)}{1-a_{\max}}} 
\right)^{a_{\max}}
}
\\ 
\le &
\frac{64
\left(
\frac{a_{\min} \mu_{\inf}}{
2 a_{\max}
}
\right)^{\frac{2(a_{\min}-1)-a_{\max}}{1-a_{\max}}}
K^{1+\frac{(2(a_{\min}-1)-a_{\max})(a_{\min}-1)}{1-a_{\max}}}
}{\mu_{\inf} \min_{k \in [K]} a_k(1-a_k ) } 
\log(2 K^2 T^2)  \log(T)
+ 5 K^2
 \\ 
 & + 
4 \left(100 (\log T)^2 + 1 \right)
  \left(
\frac{a_{\min} \mu_{\inf}}{
2 a_{\max}
}
\right)^{\frac{a_{\max}}{a_{\max}-1}}
K^{2+\frac{2 a_{\max} (1-a_{\min})}{1-a_{\max}}}
\\ 
= &
\frac{64
K^{1+\frac{(2(1-a_{\min})+a_{\max})(1-a_{\min})}{1-a_{\max}}}
}{
\left(
\frac{a_{\min} \mu_{\inf}}{
2 a_{\max}
}
\right)^{\frac{2(1-a_{\min})+a_{\max}}{1-a_{\max}}}
\mu_{\inf} \min_{k \in [K]} a_k(1-a_k ) } 
\log(2 K^2 T^2)  \log(T)
+ 5 K^2
 \\ 
 & + 
 \frac{
4 \left(100 (\log T)^2 + 1 \right)
K^{2+\frac{2 a_{\max} (1-a_{\min})}{1-a_{\max}}}
}{
  \left(
\frac{a_{\min} \mu_{\inf}}{
2 a_{\max}
}
\right)^{\frac{a_{\max}}{1-a_{\max}}}
}.
\end{align*}
This proves \eqref{eq:large:log:bound}.

We now derive the simpler estimate \eqref{eq:simpler:estimate} in the present regime, i.e., when $T\ge3$ and
\[
    \frac{1}{(KT)^2} \ge 2e^{-50(\log T)^2}.
\]
Let
\[
    b \ceq 1-a_{\max}
    \qquad\text{and}\qquad
    Q \ceq
    \frac{
    K
    }{
    a_{\min}\mu_{\inf}b
    }.
\]
We use this notation only for the rest of the derivation.
Since $K\ge 2$, $a_{\min}\le a_{\max}<1$, and $\mu_{\inf}\le 1$, we have $Q\ge K\ge2$.

First, for $T\ge3$ and $K\ge2$,
\[
\log(2K^2T^2)
=
\log 2+2\log K+2\log T
\le
6\log K\log T .
\]
Hence,
\[
    \log(2K^2T^2)\log T
    \le
    6\log K(\log T)^2 .
\]
Moreover,
\[
    100(\log T)^2+1
    \le
    101(\log T)^2 .
\]

We now bound the first term in \eqref{eq:large:log:bound}.
Set
\[
    p
    \ceq
    \frac{2(1-a_{\min})+a_{\max}}{1-a_{\max}} .
\]
Since $1-a_{\min}\le 1$ and $a_{\max}\le1$, we have
\[
    p\le \frac{3}{b}.
\]
Also,
\[
    p(1-a_{\min})\le p\le \frac{3}{b}.
\]
Furthermore,
\[
    \min_{k\in[K]} a_k(1-a_k)
    \ge
    a_{\min}(1-a_{\max})
    =
    a_{\min}b ,
\]
and
\[
    \frac{2a_{\max}}{a_{\min}\mu_{\inf}}
    \le
    \frac{K}{a_{\min}\mu_{\inf}b}
    =
    Q .
\]
Therefore, using also $\log K\le K\le Q$ and
\[
    \frac{1}{a_{\min}\mu_{\inf}b}
    =
    \frac{Q}{K}
    \le Q ,
\]
we obtain
\begin{align*}
\frac{
K^{1+p(1-a_{\min})}
}{
\left(
\frac{a_{\min}\mu_{\inf}}{2a_{\max}}
\right)^p
\mu_{\inf}\min_{k\in[K]} a_k(1-a_k)
}
\log K
&\le
K^{1+p(1-a_{\min})}
\left(
\frac{2a_{\max}}{a_{\min}\mu_{\inf}}
\right)^p
\frac{\log K}{\mu_{\inf}a_{\min}b}
\\
&\le
Q^{1+p(1-a_{\min})}
Q^p
Q
Q
=
Q^{3+p+p(1-a_{\min})}
\\
&\le
Q^{3+\frac{6}{b}}
\le
Q^{\frac{9}{b}},
\end{align*}
where the last inequality uses $b\le1$.
Thus, the first term in \eqref{eq:large:log:bound} is at most
\[
    64\cdot 6\, Q^{\frac{9}{b}}(\log T)^2
    =
    384\, Q^{\frac{9}{b}}(\log T)^2 .
\]

Next, consider the third term in \eqref{eq:large:log:bound}.
Set
\[
    r
    \ceq
    \frac{2a_{\max}(1-a_{\min})}{1-a_{\max}}
    \qquad\text{and}\qquad
    s
    \ceq
    \frac{a_{\max}}{1-a_{\max}} .
\]
Then
\[
    r\le \frac{2}{b},
    \qquad
    s\le \frac{1}{b}.
\]
Hence,
\begin{align*}
\frac{
K^{2+r}
}{
\left(
\frac{a_{\min}\mu_{\inf}}{2a_{\max}}
\right)^s
}
\le
K^{2+r}
\left(
\frac{2a_{\max}}{a_{\min}\mu_{\inf}}
\right)^s
\le
Q^{2+r+s}
\le
Q^{2+\frac{3}{b}}
\le
Q^{\frac{5}{b}}
\le
Q^{\frac{9}{b}} .
\end{align*}
Therefore, the third term in \eqref{eq:large:log:bound} is at most
\[
    4\cdot 101\, Q^{\frac{9}{b}}(\log T)^2
    =
    404\, Q^{\frac{9}{b}}(\log T)^2 .
\]

Finally, since $Q\ge K\ge2$, $b\le1$, and $T\ge3$,
\[
    5K^2
    \le
    5Q^2
    \le
    5Q^{\frac{9}{b}}(\log T)^2 .
\]
Combining the three estimates gives
\[
    R_T
    \le
    (384+404+5)Q^{\frac{9}{b}}(\log T)^2
    \le
    800 Q^{\frac{9}{b}}(\log T)^2 .
\]
Substituting back
\[
    Q=
    \frac{
    K
    }{
    a_{\min}\mu_{\inf}(1-a_{\max})
    }
    \qquad\text{and}\qquad
    b=1-a_{\max}
\]
yields
\[
    R_T
    \le
    800
    \left(
    \frac{
    K
    }{
    a_{\min}\mu_{\inf}(1-a_{\max})
    }
    \right)^{\frac{9}{1-a_{\max}}}
    (\log T)^2 .
\]

It remains only to prove \eqref{eq:simpler:estimate} in the cases not covered above.
Since the per-round reward lies in $[0,K]$, we always have
\[
    R_T \le KT .
\]
Moreover, since the rewards are $[0,1]$-valued, $\mu_{\inf}\le 1$.
Therefore,
\[
    \frac{
    K
    }{
    a_{\min}\mu_{\inf}(1-a_{\max})
    }
    \ge K,
    \qquad
    \frac{9}{1-a_{\max}}\ge 9 .
\]
If $T=2$, then, since $K\ge2$,
\[
    R_T
    \le
    KT
    =
    2K
    \le
    K^2
    \le
    800
    \left(
    \frac{
    K
    }{
    a_{\min}\mu_{\inf}(1-a_{\max})
    }
    \right)^{\frac{9}{1-a_{\max}}}
    (\log T)^2 .
\]
Finally, suppose that $T\ge3$ and
\[
    \frac{1}{(KT)^2} < 2e^{-50(\log T)^2}.
\]
Taking logarithms gives
\[
    -2\log K -2\log T
    <
    \log 2 - 50(\log T)^2 ,
\]
and therefore
\[
    \log K
    >
    25(\log T)^2-\log T-\frac{\log 2}{2}.
\]
For $T\ge3$, the right-hand side is larger than $\log T$.
Hence $K>T$, and thus
\[
    R_T
    \le
    KT
    \le
    K^2
    \le
    800
    \left(
    \frac{
    K
    }{
    a_{\min}\mu_{\inf}(1-a_{\max})
    }
    \right)^{\frac{9}{1-a_{\max}}}
    (\log T)^2 .
\]
This concludes the proof.
\end{proof}

The next lemma makes explicit the sharper upper bound of 
\Cref{prop:explicit:constant} for the example of \Cref{t:LB:linear:bandit}. 

\begin{lemma} 
\label{lemma:explicit:Ktwo}
    Consider the setting of \Cref{prop:explicit:constant} and assume further that
    $K=2$, $a_1=a_2 = \frac{1}{2}$, $\mu_1 \ge \frac{1}{4}$, and $\mu_2 \ge \frac{1}{4}$.
    Then, for any $T\ge3$,
    \[
    R_T \le
     2097152 \sqrt{2} 
\log(8 T^2)  \log(T)
+ 20
  + 256 \left(100 (\log T)^2 + 1 \right).
    \]
\end{lemma}

\begin{proof}
For $K=2$, the technical condition in \Cref{prop:explicit:constant} reads
\[
    \frac{1}{4T^2} \ge 2e^{-50(\log T)^2}.
\]
Equivalently,
\[
    50(\log T)^2 \ge 2\log T+\log 8.
\]
This holds for all $T\ge3$.
Therefore, we may apply \eqref{eq:large:log:bound} in \Cref{prop:explicit:constant}.

We have $a_{\min} = a_{\max} = \frac{1}{2}$ and $\mu_{\inf} \ge \frac{1}{4}$. 
Thus, from \eqref{eq:large:log:bound},
\begin{align*}
    R_T 
& \le
    \frac{
    64 
    \cdot
2^{1+\frac{(2(1-\frac{1}{2})+\frac{1}{2})(1-\frac{1}{2})}{1-\frac{1}{2}}}
}{
\left(
\frac{\frac{1}{2} \frac{1}{4}}{
2 \frac{1}{2}
}
\right)^{\frac{2(1-\frac{1}{2})+\frac{1}{2}}{1-\frac{1}{2}}}
\frac{1}{4} \frac{1}{4} } 
\log(8 T^2)  \log(T)
+ 20
  + 
 \frac{
4 \left(100 (\log T)^2 + 1 \right)
2^{2+\frac{2 \frac{1}{2} (1-\frac{1}{2})}{1-\frac{1}{2}}}
}{
  \left(
\frac{\frac{1}{2} \frac{1}{4}}{
2 \frac{1}{2}
}
\right)^{\frac{\frac{1}{2}}{1-\frac{1}{2}}}
}
\\ 
& =
 \frac{
    64 
    \cdot
2^{1+\frac{3}{2}}
}{
\left( \frac{1}{8} \right)^3 \frac{1}{16} } 
\log(8 T^2)  \log(T)
+ 20
  + 
 \frac{
4 \left(100 (\log T)^2 + 1 \right)
2^3
}{
\frac{1}{8}
}
\\ 
& =
2\cdot 16 \cdot 512 \cdot 32 \cdot 4 \sqrt{2} 
\log(8 T^2)  \log(T)
+ 20
  + 256 \left(100 (\log T)^2 + 1 \right)
\\ 
& =
2097152 \sqrt{2} 
\log(8 T^2)  \log(T)
+ 20
  + 256 \left(100 (\log T)^2 + 1 \right).
\end{align*}
This concludes the proof.
\end{proof}

\section{Worst-case lower bound: proof of Theorem \ref{t:lower-bound-worst-case}}
\label{s:appe-lower-bound-worst-case}

\begin{proof}[Proof of Theorem \ref{t:lower-bound-worst-case}]
Recall the notation in \Cref{s:setting}.
Assume $K\ge2$ and $T\ge16$.
Let
\[
    m \ceq \lfloor K/2\rfloor .
\]
We construct a hard instance using the first $2m$ tasks, grouped into $m$ pairs.
If $K$ is odd, the remaining task is a dummy task.
For each $k\in[2m]$, define
\[
    F_k \ceq \I_{[\nicefrac{1}{m},1]}.
\]
If $K$ is odd, define $F_K\equiv0$ for the dummy task.

Let $\varepsilon \ceq T^{-1/2}$.
Let $\mathcal{B}(q)$ be the Bernoulli distribution with parameter $q$ for $q\in[0,1]$.
For any $j_1,\dots,j_m\in[2]$, let $\Pb_{j_1,\ldots,j_m}$ be a probability measure under which the rewards are independent over time and tasks, and such that, for each $t\in[T]$ and each $i\in[m]$, if $j_i=1$, then $G_{t,2i-1}$ has distribution $\mathcal{B}(\frac12+\varepsilon)$ and $G_{t,2i}$ has distribution $\mathcal{B}(\frac12)$, while, if $j_i=2$, then $G_{t,2i-1}$ has distribution $\mathcal{B}(\frac12)$ and $G_{t,2i}$ has distribution $\mathcal{B}(\frac12+\varepsilon)$.
If $K$ is odd, the dummy reward is deterministically zero.
Throughout the proof, for any $j_1,\ldots,j_m$ under consideration, we write $\Pb=\Pb_{j_1,\ldots,j_m}$ and denote the corresponding expectation by $\E$.
We also write $R_{j_1,\ldots,j_m,T}$ for the regret up to horizon $T$ under $\Pb_{j_1,\ldots,j_m}$.

To prove the theorem, it is enough to show that there exist $j_1,\ldots,j_m\in[2]$ such that
\begin{equation}
\label{e:regret-stuff:two}
R_{j_1,\ldots,j_m,T}
=
\Omega\brb{m\sqrt{T}}.
\end{equation}
Indeed, since $m=\lfloor K/2\rfloor\ge K/3$ for $K\ge2$, this implies
\[
    R_T^\star
    =
    \Omega\brb{K\sqrt{T}}.
\]

Write
\[
    I_{t,k}\ceq F_k(X_{t,k})
\]
for $t\in[T]$ and $k\in[2m]$.
We first show that, without loss of generality, we may restrict attention to algorithms satisfying
\[
    \sum_{k=1}^{2m} I_{t,k}=m
    \qquad\text{for every }t\in[T].
\]
Indeed, consider any full-feedback algorithm.
At any time $t$, replace every active coordinate $X_{t,k}\ge 1/m$, $k\in[2m]$, by $1/m$, replace every inactive coordinate and the dummy coordinate, if present, by $0$, and then choose $m-\sum_{k=1}^{2m}I_{t,k}$ inactive indices among $[2m]$ and set their allocations to $1/m$.
This produces a feasible action and can only increase the realized reward, because it preserves all previously active non-dummy tasks and activates additional non-dummy tasks with nonnegative rewards.
Moreover, since the completion indicators are deterministic functions of the played allocation and the algorithm has full-feedback access to the entire reward vector $G_t$, the transformed algorithm can internally simulate the history that the original algorithm would have observed.
Thus the transformation defines a valid online algorithm and does not reduce the information available to the simulation.
It is therefore sufficient to prove the lower bound for algorithms satisfying $\sum_{k=1}^{2m}I_{t,k}=m$ almost surely.

We next transform such an algorithm into one satisfying
\[
    I_{t,2i-1}+I_{t,2i}=1
    \qquad\text{for every }t\in[T]\text{ and }i\in[m].
\]
Fix a time $t$.
If there are pairs with both tasks active, let $\{u_1<\cdots<u_N\}$ be the set of indices $i\in[m]$ for which $I_{t,2i-1}=I_{t,2i}=1$.
Let $\{v_1<\cdots<v_N\}$ be the set of indices $i\in[m]$ for which $I_{t,2i-1}=I_{t,2i}=0$.
These two sets have the same cardinality because $\sum_{k=1}^{2m}I_{t,k}=m$.
For each $\ell\in[N]$, draw $U_\ell$ uniformly from $\{2u_\ell-1,2u_\ell\}$ and $V_\ell$ uniformly from $\{2v_\ell-1,2v_\ell\}$, independently of everything else conditionally on the past.
Then set $X_{t,U_\ell}=0$ and $X_{t,V_\ell}=1/m$.

We claim that this transformation preserves the expected reward under every instance $\Pb_{j_1,\ldots,j_m}$.
Fix $j_1,\ldots,j_m\in[2]$.
Let $\mathcal{G}_{t-1}$ denote the history up to the moment before the auxiliary random variables $U_1,V_1,\ldots,U_N,V_N$ are sampled.
For each $\ell\in[N]$, let $(a_\ell,b_\ell)=(2u_\ell-1,2u_\ell)$ if $j_{u_\ell}=1$ and $(a_\ell,b_\ell)=(2u_\ell,2u_\ell-1)$ if $j_{u_\ell}=2$.
Similarly, let $(a'_\ell,b'_\ell)=(2v_\ell-1,2v_\ell)$ if $j_{v_\ell}=1$ and $(a'_\ell,b'_\ell)=(2v_\ell,2v_\ell-1)$ if $j_{v_\ell}=2$.
Thus $a_\ell$ and $a'_\ell$ are the better tasks in their respective pairs.
Conditionally on $\mathcal{G}_{t-1}$,
\[
    \E\bsb{\mu_{V_\ell}-\mu_{U_\ell}\mid\mathcal{G}_{t-1}}
    =
    \frac{1}{2}\lrb{\frac12+\varepsilon+\frac12}
    -
    \frac{1}{2}\lrb{\frac12+\varepsilon+\frac12}
    =
    0.
\]
Therefore the expected reward is unchanged by the pair-balancing transformation.
As above, the transformed algorithm can internally simulate the original one because full feedback reveals the entire reward vector and the completion indicators are deterministic functions of the played allocation.
Hence it is sufficient to prove \eqref{e:regret-stuff:two} for full-feedback algorithms satisfying
\[
    I_{t,2i-1}+I_{t,2i}=1
    \qquad\text{for every }t\in[T]\text{ and }i\in[m].
\]

We now prove that, for every $j_1,\ldots,j_m\in[2]$, the regret of any such full-feedback algorithm satisfies
\begin{equation}
\label{e:regret-stuff}
R_{j_1,\ldots,j_m,T}
\ge
\frac{\sqrt{T}}{2}
\sum_{i=1}^m
p_{j_1,\ldots,j_m}(i),
\end{equation}
where
\[
p_{j_1,\ldots,j_m}(i)
\ceq
\begin{cases}
    \Pb_{j_1,\ldots,j_m}\lsb{E_i^c} & \text{if } j_i=1, \\
    \Pb_{j_1,\ldots,j_m}\lsb{E_i} & \text{if } j_i=2,
\end{cases}
\]
and
\[
    E_i
    \ceq
    \bcb{
    \sum_{t=1}^T I_{t,2i-1}\ge \frac{T}{2}
    }.
\]

Fix $j_1,\ldots,j_m\in[2]$.
Since the optimal allocation activates the better task in every pair, we have
\[
    \max_{x\in\Delta_K}
    \ban{F(x),\mu}
    =
    m\lrb{\frac12+\varepsilon}.
\]
Therefore,
\begin{align*}
R_{j_1,\ldots,j_m,T}
&=
    Tm\lrb{\frac12+\varepsilon}
    -
    \E\bsb{
    \sum_{t=1}^T
    \sum_{i=1}^m
    \lrb{
    I_{t,2i-1}\mu_{2i-1}
    +
    I_{t,2i}\mu_{2i}
    }
    }
\\
&=
    \varepsilon
    \sum_{i=1}^m
    \E\bsb{
    \sum_{t=1}^T
    \lrb{
    I_{t,2i-1}\I\{j_i=2\}
    +
    I_{t,2i}\I\{j_i=1\}
    }
    }.
\end{align*}
If $j_i=2$, then the inferior task in pair $i$ is $2i-1$, and hence
\[
    \E\bsb{\sum_{t=1}^T I_{t,2i-1}}
    \ge
    \frac{T}{2}\Pb\bsb{E_i}.
\]
If $j_i=1$, then the inferior task is $2i$.
Moreover, since $I_{t,2i-1}+I_{t,2i}=1$ for every $t$, the event $E_i^c$ implies
\[
    \sum_{t=1}^T I_{t,2i}
    \ge
    \frac{T}{2}.
\]
Thus
\[
    \E\bsb{\sum_{t=1}^T I_{t,2i}}
    \ge
    \frac{T}{2}\Pb\bsb{E_i^c}.
\]
Combining these inequalities and using $\varepsilon=T^{-1/2}$ gives \eqref{e:regret-stuff}.

We next recall the elementary bound that, for $p,q\in[\nicefrac14,\nicefrac34]$,
\begin{equation}
\label{eq:KL:bern}
\KL\lrb{\mathcal{B}(p),\mathcal{B}(q)}
\le
8(p-q)^2.
\end{equation}
Indeed, for fixed $p$, the function
\[
    q\mapsto \KL\lrb{\mathcal{B}(p),\mathcal{B}(q)}
\]
is convex, is minimized at $q=p$, and has second derivative
\[
    \frac{p}{q^2}+\frac{1-p}{(1-q)^2}
    \le 16
\]
on $[\nicefrac14,\nicefrac34]$.
A Taylor expansion gives \eqref{eq:KL:bern}.

Consider two vectors $j_1,\ldots,j_m,j'_1,\ldots,j'_m\in[2]$ that differ only in coordinate $k\in[m]$.
Since $T\ge16$, both $\frac12$ and $\frac12+\varepsilon$ belong to $[\nicefrac14,\nicefrac34]$.
By the chain rule for KL, and because full feedback reveals both rewards in the changed pair at every round,
\begin{align*}
&
\KL\lrb{
\Pb_{j_1,\ldots,j_m},
\Pb_{j'_1,\ldots,j'_m}
}
\\
&\quad
\le
T
\Bigg[
\KL\lrb{
\mathcal{B}\lrb{\frac12+\varepsilon},
\mathcal{B}\lrb{\frac12}
}
+
\KL\lrb{
\mathcal{B}\lrb{\frac12},
\mathcal{B}\lrb{\frac12+\varepsilon}
}
\Bigg]
\\
&\quad
\le
T\lrb{8\varepsilon^2+8\varepsilon^2}
=
16.
\end{align*}
Hence, by the Bretagnolle--Huber inequality \cite[Theorem 14.2]{lattimore2020bandit}, there is a universal constant $c_0>0$ such that
\begin{align}
\label{eq:hubert:bretagnolle}
&
p_{j_1,\ldots,j_m}(k)
+
p_{j'_1,\ldots,j'_m}(k)
\nonumber
\\
&\quad
=
\I\{j_k=1\}
\lrb{
\Pb_{j_1,\ldots,j_m}\bsb{E_k^c}
+
\Pb_{j'_1,\ldots,j'_m}\bsb{E_k}
}
\nonumber
\\
&\qquad
+
\I\{j_k=2\}
\lrb{
\Pb_{j_1,\ldots,j_m}\bsb{E_k}
+
\Pb_{j'_1,\ldots,j'_m}\bsb{E_k^c}
}
\ge
c_0.
\end{align}

Using \eqref{e:regret-stuff} and averaging over all instances,
\begin{align*}
&
\frac{1}{2^m}
\sum_{j_1=1}^2
\dots
\sum_{j_m=1}^2
R_{j_1,\ldots,j_m,T}
\\
&\quad
\ge
\frac{\sqrt{T}}{2}
\frac{1}{2^m}
\sum_{j_1=1}^2
\dots
\sum_{j_m=1}^2
\sum_{i=1}^m
p_{j_1,\ldots,j_m}(i)
\\
&\quad
=
\frac{\sqrt{T}}{2}
\frac{1}{2^m}
\sum_{i=1}^m
\sum_{j_1=1}^2
\dots
\sum_{j_{i-1}=1}^2
\sum_{j_{i+1}=1}^2
\dots
\sum_{j_m=1}^2
\Bigl(
p_{j_1,\ldots,j_{i-1},1,j_{i+1},\ldots,j_m}(i)
\\
&\hspace{33mm}
+
p_{j_1,\ldots,j_{i-1},2,j_{i+1},\ldots,j_m}(i)
\Bigr)
\\
&\quad
\ge
\frac{\sqrt{T}}{2}
\frac{1}{2^m}
\sum_{i=1}^m
\sum_{j_1=1}^2
\dots
\sum_{j_{i-1}=1}^2
\sum_{j_{i+1}=1}^2
\dots
\sum_{j_m=1}^2
c_0
\\
&\quad
=
\frac{c_0}{4}m\sqrt{T}.
\end{align*}
Therefore, there exists $j_1,\ldots,j_m\in[2]$ such that
\[
    R_{j_1,\ldots,j_m,T}
    \ge
    \frac{c_0}{4}m\sqrt{T}.
\]
Since $m=\lfloor K/2\rfloor\ge K/3$ for $K\ge2$, this proves
\[
    R_T^\star
    =
    \Omega\brb{K\sqrt{T}}.
\]
\end{proof}

\section{Worst-case upper bound: proof of Theorem \ref{t:upper-bound-worst-case}}
\label{s:appe-upper-bound-distribution-free}

Recall that, if the learner plays $X_1,X_2,\dots \in \Delta_K$, then, for any $t \in \N$, $k \in [K]$, and $\delta\in(0,1)$, we defined
\begin{equation}
\label{eq:vector:UCB}
    \UCB_{t,k}^\delta
\coloneqq
    \frac{ \sum_{s=1}^{t} B_{s,k}(X_{s,k})G_{s,k}}{\sum_{s=1}^{t} B_{s,k}(X_{s,k})}
    +
    \sqrt{\frac{\log\lrb{\frac{2}{\delta}}}{1+\sum_{s=1}^{t} B_{s,k}(X_{s,k})}}
\end{equation}
with the understanding that $0/0\coloneqq 0$.
Also, for conciseness, we use the same notation
$
    N_{t,k}
=
    \sum_{s=1}^t B_{s,k}(X_{s,k})
$
and $\tilde{G}_{\tau,k}$ as in \eqref{eq:intro:N:Gtilde} in \Cref{s:upper-bound}.
We then have, for each $t\in \N$, $k\in[K]$, and $\delta \in (0,1)$,
\begin{equation}
    \UCB^\delta_{t,k}
=
    \frac{ \sum_{\tau=1}^{N_{t,k}} \tilde{G}_{\tau,k}}{N_{t,k}}
    +
    \sqrt{\frac{\log\lrb{\frac{2}{\delta}}}{1+N_{t,k}}}.
\end{equation}

For $\delta \in (0,1/2]$, recall the definition of the ``good event'' $\cE_{T}^\delta$ in \eqref{eq:good:event:ET-main} in \Cref{s:upper-bound}.
From \Cref{cor:Hoeffding:gains} and a union bound, we have
\[
\Pb \lsb{\cE_{T}^\delta}
\ge 1- K T \delta.
\]
With the convention that $\sum_{t=1}^0 = 0$, and recalling that $\UCB^\delta_{0,k} = \sqrt{\log \frac{2}{\delta}}$, note that under the good event, for all $t\in\{0,\ldots,T\}$, for all $k \in [K]$,
\begin{equation}
\label{eq:UCB:good:bis}
\left|
\frac{\sum_{\tau=1}^{N_{t,k}}
	\tilde{G}_{\tau,k}
}{N_{t,k}}
-
\mu_k
\right|
\le
\sqrt{
	\frac{\log\lrb{\frac{2}{\delta}}}{1+N_{t,k}}}.
\end{equation}
When $N_{t,k} = 0$, this is indeed equivalent to $| \mu_k| \le \sqrt{\log\lrb{\frac{2}{\delta}}}$, which is true because $\log\lrb{\frac{2}{\delta}} \ge \log(4) >1$.
When $N_{t,k} \ge 1$, \eqref{eq:UCB:good:bis} holds because
\[
\sqrt{
	\frac{\log\lrb{\frac{2}{\delta}}}{1+N_{t,k}}}
\ge
\sqrt{
	\frac{\log\lrb{\frac{2}{\delta}}}{2N_{t,k}}}
\]
and because the good event holds.

We then recall that, for $t \ge 1$, the allocation $X_t \in \Delta_K$ played by the algorithm at time $t$ is such that
\begin{equation}
\label{eq:UCB}
\sup_{x \in \Delta_K} \lan{\UCB^\delta_{t-1} , F(x)}
=
\lan{\UCB^\delta_{t-1} , F(X_t)}.
\end{equation}
Let $x^\star \in \Delta_K$.
We have
\begin{align*}
    \sum_{t=1}^T
    \E \Bsb{
    \ban{\mu,F(x^\star)}- \ban{\mu,F(X_t)}
    }
&\le
    K T
    \lrb{1 -  \Pb \bsb{ \cE_{T}^\delta }}
+
    \E \left[
    \I_{ \cE_{T}^\delta }
    \sum_{t=1}^T
    \Brb{\ban{\mu,F(x^\star)}- \ban{\mu,F(X_t)}}\right]
\\
&\le
    K^2 T^2 \delta
+
    \E \left[
    \I_{ \cE_{T}^\delta }
    \sum_{t=1}^T
    \Brb{\ban{\mu,F(x^\star)}- \ban{\mu,F(X_t)}}\right].
\end{align*}

Next,
\begin{multline*}
\E \left[
	\I_{ \cE_{T}^\delta }
	\sum_{t=1}^T
	\Brb{\ban{
	\mu, F(x^{\star})}
	-
	\lan{\mu,
	F(X_{t})}}
	\right]
	\\
	=
	\E \left[
	\I_{ \cE_{T}^\delta }
	\sum_{t=1}^T
	\ban{\UCB^\delta_{t-1}
	,F(x^{\star})
	-
	F(X_{t})}
	\right]
	+
\E \left[
    \I_{ \cE_{T}^\delta }
    \sum_{t=1}^T
    \ban{\mu - \UCB^\delta_{t-1}
    ,F(x^{\star})
    -
    F(X_{t})}
\right].
\end{multline*}
Above, the first expectation is upper bounded by $0$ by definition of the algorithm.
Furthermore, if for any $y \in \bbR^K$ we define $y_+ \coloneqq \brb{ (y_1)_+,\dots,(y_K)_+}$ and, for all $z\in\R$, $z_+ \ceq \max\{z, 0 \}$, recalling that for each $k \in [K]$ and $t \in [T]$, when $\cE_T^\delta$ holds, we have $\UCB^\delta_{t-1,k} \ge\mu_k$, we get
\begin{align*}
	\E \left[
	\I_{ \cE_{T}^\delta }
	\sum_{t=1}^T
	\Brb{\ban{
	\mu, F(x^{\star})}
	-
	\ban{\mu,
	F(X_{t})}}
	\right]
    &\le
        \E \left[
	\I_{ \cE_{T}^\delta }
	\sum_{t=1}^T
	\ban{\mu - \UCB^\delta_{t-1}
	,F(x^{\star})
	-
	F(X_{t})}
	\right]
    \\
    & =
        \E \left[
	\I_{ \cE_{T}^\delta }
	\sum_{t=1}^T
	\ban{\UCB^\delta_{t-1} -\mu
	,
	F(X_{t}) - F(x^{\star})}
	\right]
    \\
    & \le
        \E \left[
	\I_{ \cE_{T}^\delta }
	\sum_{t=1}^T
	\Ban{\UCB^\delta_{t-1} -\mu
	,
	\brb{F(X_{t}) - F(x^{\star})}_+}
	\right].
\end{align*}
Using then \eqref{eq:UCB:good:bis}, we obtain
\begin{align*}
	&
	\E \lsb{
	\I_{ \cE_{T}^\delta }
	\sum_{t=1}^T
	\Ban{\UCB^\delta_{t-1} -\mu
	,
	\brb{F(X_{t}) - F(x^{\star})}_+}
	}
	\\
	& \qquad \le
	\E \lsb{
\I_{ \cE_{T}^\delta}
\sum_{t=1}^T
\sum_{k=1}^K
2
\sqrt{
	\frac{\log\lrb{\frac{2}{\delta}}}{1+\sum_{s=1}^{t-1}
B_{s,k}(X_{s,k})
}
}
	\lrb{
F_k(X_{t,k})
-
F_k(x^{\star}_k)
}_+
}
\\
&  \qquad  \le
2
\sqrt{\log\lrb{\frac{2}{\delta}}}
	\E \lsb{
\sum_{t=1}^T
\sum_{k=1}^K
	\frac{
			\brb{
		F_k(X_{t,k})
		-
		F_k(x^{\star}_k)
		}_+
	}{
		\sqrt{1+\sum_{s=1}^{t-1}
		B_{s,k}(X_{s,k})}
	}
}
\\
& \qquad \le
2
\sqrt{\log\lrb{\frac{2}{\delta}}}
\E \lsb{
\sum_{t=1}^T
\sum_{k=1}^K
\frac{
	F_k(X_{t,k})
}{
	\sqrt{1+\sum_{s=1}^{t-1}
		B_{s,k}(X_{s,k})}
}
}.
\end{align*}
We have, recalling the filtration $\cF_t$ from \eqref{eq:Gt}, and recalling that given $\cF_{t-1}$, $B_{t,k}(X_{t,k})$ has expectation $F_k(X_{t,k})$,
\begin{align*}
\E \left[
\frac{
	F_k(X_{t,k})
}{
	\sqrt{1+\sum_{s=1}^{t-1}
		B_{s,k}(X_{s,k})}
}
\right]
&
=
\E \left[
\frac{
1
}{
	\sqrt{1+\sum_{s=1}^{t-1}
		B_{s,k}(X_{s,k})}
}
\E \bsb{
B_{t,k}(X_{t,k})
\mid
\cF_{t-1}
}
\right]
\\
&
=
\E \left[
\E \left[
\left.
\frac{
	B_{t,k}(X_{t,k})
}{
	\sqrt{1+\sum_{s=1}^{t-1}
		B_{s,k}(X_{s,k})}
}
\right|
\cF_{t-1}
\right]
\right]
\\
&=
\E \left[
\frac{
	B_{t,k}(X_{t,k})
}{
	\sqrt{1+\sum_{s=1}^{t-1}
		B_{s,k}(X_{s,k})}
}
\right].
\end{align*}
Hence, we get
\[
\E \lsb{
	\I_{ \cE_{T}^\delta }
	\sum_{t=1}^T
	\Ban{\UCB^\delta_{t-1} -\mu
	,
	\brb{F(X_{t}) - F(x^{\star})}_+}
	}
\le
2
\sqrt{\log \lrb{\frac{2}{\delta}}}
\E \left[
\sum_{t=1}^T
\sum_{k=1}^K
\frac{
	B_{t,k}(X_{t,k})
}{
	\sqrt{1+\sum_{s=1}^{t-1}
		B_{s,k}(X_{s,k})}
}
\right].
\]
Recall the definition of $N_{T,k}$ above, and note that since $B_{s,k}(X_{s,k}) \in \{0,1\}$ for $k \in [K]$ and $s \in [T]$, we have
\[
\sum_{t=1}^T
\frac{
	B_{t,k}(X_{t,k})
}{
	\sqrt{1+\sum_{s=1}^{t-1}
		B_{s,k}(X_{s,k})}
}
=
\sum_{\tau=1}^{N_{T,k}}
\frac{1}{\sqrt{\tau}}
\le
2
\sqrt{N_{T,k}}
=
2
\sqrt{\sum_{t=1}^{T}
		B_{t,k}(X_{t,k})}.
\]
The inequality above is also valid when $N_{T,k}=0$, since both sides are equal to zero.

Hence we obtain
\begin{align*}
&
\E \left[
	\I_{ \cE_{T}^\delta }
	\sum_{t=1}^T
	\Brb{\ban{
	\mu, F(x^{\star})}
	-
	\ban{\mu,
	F(X_{t})}}
	\right]
\\
&\le
4
\sqrt{\log\lrb{\frac{2}{\delta}}}
\E \left[
\sum_{k=1}^K
	\sqrt{\sum_{t=1}^{T}
		B_{t,k}(X_{t,k})}
\right]
\\
&\le
4
\sqrt{K\log\lrb{\frac{2}{\delta}}}
\E \left[
\sqrt{
\sum_{k=1}^K \sum_{t=1}^{T}
	B_{t,k}(X_{t,k})
}
\right]
\\
&\le
4
\sqrt{K\log\lrb{\frac{2}{\delta}}}
\sqrt{
\sum_{k=1}^K \sum_{t=1}^{T}
\E \bsb{ B_{t,k}(X_{t,k})}
}
\\
&=
4
\sqrt{K \log \lrb{\frac{2}{\delta}}}
\sqrt{
\sum_{k=1}^K \sum_{t=1}^{T}
\E \bsb{ F_k(X_{t,k})}
}
\;.
\end{align*}
The second inequality follows from Cauchy's inequality.
The third inequality follows from Jensen's inequality.

Hence
\begin{align*}
    T \ban{\mu,F(x^\star)} - \sum_{t=1}^T
    \E \Bsb{
    \ban{\mu,F(X_t)}
    }
&=
    \sum_{t=1}^T
    \E \Bsb{
    \ban{\mu,F(x^\star)}- \ban{\mu,F(X_t)}
    }
\\
&\le
    K^2 T^2 \delta
    +
     4
    \sqrt{K \log \lrb{\frac{2}{\delta}}}
    \sqrt{
    \sum_{k=1}^K \sum_{t=1}^{T}
    \E \bsb{ F_k(X_{t,k})
    }
}\;.
\end{align*}
Since $x^\star$ was chosen arbitrarily, we can take the supremum over
$x^\star \in \Delta_K$. Substituting
$\delta \coloneqq 1/(KT)^2$ gives $K^2T^2\delta=1$, and therefore
\[
R_T
\le
1+
4
\sqrt{K \log \lrb{\frac{2}{\delta}}}
\sqrt{
\sum_{k=1}^K \sum_{t=1}^{T}
\E \bsb{ F_k(X_{t,k})}
},
\]
which is the claimed bound, up to logarithmic factors.

\section{\texorpdfstring{$\sqrt{T}$}{sqrt(T)} bandit feedback counterexample} \label{s:bandit:counter}

This section is devoted to 
providing the full setting of \Cref{t:LB:linear:bandit} and proving it.

\paragraph{Setting.}
Set $\gamma \coloneqq \frac{1}{6}$ and fix a time horizon $T \ge 4$.
Set $\e \coloneqq \gamma \cdot T^{-1/4}$.
Notice that with our choice of $\gamma$ and $T$, we have that $\e\le \frac{1}{8}$, a fact that we will use several times in what follows.
Define the two instances
\[
    \mu^{+} \coloneqq \lrb{\tfrac12+\e,\ \tfrac12-\e},
\qquad
    \mu^{-} \coloneqq \lrb{\tfrac12-\e,\ \tfrac12+\e}.
\]
For each $\sigma\in\lcb{+,-}$, consider two i.i.d.\ sequences $\brb{Y_t^{\sigma}(1)}_{t \in \N}$ and $\lrb{Y_t^{\sigma}(2)}_{t\in\N}$, with
\[
    Y_t^{\sigma}(1) \sim \mathrm{Bernoulli}\lrb{\tfrac12+\sigma\e}\;,
\qquad
    Y_t^{\sigma}(2) \sim \mathrm{Bernoulli}\lrb{\tfrac12-\sigma\e}\;,
\]
which we further build to be independent of each other.

Independently of $\brb{Y_t^{\sigma}(1),Y_t^{\sigma}(2)}_{t \in \N}$, let $(V_{t,k})_{t\in\N,k\in\{1,2\}}$ be an i.i.d.\ family of uniform random variables on $[0,1]$. For all $t\in\N$, $k\in\{1,2\}$, and $b\in[0,1]$, define
\[
    B_{t,k}(b)\coloneqq\I\{V_{t,k}\le\sqrt b\}.
\]
Then $B_{t,k}(b)\sim\mathrm{Bernoulli}(\sqrt b)$, and the random functions $B_{t,k}$ are independent across $(t,k)$ and jointly measurable in $(b,\omega)$.
Since we want to allow the learner to use randomized algorithms, we also assume that they can have sequential access to an i.i.d.\ sequence of $[0,1]$-uniforms $(U_t)_{t \in \N}$ used as random seeds, generated independently of the previous families of random variables.
Recall that our budget-to-success curves are defined by $F_1(b) \coloneqq \sqrt{b} \eqqcolon F_2(b)$, for each $b \in [0,1]$, and hence, each budget allocation $x \in \Delta_2$ satisfies
\[
    \Brb{F_1\brb{x_1}}^2 + \Brb{F_2\brb{x_2}}^2 = x_1 + x_2 = 1\;,
\]
so there exists $\vartheta \in [0, \pi/2]$ such that $F_1\brb{x_1} = \cos\vartheta$ and $F_2\brb{x_2} = \sin\vartheta$.
Vice versa, for each $\vartheta \in [0, \pi/2]$, if we set $x_1 = (\cos \vartheta)^2$ and $x_2 = (\sin \vartheta)^2$ we have that $x \in \Delta_2$ and $F_1\brb{x_1} = \cos \vartheta$ and $F_2\brb{x_2} = \sin \vartheta$.
Hence, we see that there is a 1-to-1 correspondence between choices of $\vartheta \in [0,\pi/2]$ and $x \in \Delta_2$ with a correspondence given by $\vartheta \mapsto \brb{(\cos \vartheta)^2,(\sin \vartheta)^2}$.

It follows that, when the learner interacts with the environment generated by $\sigma\in\lcb{+,-}$, the learner choosing budget allocations is equivalent to choosing angles $\vartheta_t^{\sigma}\in\bsb{0,\tfrac{\pi}{2}}$ and allocating
\[
    X_t^{\sigma}
\coloneqq
   \brb{ (\cos \vartheta_t^{\sigma})^2,(\sin \vartheta_t^{\sigma })^2}\;.
\]
We assume that the learner has access only to bandit feedback, i.e., at the end of each time $t$, they observe only the reward associated to the action they played:
\[
    O_t^{\sigma}
    \coloneqq
    Y_t^{\sigma}(1) \cdot B_{t,1}\brb{\cos^2 \lrb{\vartheta_t^{\sigma}}}
    +
    Y_t^{\sigma}(2) \cdot B_{t,2}\brb{\sin^2 \lrb{\vartheta_t^{\sigma}}}
     \in \{0,1,2\}.
\]
For each natural number $t \ge 2$, define the history space $\cH_t \coloneqq \lrb{[0,1]\times\{0,1,2\}}^{t-1} \times [0,1]$ and denote by $H_t^\sigma \coloneqq (U_1,O_1^\sigma,\dots,U_{t-1},O_{t-1}^\sigma,U_t)$ the history observed by the learner just before receiving the feedback at time $t$ when the underlying scenario is $\sigma \in \lcb{+,-}$.
Define $\cH_1\coloneqq[0,1]$ and $H_1^\sigma \coloneqq U_1$. 

Hence, the angles at time $t$ are generated via the relation $\vartheta_t^\sigma \coloneqq \alpha_t(H_t^\sigma)$ for some deterministic function $\alpha_t \colon \cH_t \to \lsb{0,\tfrac\pi 2}$ (the algorithm).

\paragraph{Regret.}
For $\sigma\in\lcb{+,-}$, the regret of the algorithm up to time $T$ when the underlying instance is determined by $\sigma$ is
\[
    R_T^{\sigma}
    \coloneqq
    T\cdot \max_{\vartheta\in[0,\pi/2]}\ \langle \mu^{\sigma}, \lrb{\cos\vartheta,\sin\vartheta}\rangle
    -
    \E\lsb{\sum_{t=1}^T \langle \mu^{\sigma}, (\cos\vartheta_t^{\sigma},\sin\vartheta_t^{\sigma})\rangle}\;.
\]

\paragraph{Per-round law.}
Now, under instance $\sigma \in \lcb{+,-}$ and for each $\vartheta \in \lsb{0,\tfrac \pi 2}$, define
\[
    p_{\sigma }(\vartheta) \coloneqq \lrb{\tfrac12+\sigma\e} \cdot \cos \vartheta,
    \qquad
    q_{\sigma}(\vartheta) \coloneqq \lrb{\tfrac12-\sigma\e} \cdot \sin \vartheta.
\]
Then, conditional on $\vartheta_t^\sigma$,
$Y^{\sigma}_t(1)B_{t,1}(\cos^2\vartheta_t^\sigma)\sim \mathrm{Bernoulli}(p_{\sigma}(\vartheta_t^\sigma))$ and
$Y^{\sigma}_t(2)B_{t,2}(\sin^2 \vartheta_t^{\sigma})\sim \mathrm{Bernoulli}(q_{\sigma}(\vartheta_t^{\sigma}))$,
independently, hence $O^{\sigma}_t$ is a sum of two independent Bernoulli random variables of parameters $p_{\sigma}(\vartheta_t^{\sigma})$ and 
$q_{\sigma}(\vartheta_t^{\sigma})$. Therefore
\begin{align*}  
    \Pb \lsb{O^{\sigma}_t=2 \mid \vartheta^\sigma_t = \vartheta}
&=
    p_{\sigma}(\vartheta)q_{\sigma}(\vartheta)
    = \lrb{\tfrac{1}{4}-\e^2} \cos\vartheta \sin \vartheta \eqqcolon \eta(\vartheta)\;,
\\
    \Pb\lsb{O^{\sigma}_t=1\mid \vartheta^\sigma_t = \vartheta}
&=
    p_{\sigma}(\vartheta)+q_{\sigma}(\vartheta)-2p_{\sigma}(\vartheta )q_{\sigma}(\vartheta)\;,
\\
    \Pb\lsb{O^{\sigma}_t=0\mid \vartheta^\sigma_t = \vartheta}
&=
    1-p_{\sigma}(\vartheta)-q_{\sigma}( \vartheta)+p_{\sigma}(\vartheta)q_{\sigma}( \vartheta).
\end{align*}

\paragraph{One-step KL expansion.}
For $\vartheta \in [0, \tfrac \pi 2]$, define $P_{t,\vartheta}^{\sigma} \coloneqq \Pb_{O^{\sigma}_t \mid \vartheta_t^\sigma = \vartheta}$ as the law of
$O^{\sigma}_t$ conditioned on $\vartheta_t^{\sigma} = \vartheta$.
Then
\begin{align*}
    \KL \lrb{P_{t,\vartheta}^{+},P_{t,\vartheta}^{-}}
&=
    \sum_{k\in\{0,1,2\}}
    \Pb \lsb{O^{+}_t=k \mid \vartheta_t^{+}=\vartheta}
    \log\frac{\Pb \lsb{O^{+}_t=k \mid  \vartheta_t^{+}=\vartheta}}{\Pb \lsb{O^{-}_t=k \mid  \vartheta_t^{-}=\vartheta}}.
\end{align*}

\begin{lemma}[Per-round KL upper bound.]
We have that
\[
    \KL\lrb{P_{t,\vartheta}^{+},P_{t,\vartheta}^{-}}
\le
    84\cdot \e^2 \cdot \labs{\vartheta - \tfrac\pi 4}^2\;.
\]
\end{lemma}

\begin{proof}
Fix $\vartheta\in\lsb{0,\tfrac{\pi}{2}}$.
For brevity, write
\begin{align*}
a&\coloneqq \cos\vartheta\;,\qquad b\coloneqq \sin\vartheta\;, \qquad
p_\sigma \coloneqq p_\sigma(\vartheta)\;, \qquad q_\sigma \coloneqq q_\sigma(\vartheta)\;,\qquad \eta \coloneqq \eta(\vartheta)\;.
\end{align*}

Set, for $k\in\{0,1,2\}$,
\[
r_k \coloneqq P_{t,\vartheta}^{+}\lsb{\{k\}}\;,\qquad s_k \coloneqq P_{t,\vartheta}^{-}\lsb{\{k\}}\;.
\]
Then $r_2=s_2=\eta$ and $r_0+r_1=s_0+s_1=1-\eta$.
Define:
\[
\alpha \coloneqq \frac{r_1}{1-\eta},\qquad \beta \coloneqq \frac{s_1}{1-\eta}.
\]

\medskip
\noindent\textbf{Step 1: Reducing to Bernoulli KL.}
By direct expansion,
\begin{align*}
\KL\lrb{P_{t,\vartheta}^{+},P_{t,\vartheta}^{-}}
&=\sum_{k=0}^2 r_k \log\frac{r_k}{s_k}\\
&= r_2\log\frac{r_2}{s_2} + r_1\log\frac{r_1}{s_1}+ r_0\log\frac{r_0}{s_0}\\
&= r_1\log\frac{r_1}{s_1}+ r_0\log\frac{r_0}{s_0}
\qquad(\text{since }r_2=s_2)\\
&=(1-\eta)\lrb{\alpha\log\frac{\alpha}{\beta}+(1-\alpha)\log\frac{1-\alpha}{1-\beta}}\\
&=(1-\eta)\,\mathrm{kl}(\alpha\|\beta),
\end{align*}
where $\mathrm{kl}$ denotes the Bernoulli KL divergence.

\medskip
\noindent\textbf{Step 2: bound $\mathrm{kl}(\alpha\|\beta)$ by a $\chi^2$-type inequality.}
Recall that, for Bernoulli laws, for all $\alpha,\beta\in(0,1)$, from \cite[Lemma 2.7]{Tsybakov2008}, 
\begin{equation}\label{eq:kl-chi2}
    \mathrm{kl}(\alpha\|\beta)
\le
    \chi^2(\alpha\|\beta)
=
    \frac{(\alpha-\beta)^2}{\beta(1-\beta)}\;.
\end{equation}
Thus,
\begin{equation}\label{eq:kl-reduced}
    \KL\lrb{P_{t,\vartheta}^{+},P_{t,\vartheta}^{-}}
\le
    (1-\eta)\cdot \frac{(\alpha-\beta)^2}{\beta(1-\beta)}\;.
\end{equation}

\medskip
\noindent\textbf{Step 3: compute $\alpha-\beta$ explicitly.}
Since $r_1=p_{+}+q_{+}-2\eta$ and $s_1=p_{-}+q_{-}-2\eta$,
\[
r_1-s_1=(p_{+}+q_{+})-(p_{-}+q_{-}).
\]
But
\[
p_{+}-p_{-} = 2\e a,\qquad q_{+}-q_{-} = -2\e b,
\]
hence
\[
r_1-s_1 = 2\e(a-b).
\]
Therefore,
\begin{equation}\label{eq:alpha-beta}
\alpha-\beta = \frac{r_1-s_1}{1-\eta} = \frac{2\e(a-b)}{1-\eta}.
\end{equation}

Plugging \eqref{eq:alpha-beta} into \eqref{eq:kl-reduced} yields
\begin{equation}\label{eq:KL-pre}
\KL\lrb{P_{t,\vartheta}^{+},P_{t,\vartheta}^{-}}
\le \frac{4\e^2(a-b)^2}{1-\eta}\cdot \frac{1}{\beta(1-\beta)}.
\end{equation}

\medskip
\noindent\textbf{Step 4: uniform lower bounds on $1-\eta$ and on $\beta(1-\beta)$.}

\emph{Lower bound on $1-\eta$.}
We have $ab\le \frac12$ because $a^2+b^2=1$ and $a,b\ge 0$ (maximum at $a=b=1/\sqrt2$).
Thus, recalling $\e \le \tfrac{1}{8}$, we have
\[
    \eta
=
    \lrb{\tfrac14-\e^2}ab
\le
    \tfrac14\cdot \tfrac12
=
    \tfrac18,
\qquad\text{hence}\qquad
    1-\eta
\ge
    \tfrac78\;.
\]
Therefore
\begin{equation}\label{eq:1minuseta}
\frac{1}{1-\eta}\le \frac{8}{7}.
\end{equation}

\emph{Lower bound on $\beta(1-\beta)$.}
Recall $\beta=s_1/(1-\eta)$ and
\[
    s_1
=
    p_{-}+q_{-}-2\eta\;.
\]
First, we lower bound $s_1$.
We have
\[
    p_-+q_- - \lrb{\tfrac12-\e}
=
    \lrb{\tfrac12-\e}(a-1)+\lrb{\tfrac12+\e}b
=
    -\lrb{\tfrac12-\e}(1-a)+\lrb{\tfrac12+\e}b\;.
\]
Using $1-a= 1- \cos \vartheta = 2\lrb{\sin(\vartheta/2)}^2$ and $b = \sin \vartheta =2\sin(\vartheta/2)\cos(\vartheta/2)$, this becomes
\[
    p_-+q_- - \lrb{\tfrac12-\e}
=
    2\sin\lrb{\tfrac{\vartheta}{2}}
    \Brb{
        \lrb{\tfrac12+\e}\cos\lrb{\tfrac{\vartheta}{2}}
        -
        \lrb{\tfrac12-\e}\sin\lrb{\tfrac{\vartheta}{2}}
    }\;.
\]
Since $\vartheta/2\in[0,\pi/4]$ we have $\cos(\vartheta/2)\ge \sin(\vartheta/2)$, and thus
\[
    \lrb{\tfrac12+\e}\cos\lrb{\tfrac{\vartheta}{2}}
    -
    \lrb{\tfrac12-\e}\sin\lrb{\tfrac{\vartheta}{2}}
=
    \lrb{\tfrac12-\e} \cdot\lrb{\cos\lrb{\tfrac{\vartheta}{2}}-\sin\lrb{\tfrac{\vartheta}{2}}}
    +2\e\cdot\cos\lrb{\tfrac{\vartheta}{2}}
\ge 0\;.
\]
Therefore
\[
    p_-+q_- \ge \tfrac12-\e\;,
\]
and hence
\[
    s_1
=
    (p_-+q_-)-2\eta
\ge
    \lrb{\tfrac12-\e}-2\cdot\tfrac18
\ge
    \tfrac14-\e
\ge
    \tfrac18\;.
\]
Second, we upper bound $s_1$.
Since $s_1=p_-+q_- -2\eta$ and $\eta=\lrb{\tfrac14-\e^2}ab$, we can rewrite
\begin{align*}
    s_1
=
    \lrb{\tfrac12-\e}a+\lrb{\tfrac12+\e}b-\lrb{\tfrac12-2\e^2}ab
=
    \lrb{\tfrac12+\e}(a+b-ab)
    +\e a\brb{\lrb{1+2\e}b-2}\;.
\end{align*}
First, notice that
\[
    \lrb{1+2\e}b-2 \le (1+2\e)-2 = -1+2\e \le -\tfrac34\;,
\]
so the correction term  $\e a(\lrb{1+2\e}b-2)$ is non-positive and, since for $a,b\in[0,1]$ we have $a+b-ab=1-(1-a)(1-b)\le 1$, recalling that $\e\le 1/8$, we get
\[
    s_1 \le \lrb{\tfrac12+\e}(a+b-ab)\le \tfrac12+\e \le \tfrac58\;.
\]

Putting the upper and lower bounds on $s_1$ together, we get
\[
    \frac57
=
    \frac{5/8}{7/8}
\ge
    \frac{s_1}{7/8}
\ge
    \frac{s_1}{1-\eta}
=
    \beta
=
    \frac{s_1}{1-\eta}
\ge
    s_1
\ge
    \tfrac18 \;.
\]
Thus $\beta\in[1/8,5/7]$, and since $\beta(1-\beta)$ is concave on $[0,1]$, its minimum on this interval is
attained at an endpoint:
\[
    \beta(1-\beta)
\ge
    \min\lrb{\tfrac18\cdot\tfrac78,\ \tfrac57\cdot\tfrac27}
=
    \min\lrb{\tfrac{7}{64},\ \tfrac{10}{49}}
=
    \tfrac{7}{64}\;.
\]
Therefore
\begin{equation}\label{eq:beta-lb}
\frac{1}{\beta(1-\beta)}\le \frac{64}{7}.
\end{equation}

\medskip
\noindent\textbf{Step 5: conclude the explicit constant.}
Combining \eqref{eq:KL-pre}, \eqref{eq:1minuseta}, and \eqref{eq:beta-lb} gives
\[
    \KL \lrb{P_{t,\vartheta}^{+},P_{t,\vartheta}^{-}}
\le
    4\e^2(a-b)^2\cdot \frac{8}{7}\cdot \frac{64}{7}
=
    \frac{2048}{49}\cdot\e^2 \cdot (a-b)^2
\le
    42 \cdot\e^2 \cdot (\cos\vartheta-\sin\vartheta)^2\;.
\]

\medskip
\noindent\textbf{Step 6: relate $(\cos\vartheta-\sin\vartheta)^2$ to $|\vartheta-\pi/4|^2$.}
Note that
\[
    \cos\vartheta-\sin\vartheta
=
    \sqrt{2} \cdot \sin \lrb{\frac{\pi}{4}-\vartheta},
\]
and since $\labs{\sin x} \le |x|$, we get
\[
    (\cos\vartheta-\sin\vartheta)^2
\le
    2 \cdot \labs{\vartheta-\frac{\pi}{4}}^2\;.
\]
We conclude that
\[
    \KL\lrb{P_{t,\vartheta}^{+},P_{t,\vartheta}^{-}}
\le
    42 \cdot \e^2 \cdot(\cos\vartheta-\sin\vartheta)^2
\le
    84 \cdot \e^2 \cdot \labs{\vartheta-\frac{\pi}{4}}^2\;.
\]
This concludes the proof.
\end{proof}

\begin{lemma}[KL chain rule]\label{lem:KL-chain}
It holds that
\begin{align}
    \KL \lrb{\Pb_{(H_T^+,O_T^+)},\Pb_{(H_T^-,O_T^-)}}
\le
    84 \cdot \e^2 \cdot \sum_{t=1}^T \E \lsb{ \labs{\vartheta_t^+-\tfrac\pi 4}^2}.
    \label{eq:KL-chain-transcript}
\end{align}
\end{lemma}
\begin{proof}
\begin{align*}
&
    \KL \lrb{\Pb_{(H_T^+,O_T^+)},\Pb_{(H_T^-,O_T^-)}}
\\
&\quad=
    \KL \lrb{\Pb_{H_{T}^+},\Pb_{H_{T}^-}}
    +
    \int_{\cH_{T}} \KL\lrb{\Pb_{O_T^+\mid H_T^+ = h},\Pb_{O_T^-\mid H_T^- = h}}  \dif \Pb_{H_T^+}(h)
\\
&\quad=
    \KL \lrb{\Pb_{(H_{T-1}^+,O_{T-1}^+,U_T)},\Pb_{(H_{T-1}^-,O_{T-1}^-,U_T)}}
    +
    \E\lsb{ \lsb{\KL\lrb{P_{T,\vartheta}^{+},P_{T,\vartheta}^{-}}}_{\vartheta = \alpha_T(H_T^+)} }
\\
&\quad=
    \KL \lrb{\Pb_{(H_{T-1}^+,O_{T-1}^+)},\Pb_{(H_{T-1}^-,O_{T-1}^-)}}
    +
    \E\lsb{ \lsb{\KL\lrb{P_{T,\vartheta}^{+},P_{T,\vartheta}^{-}}}_{\vartheta = \alpha_T(H_T^+)} }
\\
&\quad\le
    \KL \lrb{\Pb_{(H_{T-1}^+,O_{T-1}^+)},\Pb_{(H_{T-1}^-,O_{T-1}^-)}}
    +
    84\cdot \e^2 \cdot \E\lsb{ \labs{\vartheta^+_T - \tfrac\pi 4}^2 }
\\
&\quad\le
\cdots
\le
    \KL \lrb{\Pb_{(H_1^+,O_1^+)},\Pb_{(H_1^-,O_{1}^-)}} + 84 \cdot \e^2 \cdot \sum_{t=2}^T \E\lsb{ \labs{\vartheta^+_t - \tfrac\pi 4}^2 }
\\
&\quad=
    \KL \lrb{\Pb_{(U_1,O_1^+)},\Pb_{(U_1,O_{1}^-)}} + 84 \cdot \e^2 \cdot \sum_{t=2}^T \E\lsb{ \labs{\vartheta^+_t - \tfrac\pi 4}^2 }
\\
&\quad=
    \KL \lrb{\Pb_{U_1},\Pb_{U_1}} + \int_{[0,1]} \KL \lrb{ \Pb_{O_1^+ \mid U_1 = u } , \Pb_{O_1^- \mid U_1 = u }} \dif\Pb_{U_1}(u) + 84 \cdot \e^2 \cdot \sum_{t=2}^T \E\lsb{ \labs{\vartheta^+_t - \tfrac\pi 4}^2 }
\\
&\quad=
    0+\E\lsb{ \lsb{\KL\lrb{ P^+_{1,\vartheta} ,P^-_{1,\vartheta} }}_{\vartheta = \alpha_1(U_1)} } + 84 \cdot \e^2 \cdot \sum_{t=2}^T \E\lsb{ \labs{\vartheta^+_t - \tfrac\pi 4}^2 }
\\
&\quad\le
    84 \cdot \e^2 \cdot \sum_{t=1}^T \E\lsb{ \labs{\vartheta^+_t - \tfrac\pi 4}^2 }\;.
\end{align*}
    This concludes the proof.
\end{proof}

We are now ready to prove \Cref{t:LB:linear:bandit}, leveraging \Cref{lem:KL-chain}. 

\begin{proof}[Proof of \Cref{t:LB:linear:bandit}]
\medskip\noindent\textbf{Step 1: optimal angle and a one-sided gap at $\pi/4$.}
For $\sigma\in\{+,-\}$ and $\vartheta \in \lsb{0,\tfrac \pi 2}$, define
\[
    f_\sigma(\vartheta)
\coloneqq
    \langle \mu^\sigma,(\cos\vartheta,\sin\vartheta)\rangle\;.
\]
The maximizer $\vartheta^\star_\sigma$ is attained at the unit vector pointing in the direction of $\mu^\sigma$, i.e.,
\[
(\cos\vartheta_\sigma^\star,\sin\vartheta_\sigma^\star)=\frac{\mu^\sigma}{\|\mu^\sigma\|_2}.
\]
Equivalently, using $\tan\vartheta_\sigma^\star=\mu^\sigma_2/\mu^\sigma_1$, we get
\begin{equation} \label{eq:with:tan}
    \vartheta_\sigma^\star
=
    \arctan\lrb{\frac{\frac12-\sigma\e}{\frac12+\sigma\e}}   
=
    \frac{\pi}{4}-\sigma\cdot\arctan(2\e)\;.
\end{equation}
Moreover,
\[
    \max_{\vartheta\in[0,\pi/2]} f_\sigma(\vartheta)=\|\mu^\sigma\|_2
=
    \frac{\sqrt2}{2}\cdot\sqrt{1+4\e^2}\;,
\]
while, at $\vartheta=\frac{\pi}{4}$, we have
\[
    f_\sigma\lrb{\tfrac{\pi}{4}}
=
    \lrb{\tfrac12+\sigma\e}\cdot\tfrac{\sqrt2}{2} + \lrb{\tfrac12-\sigma\e}\cdot\tfrac{\sqrt2}{2}
=
    \tfrac{\sqrt2}{2}\;.
\]
We define the gap (which is independent of $\sigma \in \{+,-\}$)
\[
    g(\e)
\coloneqq
    \lno{\mu^\sigma}_2-f_\sigma\lrb{\tfrac{\pi}{4}}
=
    \frac{\sqrt2}{2} \cdot\lrb{\sqrt{1+4\e^2}-1}
=
    \frac{\sqrt2}{2}\cdot \frac{4\e^2}{\sqrt{1+4\e^2}+1}\;,
\]
and, since $\sqrt{1+4\e^2}\le 1+2\e^2\le 1+\frac{1}{32}=\frac{33}{32}$ (because $\e \le \frac18$), we have
\[
    g(\e)
\ge
    \frac{\sqrt2}{2}\cdot \frac{4\e^2}{\frac{33}{32}+1}
=
    \frac{\sqrt2}{2}\cdot \frac{128}{65}\cdot \e^2
=
    \frac{64\sqrt2}{65}\cdot \e^2\;.
\]

\medskip\noindent\textbf{Step 2: playing on the wrong side costs at least $g(\e)$.}
For $\vartheta\in\lsb{\frac{\pi}{4},\frac{\pi}{2}}$,
\[
    f_+'(\vartheta)
=
    -\lrb{\tfrac12+\e}\sin\vartheta+\lrb{\tfrac12-\e}\cos\vartheta
\le
    -\lrb{\tfrac12+\e}\cos\vartheta+\lrb{\tfrac12-\e}\cos\vartheta
=
    -2\e\cos\vartheta\le0\;,
\]
so $f_+$ is nonincreasing on $\lsb{\frac{\pi}{4},\frac{\pi}{2}}$, hence
\[
    \vartheta
\ge
    \frac{\pi}{4}
\quad\Longrightarrow\quad
    \lno{\mu^+}_2-f_+(\vartheta)
\ge
    \lno{\mu^+}_2-f_+\lrb{\tfrac{\pi}{4}}
=
    g(\e).
\]
Similarly, $f_-$ is increasing on $\lsb{0,\frac{\pi}{4}}$, hence
\[
    \vartheta
\le
    \frac{\pi}{4}
\quad\Longrightarrow\quad
    \lno{\mu^-}_2-f_-(\vartheta)
\ge
    g(\e)\;.
\]
Therefore, if we consider the instantaneous pointwise regret
\[
    \Delta_t^\sigma
\coloneqq
    \lno{\mu^\sigma}_2-\langle \mu^\sigma,(\cos\vartheta_t^\sigma,\sin\vartheta_t^\sigma)\rangle
\]
we have that
\[
    \Delta_t^+
\ge
    g(\e) \cdot\I\lcb{\vartheta_t^+\ge \tfrac{\pi}{4}}\;,
\qquad
    \Delta_t^-
\ge
    g(\e)\cdot \I\lcb{\vartheta_t^-< \tfrac{\pi}{4}}.
\]

\medskip\noindent\textbf{Step 3: change-of-measure inequality.}
Consider the events
\[
    A_t
\coloneqq
    \lcb{(u_1,k_1,\dots,u_{T-1},k_{T-1},u_T,k_T) \in \cH_T \times\lcb{0,1,2} \mid \alpha_t(u_1,k_1,\dots,u_{t-1},k_{t-1},u_t) < \frac{\pi}{4}}\;,
\]for $t = 1,2,\dots,T$.
Recall that for any two probability measures $\mu$ and $\nu$ over some sigma algebra $\cG$ and any event $G \in \cG$, we have that 
\[
    \labs{\mu[G]-\nu[G]}
\le
    \lno{\mu-\nu}_{\textrm{TV}}\;,
\]
and hence, if we define the two push-forward probability measures on the sample space $\cH_T \times \lcb{0,1,2}$ as $\Pb_+ \coloneqq \Pb_{(H_T^+,O_T^+)}$ and $\Pb_-\coloneqq \Pb_{(H_T^-,O_T^-)}$, we have that
\[
    \Pb_+\lsb{A_t^c}+\Pb_-\lsb{A_t}
=
    1-\lrb{\Pb_+\lsb{A_t}-\Pb_-\lsb{A_t}}
\ge
    1-\lno{\Pb_{(H_T^+,O_T^+)}-\Pb_{(H_T^-,O_T^-)}}_{\textrm{TV}}\;.
\]
Combining with Step 2, summing over $t$, using Pinsker inequality \citep{Tsybakov2008}, and Lemma~\ref{lem:KL-chain}, we get
\begin{align}
    R_T^+ + R_T^-
&=
    \sum_{t=1}^T \E[\Delta_t^+]+\sum_{t=1}^T \E[\Delta_t^-]
\ge
    g(\e) \cdot\sum_{t=1}^T\lrb{\Pb_+\lsb{A_t^c}+\Pb_-\lsb{A_t}}
\nonumber\\
&\ge
    \frac{64\sqrt2}{65}\cdot \e^2 \cdot\sum_{t=1}^T\lrb{\Pb_+\lsb{A_t^c}+\Pb_-\lsb{A_t}}
\nonumber\\
&\ge
    \frac{64\sqrt2}{65}\cdot \e^2 \cdot T \cdot \lrb{1-\lno{\Pb_{(H_T^+,O_T^+)}-\Pb_{(H_T^-,O_T^-)}}_{\textrm{TV}}}
\nonumber\\
&\ge
    \frac{64\sqrt2}{65}\cdot \e^2 \cdot T \cdot \lrb{1-\sqrt{\tfrac12 \cdot  \KL\lrb{\Pb_{(H_T^+,O_T^+)},\Pb_{(H_T^-,O_T^-)}}}}
\nonumber\\
&\ge
    \frac{64\sqrt2}{65}\cdot \e^2 \cdot T \cdot \lrb{1-\sqrt{42 \cdot \e^2 \cdot \sum_{t=1}^T \E \lsb{\labs{\vartheta_t^+-\tfrac{\pi}{4}}^2}  }}\;.
\label{eq:sum-regret-vs-tv}
\end{align}

\medskip\noindent\textbf{Step 4: angles to regret.}
Next, let's relate $\sum_{t=1}^T\E\bsb{ |\vartheta_t^+-\frac{\pi}{4}|^2}$ to $R_T^+$.

For $\vartheta \in \lsb{0, \tfrac{\pi}{2}}$, let
$
    z(\vartheta)
=
    (\cos\vartheta,\sin\vartheta)
$
and
$
    z_+^\star
=
    z(\vartheta_+^\star)
=
    \mu^+/\lno{\mu^+}_2$.
Then
\[
    \Delta_t^+
=
    \lno{\mu^+}_2 \cdot \brb{1-\langle z(\vartheta_t^+),z_+^\star\rangle}
=
    \lno{\mu^+}_2 \cdot \lrb{1-\cos(\vartheta_t^+-\vartheta_+^\star)}
=
    2\cdot \lno{\mu^+}_2 \cdot \lrb{\sin \lrb{\tfrac{\vartheta_t^+-\vartheta_+^\star}{2}}}^2.
\]
Since
$\labs{\sin\lrb{\tfrac{x}{2}}}\ge \tfrac{|x|}{\pi}$ for $x\in\lsb{-\frac{\pi}{2},\frac{\pi}{2}}$, leveraging that $|\vartheta_t^+-\vartheta_+^\star|\le \tfrac{\pi}{2}$, we have
\[
    \Delta_t^+
\ge
    2\cdot \lno{\mu^+}_2 \cdot \lrb{\frac{\labs{\vartheta_t^+-\vartheta_+^\star}}{\pi}}^2
\quad\Longrightarrow\quad
    \labs{\vartheta_t^+-\vartheta_+^\star}^2
\le
    \frac{\pi^2}{2\cdot \lno{\mu^+}_2}\cdot \Delta_t^+\;.
\]
Using $\|\mu^+\|_2\ge \frac{\sqrt{2}}{2}$, we then get
\[
    \labs{\vartheta_t^+-\vartheta_+^\star}^2
\le
    \frac{\pi^2\sqrt2}{2}\cdot\Delta_t^+\;.
\]
Also, since $\vartheta_+^\star=\frac{\pi}{4}-\arctan(2\e)$,
from \eqref{eq:with:tan},
we get
\[
    \labs{\vartheta_+^\star-\tfrac{\pi}{4}}
=
    \arctan(2\e)
\le
    2\e\;.
\]
Therefore,
\[
    \labs{\vartheta_t^+-\tfrac{\pi}{4}}^2
\le
    2\labs{\vartheta_t^+-\vartheta_+^\star}^2 + 2\labs{\vartheta_+^\star-\tfrac{\pi}{4}}^2
\le
    \pi^2 \sqrt2\cdot \Delta_t^+ + 8\cdot \e^2.
\]
Summing over $t$ and taking expectations yields
\[
    \sum_{t=1}^T \E\lsb{\left|\vartheta_t^+-\tfrac{\pi}{4}\right|^2}
\le
    \pi^2\sqrt2 \cdot R_T^+ + 8 \cdot \e^2\cdot T.
\]
Plugging back into \eqref{eq:sum-regret-vs-tv}, and recalling that $\e = \gamma \cdot T^{-1/4}$ and that $\gamma = \tfrac{1}{6}$, we get
\begin{align}
    R_T^+ + R_T^-
&\ge
    \frac{64\sqrt2}{65}\cdot \e^2 \cdot T \cdot \lrb{1-\sqrt{42 \cdot \e^2 \cdot \lrb{\pi^2\sqrt2 \cdot R_T^+ + 8 \cdot \e^2\cdot T}  }}
\nonumber
\\
&=
    \frac{64\sqrt2}{65}\cdot \gamma^2  \cdot \lrb{1-\gamma \cdot\sqrt{42  \cdot  \lrb{\pi^2\sqrt2  \cdot \frac{R_T^+}{\sqrt{T}} + 8 \cdot\gamma^2}  }}\cdot\sqrt{T}
\nonumber
\\
&=
    \frac{16\sqrt2}{585}  \cdot \lrb{1-\frac{1}{6} \cdot\sqrt{42  \cdot  \lrb{\pi^2\sqrt2  \cdot \frac{R_T^+}{\sqrt{T}} + \frac{2}{9}}  }}\cdot\sqrt{T}.
\label{eq:kl-vs-regret}
\end{align}

\medskip\noindent\textbf{Step 5: conclusion.}
Now, if $R_T^+ \ge \frac{1}{140} \cdot \sqrt{T}$ then the conclusion follows.
Otherwise, $R_T^+ < \frac{1}{140} \cdot \sqrt{T}$.
Then, from \eqref{eq:kl-vs-regret}, we have that
\[
    R_T^-
\ge
    \lrb{\frac{16\sqrt2}{585} \cdot \lrb{1-\frac{1}{6} \cdot\sqrt{42  \cdot  \lrb{\pi^2\sqrt2  \cdot \frac{1}{140} + \frac{2}{9}}  }} -\frac{1}{140} }\cdot\sqrt{T}
\ge
    \frac{1}{140} \cdot \sqrt{T}\;.
\]
This concludes the proof.
\end{proof}

\section{Related work}
\label{s:related_works}

Concurrent and independent work by \citet{montanari2026learning} studies a related censored resource-allocation problem with Bernoulli rewards. 
The models overlap when threshold distributions are known and rewards are independent Bernoulli variables.
Their feedback additionally reveals the threshold whenever success occurs. In the shared setting with known threshold distributions and independent Bernoulli rewards, this observation can be simulated from our feedback.
Indeed, success means that the Bernoulli reward equals one and the threshold does not exceed the allocated budget $X_{t,k}$. By independence, the conditional distribution of the threshold is therefore simply its known distribution truncated to $[0,X_{t,k}]$.
Whenever $B_{t,k}(X_{t,k})G_{t,k}=1$, we can draw a fresh sample from this truncated distribution and supply it to their algorithm together with the success indicator.
On failure, no threshold value is required.
The resulting feedback has the same conditional distribution as in their model. 
We now make a comparison between their stated results and ours in the shared setting, considering first the $\sqrt T$ rates.
While our optimism-based analysis yields \emph{worst-case} $K\sqrt{T}$ guarantees that hold for
arbitrary\footnote{Up to the existence of the argmax in \Cref{algo:ucbowski}, see \Cref{rem:approx-oracle}.} known nondecreasing curves, including discontinuous ones, their stated bound is instance-dependent, with a reported leading term of order $\widetilde{\mathcal O}(K^{3/2}\sqrt T)$ and an additive $\mathcal O(K^2\log T)$ term.
Additional dependence remains through instance-specific quantities, so the displayed powers of $K$ do not constitute a uniform comparison of the two bounds.
Their Theorem~3.1 also establishes an $\Omega(\sqrt{KT})$ lower bound with known uniform threshold distributions.
After normalizing the budget, these correspond to $F_k(x)=x$.
On this subclass, $\sum_k F_k(x_k)\le 1$, so our refined bound in \Cref{t:upper-bound-worst-case} yields $\widetilde{\mathcal O}(\sqrt{KT})$.
Our $\Omega(K\sqrt T)$ lower bound concerns the larger class of arbitrary nondecreasing curves.
We now make a comparison between our fast-rate regime and theirs.
Our fast-rate analysis uses the divergence of the slope of the known budget-to-success curves at the origin to \emph{prove} that optimal allocations are in the interior of the simplex, and then combines this structural property with optimism to obtain sufficiently many censored semi-bandit observations.
By contrast, their stated polylogarithmic bound (Theorem~4.6) assumes an interior optimum, activation probabilities bounded away from zero, and a globally smooth, uniformly strongly concave surrogate satisfying parameter regularity conditions (Assumption~3).
The theorem also requires the oracle to select a budget-saturating maximizer, and the algorithm uses explicit cyclic boosted exploration.
Consequently, their and our fast-rate assumptions are not directly comparable: theirs do not cover singular power-law curves $x\mapsto x^{\alpha_i}$, $\alpha_i\in(0,1)$, which are central examples in our fast-rate regime, while ours do not cover finite-slope exponential CDFs such as $x\mapsto 1-e^{-\lambda x}$, which are central in their parametric threshold setting.

Our model is also related to stochastic linear bandits; see, e.g., \citet[Chapter~19]{lattimore2020bandit}. 
In the linear bandit setting, the learner chooses an action $X_t$ and observes a scalar reward
$\langle \mu, X_t\rangle + \eta_t$, where $\mu$ is an unknown parameter and $\eta_t$ is a zero-mean subgaussian noise term, so that $\langle \mu, X_t\rangle$ is the expected reward.
In our model, the expected reward has the analogous linear form $\langle \mu, F(X_t)\rangle$, where $X_t\in\Delta_K$ is chosen by the learner and $F$ is known.
Equivalently, the learner chooses the feature vector $F(X_t)$, and the dependence on $\mu$ is linear as in stochastic linear bandits.

The feedback, however, is fundamentally different.
Linear bandits provide exactly one scalar observation per round, which is an unbiased estimate of a linear functional of $\mu$.
In contrast, under censored semi-bandit feedback our learner observes a \emph{random number} of scalar rewards per round (one per completed task), and this number can have expectation much larger than one depending on $F_1,\dots,F_K$.
Thus the feedback can be substantially more informative than in linear bandits (and the full-feedback model is richer still).
This difference underlies some of our main results.
In particular, our fast regret bound (\Cref{t:upper-bound-speed-up}) arises in regimes where the feedback is ``rich'' with high probability; see \Cref{lemma:lb:cum:budget} in Appendix~\ref{s:appe-speed-ups}.
In \Cref{t:LB:linear:bandit}, we show that an analogous fast rate does \emph{not} occur in the corresponding linear bandit instance. 

For our worst-case upper bound (\Cref{t:upper-bound-worst-case}), another subtle, but fundamental distinction from linear bandits arises.
In our model the per-round reward lies in $[0,K]$, and we obtain worst-case regret of order $K\sqrt{T}$.
In the linear bandit literature (e.g., \citet{bartlett2008high,dani2008stochastic}), the corresponding worst-case regret scales as $K\sqrt{T}$ under a standardized reward range such as $[-1,1]$.
If one rescales that range to length $\Theta(K)$, then those bounds incur an additional factor $K$ relative to ours.
This gap intuitively reflects the different feedback: when the budget-to-success curves $F_1,\dots,F_K$ are small, both the expected reward $\langle \mu, F(X_t)\rangle$ and the effective regret range are small, whereas when they are large the regret range grows but the expected number of observed scalar rewards per round (one per completed task) can exceed one, unlike in linear bandits.

Stochastic linear bandits are traditionally formulated with $X$ and $\mu$ lying in Euclidean balls (``$(\ell_2,\ell_2)$ geometry''). 
More generally, regret bounds for arbitrary $(\ell_p,\ell_q)$ geometries (with $p$ and $q$ being conjugate exponents) were obtained in \citet{bubeck2012towards,bubeck2018sparsity}.
In the adversarial \emph{bandits with experts} model of \citet{auer2002nonstochastic}, the geometry is $(\ell_1,\ell_\infty)$ as in our setting.
However, the feedback there is a single coordinate $\theta_{t,I_t}$, where $\theta_t$ is the adversarial reward vector at time $t$ and $I_t\in[K]$ is drawn according to the chosen distribution $X_t\in\Delta_K$.

Closely related resource-allocation models were studied by
\citet{lattimore2014optimal} and \citet{lattimore2015linear}.
The former consider capped-linear completion curves
$F_k(x_k)=\min\{1,x_k/\nu_k\}$, where the cutoff parameters $\nu_k$ are unknown,
and all completed jobs have unit value ($\mu_k=1$).
The latter extends this line of work to multiple resource types: task $k$ receives
a resource vector $M_{t,k}$ and succeeds with probability
$\min\{1,\langle M_{t,k},\nu_k\rangle\}$, where $\nu_k$ is unknown.
These settings are complementary to ours.
They learn the resource-to-completion map within a clipped-linear parametric
family, whereas we assume that the budget-to-success curves $F_k$ are known and
learn the reward means $\mu_k$ from censored reward observations.
The models overlap in special cases; for example, when there is a single
resource type, $F_k(x)=x$, and completed rewards are Bernoulli with mean
$\mu_k$, the product of completion and reward has mean $x\mu_k$, matching their
clipped-linear form after coarsening our feedback.
However, our model allows general known nondecreasing curves and observes
completions separately from rewards.
Accordingly, the mechanisms and techniques are different: their analyses rely
on learning clipped-linear completion parameters, using weighted least squares
in the multi-resource case, whereas our fast-rate result comes from curvature,
interior optimality, and censored observations of task rewards.

Another related line of work studies online resource allocation with diminishing returns under first-order feedback. 
\citet{fontaine2020adaptive} consider repeated allocation over $K$ resources with unknown concave, nondecreasing reward functions and noisy gradient observations, obtaining adaptive rates depending on \L{}ojasiewicz-type regularity. 
Our fast-rate regime is related through the separable concave objective over the simplex, but differs in the information structure: in our model $F_k$ describes a known allocation-dependent completion probability, while the reward associated with a completed task has unknown mean $\mu_k$ and is observed only upon completion. 
Thus, the effective reward function is $\mu_k F_k(x_k)$, but it is learned through censored semi-bandit observations rather than direct noisy gradients. 
In smooth interior regimes, these observations can be importance-weighted to estimate gradients, e.g., via $\frac{F_k'(x_k)}{F_k(x_k)}B_{t,k}(x_k)G_{t,k}$, but this reduction may have unbounded variance near the boundary, and does not apply to non-smooth or threshold curves.
A related setting has been investigated by \citet{thananjeyan2021resource} for best arm identification. They assume the player has access to a divisible unit resource which is used to pull the arms in each round. The time taken to execute a single pull using a fraction $\alpha$ of this resource is $\lambda\big(\frac{1}{\alpha}\big)$, where $\lambda$ is a monotone increasing known function. For example, the player can pull $m$ arms in parallel in time $\lambda(m)$. Their goal is to probabilistically identify the best arm while minimizing the total time.
Another related model is combinatorial bandits with probabilistically triggered arms
\citep{chen2016combinatorial}. Finite discretizations of our action space
$\Delta_K$ can be encoded in this framework: each allocation $x$ acts as a
super-action that triggers base arm $k$ with probability $F_k(x_k)$ and, upon
triggering, reveals a reward sample with mean $\mu_k$. A difference is that in
CMAB-T the triggering probabilities may be unknown and are part of the stochastic
triggering mechanism, whereas in our setting the curves $F_k$ are known. This
makes the optimistic oracle in our problem explicit, but existing CMAB-T regret
guarantees still do not directly imply our results. Their distribution-dependent
bounds are finite-action and gap-based, while their distribution-independent
bounds retain dependence on triggering-probability and smoothness parameters. In
our continuous model, the probabilities $F_k(x_k)$ can become arbitrarily small
near the boundary, and reducing $\Delta_K$ to a finite grid would require
additional regularity assumptions on the curves $F_k$; even under Lipschitzness,
such a discretization suffers from the usual curse of dimensionality. Our
analysis instead works directly over the continuous simplex and exploits the
coupling between allocation, completion probability, and reward observability.
Our model is also reminiscent of the binary-classification-with-paid-experts setting studied in \citet{van2023trading} (see also \citealt{zhang2015bandit} for earlier work).
That line considers payments $x\in[0,1]^K$ and unknown Lipschitz functions $F_1,\dots,F_K$ mapping them to accuracies $(\varepsilon_1,\dots,\varepsilon_K)$ via $\varepsilon_k = F_k(x_k)$, where $\varepsilon_k$ is the probability that expert $k$ correctly predicts the $t$-th bit of an unknown bit sequence.
The resulting objective is not of the form $\langle \mu, F(x)\rangle$, so their guarantees do not directly apply to our problem.

Prior bandit formulations for crowdsourcing \citep{tran2014efficient}, multi-platform autobidding \citep{aggarwal2025multi}, and multi-channel advertising \citep{avadhanula2021stochastic,deng2023multi} typically impose \emph{global} (across-round) budget constraints.
By contrast, we consider \emph{per-round} budget constraints; replacing them with knapsack-style constraints changes the nature of the problem and makes it closer to constrained online optimization with bandit feedback.
Moreover, per-round constraints better capture some applications.
For example, in cost-efficient inference for LLMs \citep{li2025llm,nguyen2024metallm}, one selects which model to query for each input prompt by trading off generation cost against output quality on a per-prompt basis.
In crowdsourcing, stochastic models of acceptance probability and reliability were investigated by \citet{ul2016efficient} without providing performance guarantees.
A setting closer to ours is studied in \citet{kang2020task}, where $F_1,\dots,F_K$ are restricted to be linear.
In that work, these linear functions are known only for some $k\in[K]$, and the learner cannot observe whether tasks with unknown $F_k$ complete successfully.
They prove a $\sqrt{T}$-type regret bound for a notion of reward different from ours.

Regarding censored feedback in the literature,
the reference \cite{agarwal2024semi} addresses a semi-bandit setting, that can be censored, similarly as in our work. Otherwise, the setting there is different from ours, since it is based on a decision tree, and the feedback depends on which path on the tree is sampled/selected, but not on additional random variables that would play the role of our variables $B_{t,k}(x_k)$. 
Also
\cite{van2000revenue,huh2006non} study revenue management with censored demand, again with settings different from ours.
Related \emph{threshold-based} resource-allocation models with censored semi-bandit feedback were studied by \cite{verma2019censored}, but their censoring mechanism differs from ours: in their model allocating enough resource suppresses loss and censors the observation, whereas in our model allocating more budget increases the probability of task completion and hence reward observation.

\end{document}